\documentclass[a4paper,12pt]{article}
\usepackage{amsmath,amsthm,amsfonts,amssymb,bm,mathrsfs}
\usepackage{enumerate}
\usepackage{graphicx}
\usepackage{natbib}
\usepackage{xcolor}
\usepackage{hyperref}
\usepackage{sgame}

\usepackage{setspace}
\allowdisplaybreaks

\newtheorem{defn}{Definition}
\newtheorem{thm}{Theorem}
\newtheorem{lem}{Lemma}
\newtheorem{assum}{Assumption}
\newtheorem{prop}{Proposition}
\newtheorem{rmk}{Remark}

\newtheorem{coro}{Corollary}

\newcommand{\bR}{\mathbb{R}}
\newcommand{\cB}{\mathcal{B}}
\newcommand{\cD}{\mathcal{D}}
\newcommand{\cE}{\mathscr{E}}

\newcommand{\cJ}{\mathcal{J}}
\newcommand{\cK}{\mathcal{K}}

\newcommand{\cM}{\mathcal{M}}

\newcommand{\cS}{\mathcal{S}}

\newcommand{\cX}{\mathcal{X}}

\DeclareMathOperator{\rmd}{d\!}

\hypersetup{colorlinks=true,
            citecolor=blue,
            linkcolor=blue,
            urlcolor=blue,
            bookmarksopen=true,
            pdfstartview={XYZ null null 1.00},
            pdfpagelayout={SinglePage}
}

\begin{document}
\title{Dynamic Games with Almost Perfect Information}
\author{Wei~He\thanks{Department of Economics, The University of Iowa, W249 Pappajohn Business Building, Iowa City, IA 52242. E-mail: he.wei2126@gmail.com.}
\and
Yeneng~Sun\thanks{Department of Economics, National University of Singapore, 1 Arts Link, Singapore 117570. Email: ynsun@nus.edu.sg}
}
\date{This version: March 30, 2015}
\maketitle

%\centerline{Incomplete Preliminary Draft}

%\textbf{JEL classification}:

\abstract{This paper aims to solve two fundamental problems on finite or infinite horizon dynamic games with perfect or almost perfect information. Under some mild conditions, we prove (1) the existence of subgame-perfect equilibria in general dynamic games with almost perfect information, and (2) the existence of {\it pure-strategy} subgame-perfect equilibria in perfect-information dynamic games with uncertainty. Our results go beyond  previous works on continuous dynamic games in the sense that public randomization and the continuity requirement on the state variables are not needed. As an illustrative application, a dynamic stochastic oligopoly market with intertemporally dependent payoffs is considered.

}

%\smallskip
%\textbf{Keywords}: }

\newpage

\tableofcontents

\newpage

\section{Introduction}\label{sec-intro}
Dynamic games with complete information and subgame-perfect equilibria are fundamental game-theoretic concepts with wide applications\footnote{See, for example, Part II of \cite{FT1991}.}. For games with finitely many actions and stages, \cite{Selten1965} showed the existence of subgame-perfect equilibria. The infinite horizon but finite-action case is covered by \cite{FL1983}.

Since the agents in many economic models need to make continuous choices, it is important to consider dynamic games with general action spaces. For deterministic continuous games with perfect information where only one player moves at each stage and all previous moves are observable by the players, the existence of pure-strategy subgame-perfect equilibria is shown in \cite{Harris1985}, \cite{HL1987}, \cite{Borgers1989, Borgers1991} and  \cite{HLRR1990}. However, if the deterministic assumption is dropped by introducing a passive player - Nature, then pure-strategy subgame-perfect equilibrium need not exist as shown by a four-stage game in \citet[p.~538]{HRR1995}. In fact, \cite{LM2003} even demonstrated the nonexistence of mixed-strategy subgame-perfect equilibrium in a five-stage game. Thus, it has remained an open problem to prove the existence of (pure or mixed-strategy) subgame-perfect equilibria in (finite or infinite horizon) perfect-information dynamic games with uncertainty under some general condition.

\cite{HRR1995} considered continuous dynamic games with almost perfect information. In such games, there is a finite number of active players and a passive player, Nature. The players (active and passive) know all the previous moves and
choose their actions simultaneously. All the relevant model parameters are assumed to be continuous in both action and state variables (i.e., Nature's moves). \cite{HRR1995} showed the existence of subgame-perfect correlated equilibria by introducing a public randomization device,\footnote{See also \cite{Mariotti2000} and \cite{RR2002}.} and also demonstrated the possible nonexistence of subgame-perfect equilibrium through a simple example with two players in each of the two stages. This means that the existence of subgame-perfect equilibria under some suitable condition is an open problem even for two-stage dynamic games with almost perfect information.

For dynamic games with perfect or almost perfect information, the earlier works have focused on continuous dynamic games.
The purpose of this paper is to solve the two open problems for (finite or infinite horizon) general dynamic games in which the relevant model parameters are assumed to be continuous in actions, but only measurable in states.\footnote{While continuity in terms of actions is natural and widely adopted, the state continuity requirement as in continuous dynamic games is rather restrictive. The state measurability assumption is the minimal regularity condition one would expect for the model parameters.}
In particular, we show the existence of a subgame-perfect equilibrium in a general dynamic game with almost perfect information under some suitable conditions on the state transitions. Theorem~\ref{thm-general} (and also Proposition~\ref{prop-infinite arm'}) below goes beyond earlier works on continuous dynamic games by dropping public randomization and the continuity requirement on the state variables.  Thus, the class of games considered here includes general stochastic games, where the stage payoffs are usually assumed to be continuous in actions and measurable in states.\footnote{Proposition~\ref{prop-infinite arm'} implies a new existence result on subgame-perfect equilibrium for a general stochastic game; see Remark~\ref{rmk-stochastic game} below. }
Proposition \ref{prop-general} also presents some regularity properties of the equilibrium payoff correspondences, including  compactness and upper hemicontinuity in the action variables.\footnote{Such an upper hemicontinuity property in terms of correspondences of equilibrium payoffs, or outcomes, or correlated strategies has been the key for proving the relevant existence results as in
 \cite{Harris1985}, \cite{HL1987}, \cite{Borgers1989, Borgers1991}, \cite{HLRR1990}, \cite{HRR1995} and \cite{Mariotti2000}.}
As an illustrative application of Theorem~\ref{thm-general}, we consider a dynamic oligopoly market in which firms face stochastic demand/cost and intertemporally dependent payoffs.

We work with the condition that the state transition in each period (except for those periods with one active player) has a component with a suitable density function with respect to some atomless reference measure. This condition is also minimal in the particular sense that the existence result may fail to hold if (1) the passive player, Nature, is not present in the model as shown in \cite{HRR1995}, or (2) with the presence of Nature, the reference measure is not atomless as shown in \cite{LM2003}.

For the special class of continuous dynamic games with almost perfect information, we can weaken the atomless reference measure condition slightly. In particular, we simply assume the state transition in each period (except for those periods with one active player) to be an atomless probability measure for any given history, without the requirement of a common reference measure. Thus, the introduction of a public randomization device as in \cite{HRR1995} is an obvious special case.

For dynamic games with almost perfect information, our main result allows the players to take mixed strategies.
However, for the special class of dynamic games with perfect information\footnote{Dynamic games with perfect information do have  wide applications. For some examples, see \cite{PP1968} for an intergenerational bequest game, and \cite{PY1973} and \cite{Goldman1980}  for intrapersonal games in which consumers have changing preferences. \label{fn-dgpi}}, we obtain the existence of pure-strategy subgame-perfect equilibria in Corollaries \ref{coro-perfect general} and \ref{coro-perfect cont}. When Nature is present, the only known general existence result for dynamic games with perfect information is, to the best of our knowledge, for continuous games with public randomization. On the contrary, our Corollary \ref{coro-perfect general} needs neither continuity in the state variables nor public randomization. Furthermore, our Corollary \ref{coro-perfect cont} provides a new existence result for continuous dynamic games with perfect information, which generalizes the results of \cite{Harris1985}, \cite{HL1987}, \cite{Borgers1989}, and \cite{HLRR1990} to the case when Nature is present.

We follow the standard three-step procedure in obtaining subgame-perfect equilibria of dynamic games, namely, backward induction, forward induction, and approximation of infinite horizon by finite horizon. Because we drop public randomization and the continuity requirement on the state variables, new technical difficulties arise in each step of the proof. In the step of backward induction, we obtain a new existence result for discontinuous games with stochastic endogenous sharing rules, which extends the main result of \cite{SZ1990} by allowing the payoff correspondence to be measurable (instead of upper hemicontinuous) in states. For forward induction, we need to obtain strategies that are jointly measurable in history. When there is a public randomization device, the joint measurability follows from the measurable version of Skorokhod's representation theorem and implicit function theorem respectively as in \cite{HRR1995} and \cite{RR2002}. Here we need to work with the deep  ``measurable" measurable choice theorem of \cite{Mertens2003}. Lastly, in order to obtain results for the infinite horizon case, we need to handle various subtle measurability issues due to the lack of continuity on the state variables in our model.\footnote{We cannot adopt the usual method of approximating a limit continuous dynamic game by a sequence of finite games as used in \cite{HLRR1990}, \cite{Borgers1991} and \cite{HRR1995}.} As noted in Subsection \ref{subsec-proof continuous} below, a considerably simpler proof could be obtained for the case of continuous dynamic games.

The rest of the paper is organized as follows. The model and main result are presented in Section \ref{sec-model}. An illustrative application of Theorem~\ref{thm-general} to a dynamic oligopoly market with stochastic demand/cost and intertemporally dependent payoffs is given in Section \ref{sec-application}. Section \ref{sec-variations} provides several variations of the main result. All the proofs are left in the Appendix.

\section{Model and main result}\label{sec-model}

In this section, we shall present the model for an infinite-horizon dynamic game with almost perfect information.

The set of players is $I_0 = \{0,1,\ldots, n\}$, where the players in $I = \{1,\ldots, n\}$ are active and player~$0$ is Nature. All players  move simultaneously. Time is discrete, and indexed by $t = 0,1,2, \ldots$.

The set of starting points is a closed set $H_0 = X_0 \times S_0$, where $X_0$ is a compact metric space and $S_0$ is a Polish space (that is, a complete separable metric space).\footnote{In each stage~$t \ge 1$, there will be a set of action profiles $X_t$ and a set of states $S_t$. Without loss of generality, we assume that the set of initial points is also a product space for notational consistency.} At stage $t \ge 1$, player~$i$'s action will be chosen from a subset of a Polish space $X_{ti}$ for each $i\in I$, and $X_t = \prod_{i\in I} X_{ti}$. Nature's action is chosen from a Polish space $S_t$.  Let $X^t = \prod_{0\le k \le t}X_k$ and $S^t = \prod_{0\le k \le t}S_k$. The Borel $\sigma$-algebras on $X_t$ and $S_t$ are denoted by $\cB(X_t)$ and $\cB(S_t)$, respectively. Given $t \ge 0$,  a history up to the stage~$t$ is a vector
$$h_{t} = (x_0, s_0, x_1, s_1, \ldots, x_{t}, s_{t} ) \in X^{t} \times S^{t}.$$
The set of all such possible histories is denoted by $H_{t}$. For any $t \ge 0$, $H_{t} \subseteq X^{t} \times S^{t}$.

For any $t \ge 1$ and $i\in I$, let $A_{ti}$  be a measurable, nonempty and compact valued correspondence from $H_{t-1}$ to $X_{ti}$ such that (1) $A_{ti}$ is sectionally continuous on $X^{t-1}$,\footnote{Suppose that $Y_1$, $Y_2$ and $Y_3$ are all Polish spaces, and $Z\subseteq Y_1\times Y_2$. Denote $Z(y_1) = \{y_2\in Y_2 \colon (y_1, y_2) \in Z\}$ for any $y_1 \in Y_1$. A function (resp. correspondence) $f\colon Z\to Y_3$ is said to be sectionally continuous on $Y_2$ if $f(y_1, \cdot)$ is continuous on $Z(y_1)$ for all $y_1$ with $Z(y_1) \neq \emptyset$. Similarly, one can define the sectional upper hemicontinuity for a correspondence.} and (2) $A_{ti}(h_{t-1})$ is the set of available actions for player~$i\in I$ given the history $h_{t-1}$.\footnote{Suppose that $Y$ and $Z$ are both Polish spaces, and $\Psi$ is a correspondence from $Y$ to $Z$. Hereafter, the measurability of $\Psi$, unless specifically indicated, is with respect to the Borel $\sigma$-algebra $\cB(Y)$ on $Y$.}  Let $A_t = \prod_{i\in I}A_{ti}$.  Then $H_t =\mbox{Gr}(A^t) \times S_t$, where $\mbox{Gr}(A^t)$ is the graph of $A^t$.

For any $x =(x_0, x_1, \ldots) \in X^\infty$, let $x^t = (x_0, \ldots, x_t) \in X^t$ be the truncation of $x$ up to the period~$t$. Truncations for $s \in S^\infty$ can be defined similarly. Let $H_{\infty}$ be the subset of $X^\infty \times S^\infty$ such that $(x,s) \in H_\infty$ if $(x^t,s^t) \in H_{t}$ for any $t\ge 0$. Then $H_\infty$ is the set of all possible histories  in the game.\footnote{A finite horizon dynamic game can be regarded as a special case of an infinite horizon dynamic game in the sense that the action correspondence $A_{ti}$ is point-valued for each player $i \in I$ and $t \ge T$ for some stage $T \ge 1$; see, for example, \cite{Borgers1989} and \cite{HRR1995}.}

For any $t\ge 1$, Nature's action is given by a Borel measurable mapping $f_{t0}$ from $H_{t-1}$ to $\cM(S_t)$ such that $f_{t0}$ is sectionally continuous on $X^{t-1}$, where $\cM(S_t)$ is endowed with the topology induced by the weak convergence.\footnote{For a Polish space $A$, $\cM(A)$ denotes the set of all Borel probability measures on $A$, and $\triangle(A)$ is the set of all finite Borel measures on $A$.} For each $t \ge 0$, suppose that $\lambda_t$ is a Borel probability measure on $S_t$ and $\lambda_t$ is atomless for $t\ge 1$. Let $\lambda^t = \otimes_{0 \le k \le t}\lambda_t$ for $t \ge 0$. We shall assume the following condition on the state transitions.

\begin{assum}[Atomless Reference Measure (ARM)]
A dynamic game is said to satisfy the ``atomless reference measure (ARM)'' condition if for each $t \ge 1$,
\begin{enumerate}
  \item the probability $f_{t0}(\cdot|h_{t-1})$ is absolutely continuous with respect to $\lambda_t$ on $S_t$ with the Radon-Nikodym derivative $\varphi_{t0}(h_{t-1}, s_t)$ for all $h_{t-1} \in H_{t-1}$;\footnote{It is common to have a reference measure when one considers a game with uncountable states. For example, if $S_t$ is a subset of $\bR^l$, then the Lebesgue measure is a natural reference measure.}
  \item the mapping $\varphi_{t0}$ is Borel measurable and sectionally continuous in $X^{t-1}$, and integrably bounded in the sense that there is a $\lambda_t$-integrable function $\phi_t \colon S_t \to \bR_+$ such that $\varphi_{t0}(h_{t-1}, s_t) \le \phi_t(s_t)$ for any $h_{t-1} \in H_{t-1}$ and $s_t \in S_t$.
\end{enumerate}
\end{assum}

For each $i \in I$, the payoff function $u_i$ is a Borel measurable mapping from $H_{\infty}$ to $\bR_{++}$ which is sectionally continuous on $X^{\infty}$ and bounded by $\gamma > 0$.\footnote{Since $u_i$ is bounded, we can assume that the value of the payoff function is strictly positive without loss of generality.}

When one considers a dynamic game with infinite horizon, the following ``continuity at infinity'' condition is standard.\footnote{See, for example, \cite{FL1983}.} In particular, all discounted repeated games or stochastic games satisfy this condition.

For any $T\ge 1$, let
\begin{equation} \label{eq-CaI}
w^T = \sup_{\substack{i\in I \\ (x,s)\in H_{\infty} \\ (\overline{x}, \overline{s} ) \in H_\infty \\ x^{T-1} = \overline{x}^{T-1} \\ s^{T-1} = \overline{s}^{T-1} }}
|u_i(x,s) - u_i(\overline{x}, \overline{s})|.
\end{equation}

\begin{assum}[Continuity at Infinity]
A dynamic game is said to be ``continuous at infinity'' if $w^T \to 0$ as $T \to \infty$.
\end{assum}

For player $i\in I$, a strategy $f_i$ is a sequence $\{f_{ti}\}_{t \ge 1}$ such that $f_{ti}$ is a Borel measurable mapping from $H_{t-1}$ to $\cM(X_{ti})$ with $f_{ti}(A_{ti}(h_{t-1})|h_{t-1}) = 1$ for all  $h_{t-1} \in H_{t-1}$. A strategy profile $f = \{f_i\}_{i\in I}$ is a combination of strategies of all active players.

In any subgame, a strategy combination will generate a probability distribution over the set of possible histories. This probability distribution is called the
path induced by the strategy combination in this subgame.

\begin{defn}\label{defn-path}
Suppose that a strategy profile $f = \{f_i\}_{i\in I}$ and a history $h_{t} \in H_t$ are given for some $t \ge 0$. Let $\tau_{t} = \delta_{h_t}$, where $\delta_{h_t}$ is the probability measure concentrated at the one point $h_t$. If $\tau_{t'} \in \cM(H_{t'})$ has already been defined for some $t' \ge t$, then let
$$\tau_{t'+1} = \tau_{t'}\diamond(\otimes_{i \in I_0} f_{(t'+1)i}).\footnote{Denote $\otimes_{i \in I_0} f_{(t'+1)i}$ as a transition probability from $H_{t'}$ to $\cM (X_{t'+1})$. Notice that the strategy profile is usually represented by a vector. For the notational simplicity later on, we assume that $\otimes_{i \in I_0} f_{(t'+1)i} (\cdot | h_{t'})$ represents the strategy profile in stage~$t' + 1$ for a given history $h_{t'} \in H_{t'}$, where $\otimes_{i \in I_0} f_{(t'+1)i} (\cdot | h_{t'})$ is the product of the probability measures $f_{(t'+1)i} (\cdot | h_{t'})$, $i \in I_0$. If $\lambda$ is a finite measure on $X$ and $\nu$ is a transition probability from $X$ to $Y$, then $\lambda\diamond \nu$ is a measure on $X\times Y$ such that $\lambda\diamond \nu(A\times B) = \int_A \nu(B|x) \lambda(\rmd x)$ for any measurable subsets $A \subseteq X$ and $B\subseteq Y$.}$$
Finally, let $\tau \in \cM(H_{\infty})$ be the unique probability measure on $H_\infty$ such that $\mbox{Marg}_{H_{t'}}\tau = \tau_{t'}$ for all $t' \ge t$. Then $\tau$ is called the path induced by $f$ in the subgame $h_t$. For all $i\in I$, $\int_{H_\infty}u_i \rmd \tau$ is the payoff of player~$i$ in this subgame.
\end{defn}

The notion of subgame-perfect equilibrium is defined as follows.
\begin{defn}[SPE]\label{defn-SPE}
A subgame-perfect equilibrium is a strategy profile $f$ such that for all $i\in I$, $t \ge 0$,  and $\lambda^t$-almost all $h_{t} \in H_{t}$,\footnote{A property is said to hold for $\lambda^{t}$-almost all $h_t = (x^t, s^t) \in H_t$ if it is satisfied for $\lambda^t$-almost all $s^t \in S^t$ and all $x^t \in H_t(s^t)$.} player~$i$ cannot improve his payoff in the subgame $h_t$ by a unilateral change in his strategy.\footnote{When the state space is uncountable and has a reference measure, it is natural to consider the optimality for almost all sub-histories in the probabilistic sense; see, for example, \cite{APS1990} and Footnote~4 therein.}
\end{defn}

The following theorem is our main result, which shows the existence of a subgame-perfect equilibrium under the conditions of ARM and continuity at infinity. Its proof is left in the appendix.

\begin{thm}\label{thm-general}
If a dynamic game satisfies the ARM condition and is continuous at infinity, then it possesses a subgame-perfect equilibrium.
\end{thm}

Let $E_t(h_{t-1})$ be the set of subgame-perfect equilibrium payoffs in the subgame $h_{t-1}$. The following result demonstrates the compactness and upper hemicontinuity properties of the correspondence $E_t$.

\begin{prop}\label{prop-general}
If a dynamic game satisfies the ARM condition and is continuous at infinity, then $E_t$ is nonempty and compact valued, and essentially sectionally upper hemicontinuous on $X^{t-1}$.\footnote{Suppose that $Y_1$, $Y_2$ and $Y_3$ are all Polish spaces, and $Z\subseteq Y_1\times Y_2$ and $\eta$ is a Borel probability measure on $Y_1$. Denote $Z(y_1) = \{y_2\in Y_2 \colon (y_1, y_2) \in Z\}$ for any $y_1 \in Y_1$. A function (resp. correspondence) $f\colon Z\to Y_3$ is said to be essentially sectionally continuous on $Y_2$ if $f(y_1, \cdot)$ is continuous on $Z(y_1)$ for $\eta$-almost all $y_1$.  Similarly, one can define the essential sectional upper hemicontinuity for a correspondence.}
\end{prop}

\section{Dynamic oligopoly market with sticky prices}\label{sec-application}

In this section, we consider a dynamic oligopoly market in which firms face stochastic demand/cost and intertemporally dependent payoffs. Such a model is a variant of the well-known dynamic oligopoly models as considered in \cite{GP1984} and \cite{RS1986}, which examined the response of firms for demand fluctuations. The key feature of our example is the existence of sticky price effect, which means that the desirability of the good from the demand side could depend on the accumulated past output, and hence gives intertemporally dependent payoff functions.

We consider a dynamic oligopoly market in which $n$ firms produce a homogeneous good in an infinite-horizon setting. The inverse demand function is denoted by $P_t(Q_1, \ldots, Q_t, s_t)$, where $Q_t$ is the industry output and $s_t$ the observable demand shock in period~$t$. Notice that the price depends on the past outputs. One possible reason could be that the desirability of consumers will be influenced by their previous consumptions, and hence the price does not adjust instantaneously. We assume that $P_t$ is a bounded function which is continuous in $(Q_1, \ldots,Q_t)$ and measurable in $s_t$. In period~$t$, the shock $s_t$ is selected from the set $S_t = [a_t,b_t]$. We denote firm~$i$'s output in period $t$ by $q_{ti}$ so that $Q_t = \sum_{i = 1}^n q_{ti}$. The cost of firm~$i$ in period~$t$ is $c_{ti}(q_{ti}, s_t)$ given the output $q_{ti}$ and the shock $s_t$, where $c_{ti}$ is a bounded function continuous in $q_{ti}$ and measurable in $s_t$. The discount factor of firm~$i$ is $\beta_i\in [0,1)$.

The timing of events is as follows.
\begin{enumerate}
  \item At the beginning of period~$t$, all firms learn the realization of $s_t$, which is determined by the law of motion $\kappa_t(\cdot|s_1,Q_1, \ldots, s_{t-1}, Q_{t-1})$. Suppose that $\kappa_t(\cdot|s_1,Q_1, \ldots, s_{t-1}, Q_{t-1})$ is absolutely continuous with respect to the uniform distribution on $S_t$ with density $\varphi_t(s_1,Q_1, \ldots,s_{t-1}, Q_{t-1}, s_t)$, where $\varphi_t$ is bounded, continuous in $(Q_1, \ldots, Q_{t-1})$ and measurable in $(s_1, \ldots, s_t)$.
  \item Firms then simultaneously choose the level of their output $q_t = (q_{t1}, \ldots, q_{tn})$, where $q_{ti}\in A_{ti}(s_t, Q_{t-1}) \subseteq \bR^l$ for $i =1,2, \ldots, n$. In particular, the correspondence $A_{ti}$ gives the available actions of firm~$i$, which is nonempty and compact valued, measurable in $s_t$, and continuous in $Q_{t-1}$.
  \item The strategic choices of all the firms then become common knowledge and this one-period game is repeated.
\end{enumerate}

In period~$t$, given the shock $s_t$ and the output $\{q_k\}_{1 \le k \le t}$ up to time $t$ with $q_k = (q_{k1}, \ldots, q_{kn})$, the payoff of firm~$i$ is $$u_{ti}(q_1, \ldots, q_t,s_t) = \left[ P_t(\sum_{j = 1}^nq_{1j}, \ldots, \sum_{j = 1}^n q_{tj}, s_t) - c_{ti}(q_{ti}, s_t) \right]q_{ti}.$$
Given a sequence of outputs $\{q_t\}_{t\ge 1}$ and shocks $\{s_t\}_{t\ge 1}$, firm~$i$ receives the payoff
$$u_{1i}(q_1,s_1) + \sum_{t=2}^\infty \beta_i^{t-1} u_{ti}(q_1, \ldots, q_t,s_t).$$

\begin{rmk}
Our dynamic oligopoly model has a non-stationary structure. In particular, the transitions and payoffs are history-dependent. The example captures the scenario that the price of the homogeneous product does not adjust instantaneously to the price indicated by its demand function at the given level of output. For more applications with intertemporally dependent utilities, see, for example, \cite{RH1973}, \cite{FK1987} and \cite{BM1988}. If the model is stationary and the inverse demand function only depends the current output, then the example reduces to be the dynamic oligopoly game with demand fluctuations as considered in \cite{RS1986}.
\end{rmk}

By condition~(1) above, the ARM condition is satisfied. It is also easy to see that the game is continuous at infinity. By Theorem~\ref{thm-general}, we have the following result.
\begin{coro}
The dynamic oligopoly market possesses a subgame-perfect equilibrium.
\end{coro}

\section{Variations of the main result}\label{sec-variations}

In this section, we will consider several variations of our main result.

In Subsection~\ref{subsec-ARM'}, we still consider dynamic games whose parameters are continuous in actions and measurable in states. We  partially relax the ARM condition in two ways. First, we allow the possibility that there is only one active player (but no Nature) at some stages, where the ARM type condition is dropped. Second, we introduce an additional weakly continuous component on the state transitions at any other stages. In addition, we allow the state transition in each period to depend on the current actions as well as on the previous history. Thus, we combine the models for dynamic games with perfect and almost perfect information. We show the existence of a subgame perfect equilibrium such that whenever there is only one active player at some stage, the player can play pure strategy as part of the equilibrium strategies. As a byproduct, we obtain a new existence result for stochastic games. The existence of pure-strategy subgame-perfect equilibria for dynamic games with perfect information (with or without Nature) is provided as an immediate corollary.

In Subsection~\ref{subsec-continuous}, we consider the special case of continuous dynamic games in the sense that all the model parameters are continuous in both action and state variables. We can obtain the corresponding results under a slightly weaker condition. All the previous existence results for continuous dynamic games with perfect and almost perfect information are covered as our special cases.

We will follow the setting and notations in Section~\ref{sec-model} as closely as possible. For simplicity, we only describe the changes we need to make on the model. All the proofs are left in the appendix.

\subsection{Dynamic games with partially perfect information and a generalized ARM condition}\label{subsec-ARM'}

In this subsection, we will generalize the model in Section \ref{sec-model} in three directions. The ARM condition is partially relaxed such that (1) perfect information may be allowed in some stages, and (2) the state transitions have a weakly continuous component in all other stages. In addition, the state transition in any period can depend on the action profile in the current stage as well as on the previous history. The fist change allows us to combine the models of dynamic games with perfect and almost perfect information. The second generalization implies that the state transitions need not be norm continuous in the Banach space of finite measures. The last modification covers the model of stochastic games as a special case.

The changes are described below.
\begin{enumerate}
\item The state space is a product space of two Polish spaces; that is, $S_t = \hat{S}_t \times \tilde{S}_t$ for each $t\ge 1$.

\item For each $i \in I$, the action correspondence $A_{ti}$ from $H_{t-1}$ to $X_{ti}$ is measurable, nonempty and compact valued, and sectionally continuous on $X^{t-1} \times \hat{S}^{t-1}$.  The additional component of Nature is given by a measurable, nonempty and closed valued correspondence $\hat{A}_{t0}$ from $\mbox{Gr}(A_t)$ to $\hat{S}^{t}$, which is sectionally continuous on $X^{t} \times \hat{S}^{t-1}$. Then $H_t = \mbox{Gr}(\hat{A}_{t0}) \times \tilde{S}_t$, and $H_{\infty}$ is the subset of $X^\infty \times S^\infty$ such that $(x,s) \in H_\infty$ if $(x^t,s^t) \in H_{t}$ for any $t\ge 0$.

 \item The choices of Nature depend not only on the history $h_{t-1}$, but also on the action profile $x_t$ in the current stage. The state transition $f_{t0}(h_{t-1},x_t) = \hat{f}_{t0}(h_{t-1},x_t)\diamond \tilde{f}_{t0}(h_{t-1},x_t)$, where $\hat{f}_{t0}$ is a transition probability from $\mbox{Gr}(A_t)$ to $\cM(\hat{S}_t)$ such that $\hat{f}_{t0}(\hat{A}_{t0} (h_{t-1},x_t) | h_{t-1},x_t) = 1$ for all $(h_{t-1},x_t) \in \mbox{Gr}(A_t)$, and $\tilde{f}_{t0}$ is a transition probability from $\mbox{Gr}(\hat{A}_{t0})$ to $\cM(\tilde{S}_t)$.

  \item For each $i \in I$, the payoff function $u_i$ is a Borel measurable mapping from $H_\infty$ to $\bR_{++}$ which is bounded by $\gamma > 0$, and sectionally continuous on $X^{\infty} \times \hat{S}^\infty$.

\end{enumerate}

We allow the possibility for the players to have perfect information in some stages. For $t\ge 1$, let
        $$N_t =
        \begin{cases}
        1,  & \mbox{if } f_{t0}(h_{t-1}, x_t) \equiv \delta_{s_t} \mbox{ for some } s_t \mbox{ and } \\
        \quad & \left| \{i \in I \colon A_{ti} \mbox{ is not point valued} \} \right| =1; \\
        0, & \mbox{otherwise},
        \end{cases} $$
        where $\left| K \right|$ represents the number of points in the set $K$. Thus, if $N_t = 1$ for some stage~$t$, then the player who is active in the period~$t$ is the only active player and has perfect information.

We will drop the ARM condition in those periods with only one active player, and weaken the ARM condition in other periods.
\begin{assum}[ARM$^{\prime}$]
\begin{enumerate}
  \item For any $t\ge 1$ with $N_t = 1$, $S_t$ is a singleton set $\{\acute{s}_t\}$ and $\lambda_t = \delta_{\acute{s}_t}$.
  \item For each $t\ge 1$ with $N_t = 0$, $\hat{f}_{t0}$ is sectionally continuous on $X^t \times \hat{S}^{t-1}$. The probability measure $\tilde{f}_{t0}(\cdot|h_{t-1}, x_t, \hat{s}_t)$ is absolutely continuous with respect to an atomless Borel probability measure $\lambda_t$ on $\tilde{S}_t$ for all $(h_{t-1},x_t, \hat{s}_t) \in \mbox{Gr}(\hat{A}_{t0})$, and $\varphi_{t0}(h_{t-1}, x_t, \hat{s}_t, \tilde{s}_t)$ is the corresponding density.\footnote{In this subsection, a property is said to hold for $\lambda^t$-almost all $h_t \in H_t$ if it is satisfied for $\lambda^t$-almost all $\tilde{s}^t \in \tilde{S}^t$ and all $(x^t, \hat{s}^t) \in H_t(\tilde{s}^t)$.}
  \item The mapping $\varphi_{t0}$ is Borel measurable and sectionally continuous on $X^{t}\times \hat{S}^t$, and integrably bounded in the sense that there is a $\lambda_t$-integrable function $\phi_t \colon \tilde{S}_t \to \bR_+$ such that $\varphi_{t0}(h_{t-1}, x_t, \hat{s}_t, \tilde{s}_t, ) \le \phi_t(\tilde{s}_t)$ for any $(h_{t-1}, x_t,\hat{s}_t)$.
      \end{enumerate}
\end{assum}

The following proposition shows that the existence result is still true in this more general setting.

\begin{prop}\label{prop-infinite arm'}
If an infinite-horizon dynamic game satisfies the ARM$^\prime$ condition and is continuous at infinity, then it possesses a subgame-perfect equilibrium $f$. In particular, for $j \in I$ and $t \ge 1$ such that $N_t = 1$ and player~$j$ is the only active player in this period, $f_{tj}$ can be deterministic. Furthermore, the equilibrium payoff correspondence $E_t$ is nonempty and compact valued, and essentially sectionally upper hemicontinuous on $X^{t-1} \times \hat{S}^{t-1}$.
\end{prop}

\begin{rmk}\label{rmk-stochastic game}
The proposition above also implies a new existence result of subgame-perfect equilibria for stochastic games. Consider a standard stochastic game with uncountable states as in \cite{MP1987}. \cite{MP1987} proved the existence of a subgame-perfect equilibrium by assuming the state transitions to be norm continuous with respect to the actions in the previous stage. On the contrary, our Proposition~\ref{prop-infinite arm'} allows the state transitions to have a weakly continuous component.
\end{rmk}

Dynamic games with perfect information is a special class of dynamic games in which players move sequentially. As noted in Footnote \ref{fn-dgpi}, such games have been extensively studied and found wide applications in economics. As an immediate corollary, an equilibrium existence result for dynamic games with perfect information is given below.
\begin{coro}\label{coro-perfect general}
If a dynamic game with perfect information satisfies the ARM$^\prime$ condition and is continuous at infinity, then it possesses a pure-strategy subgame-perfect equilibrium.
\end{coro}

\subsection{Continuous dynamic games with partially perfect information}\label{subsec-continuous}

In this subsection, we will study an infinite-horizon dynamic game with a continuous structure.
As in the previous subsection, we allow the state transition to depend on the action profile in the current stage as well as on the previous history, and the players may have perfect information in some stages.
\begin{enumerate}
  \item For each $t \ge 1$, the choices of Nature depends not only on the history $h_{t-1}$, but also on the action profile $x_t$ in this stage. For any $t\ge 1$, suppose that $A_{t0}$ is a continuous, nonempty and closed valued correspondence from $\mbox{Gr}(A_t)$ to $S_t$. Then $H_t = \mbox{Gr}(A_{t0})$, and $H_{\infty}$ is the subset of $X^\infty \times S^\infty$ such that $(x,s) \in H_\infty$ if $(x^t,s^t) \in H_{t}$ for any $t\ge 0$.

  \item Nature's action is given by a continuous mapping $f_{t0}$ from $\mbox{Gr}(A_t)$ to $\cM(S_t)$ such that $f_{t0}(A_{t0} (h_{t-1},x_t) | h_{t-1},x_t) = 1$ for all $(h_{t-1},x_t) \in \mbox{Gr}(A_t)$.
  \item For each $t\ge 1$, let
      $$N_t =
      \begin{cases}
      1,  & \mbox{if } f_{t0}(h_{t-1}, x_t) \equiv \delta_{s_t} \mbox{ for some } s_t \mbox{ and } \\
        \quad & \left| \{i \in I \colon A_{ti} \mbox{ is not point valued} \} \right| =1; \\
        0, & \mbox{otherwise}.
      \end{cases}
      $$
\end{enumerate}

\begin{defn}
A dynamic game is said to be continuous if for each $t$ and $i$,
\begin{enumerate}
  \item the action correspondence $A_{ti}$ is continuous on $H_{t-1}$;
  \item the transition probability $f_{t0}$ is continuous on $\mbox{Gr}(A_t)$;
  \item the payoff function $u_i$ is continuous on $H_{\infty}$.
\end{enumerate}
\end{defn}

Note that the ``continuity at infinity'' condition is automatically satisfied in a continuous dynamic game.

Next, we propose the condition of ``atomless transitions'' on the state space, which means that the state transition is an atomless probability measure in any stage. This condition is slightly weaker than the ARM condition.

\begin{assum}[Atomless Transitions]
\begin{enumerate}
  \item For any $t\ge 1$ with $N_t = 1$, $S_t$ is a singleton set $\{\acute{s}_t\}$.
  \item For each $t \ge 1$ with $N_t = 0$, $f_{t0}(h_{t-1})$ is an atomless Borel probability measure for each $h_{t-1} \in H_{t-1}$.
\end{enumerate}
\end{assum}

Since we work with continuous dynamic games, we can adopt a slightly stronger notion of subgame-perfect equilibrium. That is, each  player's strategy is optimal in every subgame given the strategies of all other players.

\begin{defn}[SPE$'$]\label{defn-SPE'}
A subgame-perfect equilibrium is a strategy profile $f$ such that for all $i\in I$, $t \ge 0$,  and all $h_{t} \in H_{t}$, player~$i$ cannot improve his payoff in the subgame $h_t$ by a unilateral change in his strategy.
\end{defn}

The result on the equilibrium existence is presented below.
\begin{prop}\label{prop-continuous}
If a continuous dynamic game has atomless transitions, then it possesses a subgame-perfect equilibrium $f$. In particular, for $j \in I$ and $t \ge 1$ such that $N_t = 1$ and player~$j$ is the only active player in this period, $f_{tj}$ can be deterministic. In addition, $E_t$ is nonempty and compact valued, and upper hemicontinuous on $H_{t-1}$ for any $t\ge 1$.
\end{prop}

\begin{rmk}
Proposition~\ref{prop-continuous} goes beyond the main result of \cite{HRR1995}. They proved the existence of a subgame-perfect correlated equilibrium in a continuous dynamic game with almost perfect information by introducing a public randomization device, which does not influence the payoffs, transitions or action correspondences. It is easy to see that their model automatically satisfies the condition of atomless transitions. The state in our model is completely endogenous in the sense that it affects all the model parameters such as payoffs, transitions, and action correspondences.
\end{rmk}

\begin{rmk}\label{rmk-stochastic cont}
Proposition~\ref{prop-continuous} above provides a new existence result for continuous stochastic games.  As remarked in the previous subsection, the existence of subgame-perfect equilibria has been proved for general stochastic games with a stronger continuity assumption on the state transitions, namely the norm continuity. On the contrary, we only need to require the state transitions to be weakly continuous.
\end{rmk}

\begin{rmk}
The condition of atomless transitions is minimal. In particular, the counterexample provided by \cite{LM2003}, which is a continuous dynamic game with perfect information and Nature, does not have any subgame-perfect equilibrium. In their example, Nature is active in the third period, but the state transitions could have atoms. Thus, our condition of atomless transitions is violated.
\end{rmk}

The next corollary follows from Proposition \ref{prop-continuous}, which presents the existence result for continuous dynamic games with perfect information (and Nature).

\begin{coro}\label{coro-perfect cont}
If a continuous dynamic game with perfect information has atomless transitions, then it possesses a pure-strategy subgame-perfect equilibrium.
\end{coro}

\begin{rmk}
\cite{Harris1985}, \cite{HL1987}, \cite{Borgers1989} and \cite{HLRR1990} proved the existence of subgame-perfect equilibria in continuous dynamic games with perfect information. In particular, Nature is absent in all those papers. \cite{LM2003} provided an example of a five-stage continuous dynamic game with perfect information, in which Nature is present and no subgame-perfect equilibrium exists. The only known general existence result, to the best of our knowledge, for (finite or infinite horizon) continuous dynamic games with perfect information and Nature is the existence of subgame-perfect correlated equilibria via public randomization as in \cite{HRR1995}. Corollary~\ref{coro-perfect cont} covers all those existence results as special cases.
\end{rmk}

\section{Appendix}\label{sec-appendix}

\subsection{Technical preparations}\label{subsec-math}

In this subsection, we present several lemmas as the mathematical preparations for proving Theorem \ref{thm-general}. Since correspondences will be used extensively in the proofs, we collect, for the convenience of the reader, several known results as lemmas.

Let $(S, \cS)$ be a measurable space and $X$ a topological space with its Borel $\sigma$-algebra $\cB (X)$. A correspondence $\Psi$ from $S$ to $X$ is a function from $S$ to the space of all subsets of $X$. The upper inverse $\Psi^u$ of a subset $A \subseteq X$ is
$$\Psi^u(A) = \{s\in S \colon \Psi(s) \subseteq A \}.
$$
The lower inverse $\Psi^l$ of a subset $A \subseteq X$ is
$$\Psi^l(A) = \{s\in S \colon \Psi(s) \cap A \neq \emptyset\}.
$$
The correspondence $\Psi$ is
\begin{enumerate}
  \item weakly measurable, if $\Psi^l(O) \in \cS$ for each open subset $O \subseteq X$;
  \item measurable, if $\Psi^l(K) \in \cS$ for each closed subset $K \subseteq X$.
\end{enumerate}
The graph of $\Psi$ is denoted by $\mbox{Gr}(\Psi) = \{(s,x)\in S \times X \colon s\in S,  x\in \Psi(s)\}$. The correspondence $\Psi$ is said to have a measurable graph if $\mbox{Gr}(\Psi) \in \cS\otimes \cB(X)$.

If $S$ is a topological space, then $\Psi$ is
\begin{enumerate}
  \item upper hemicontinuous, if $\Psi^u(O)$ is open for each open subset $O \subseteq X$;
  \item lower hemicontinuous, if $\Psi^l(O)$ is open for each open subset $O \subseteq X$.
\end{enumerate}

%\begin{defn}
%Let $\{E_n\}$ be a sequence of subsets of a topological space $X$. Then:
%\begin{enumerate}
%  \item A point $x \in X$ belongs to the topological lim sup, denoted by $Ls(E_n)$, if for every neighborhood $O$ of $x$ there are infinitely many $n$ with $O\cap E_n \neq \emptyset$.
%  \item A point $x \in X$ belongs to the topological lim inf, denoted by $Li(E_n)$, if for every neighborhood $O$ of $x$, we have $O\cap E_n \neq \emptyset$ for all but finitely many $n$.
%\end{enumerate}
%\end{defn}

\begin{lem}\label{lem-polish compact}
Suppose that $X$ is a Polish space and $\cK$ is the set of all nonempty compact subsets of $X$ endowed with the Hausdorff metric topology. Then $\cK$ is a Polish space.
\end{lem}

\begin{proof}
By Theorem~3.88 (2) of \cite{AB2006}, $\cK$ is complete. In addition, Corollary~3.90 and Theorem~3.91 of \cite{AB2006} imply that $\cK$ is separable. Thus, $\cK$ is a Polish space.
\end{proof}

\begin{lem}\label{lem-measurable correspondence}
Let $(S, \cS)$ be a measurable space, $X$ a Polish space endowed with the Borel $\sigma$-algebra $\cB(X)$, and $\cK$ the space of nonempty compact subsets of $X$ endowed with its Hausdorff metric topology. Suppose that $\Psi\colon S\to X$ is a nonempty and closed valued correspondence.
\begin{enumerate}
  \item If $\Psi$ is weakly measurable, then it has a measurable graph.
  \item If $\Psi$ is compact valued, then the following statements are equivalent.
      \begin{enumerate}
          \item The correspondence $\Psi$ is weakly measurable.
          \item The correspondence $\Psi$ is measurable.
          \item The function $f\colon S\to \cK$, defined by $f(s) = \Psi(s)$, is Borel measurable.
      \end{enumerate}
%  \item Suppose that $S$ is a Polish space and $\cF$ is the $\sigma$-algebra generated by all analytically measurable subsets of $S$. If the graph of $\psi$ is an analytic set, then $\psi$ is $\cF$-measurable.
  \item Suppose that $S$ is a topological space.
  % and $\cS$ is the corresponding Borel $\sigma$-algebra.
  If $\Psi$ is compact valued, then the function $f\colon S\to \cK$ defined by $f(s) = \Psi(s)$ is continuous if and only if the correspondence $\Psi$ is continuous.
  \item Suppose that $(S, \cS, \lambda)$ is a complete probability space. Then $\Psi$ is $\cS$-measurable if and only if it has a measurable graph.
  \item For a correspondence $\Psi\colon S\to X$ between two Polish spaces, the following statements are equivalent.
      \begin{enumerate}
        \item The correspondence $\Psi$ is lower hemicontinuous at a point $s\in S$.
        \item If $s_n \to s$, then for each $x \in \Psi(s)$, there exist a subsequence $\{s_{n_k}\}$ of $\{s_n\}$ and elements $x_k \in \Psi(s_{n_k})$ for each $k$ such that $x_k \to x$.
      \end{enumerate}
  \item For a correspondence $\Psi\colon S\to X$ between two Polish spaces, the following statements are equivalent.
      \begin{enumerate}
        \item The correspondence $\Psi$ is upper hemicontinuous at a point $s\in S$ and $\Psi(s)$ is compact.
        \item If a sequence $(s_n,x_n)$ in the graph of $\Psi$ satisfies $s_n \to s$, then the sequence $\{x_n\}$ has a limit in $\Psi(s)$.
      \end{enumerate}
  \item Given correspondences $F\colon X\to Y$ and $G\colon Y\to Z$, the composition $F$ and $G$ is defined by
      $$G(F(x)) = \cup_{y\in F(x)} G(y).
      $$
        The composition of upper hemicontinuous correspondences is upper hemicontinuous. The composition of lower hemicontinuous correspondences is lower hemicontinuous.
\end{enumerate}
\end{lem}

\begin{proof}
Properties~(1), (2), (3), (5), (6) and (7) are Theorems~18.6, 18.10, 17.15, 17.20, 17.21 and 17.23 of \cite{AB2006}, respectively. Property~(4) is Theorem~4.1 (c) of \cite{Hess2002}.
\end{proof}

\begin{lem}\label{lem-measurable selection}
\begin{enumerate}
  \item A correspondence $\Psi$ from a measurable space $(S, \cS)$ into a topological space $X$ is weakly measurable if and only if its closure correspondence $\overline{\Psi}$ is weakly measurable, where for each $s\in S$, $\overline{\Psi}(s) = \overline{\Psi(s)}$ and $\overline{\Psi(s)}$ is the closure of the set $\Psi(s)$ in $X$.
  \item For a sequence $\{\Psi_m\}$ of correspondences from a measurable space $(S, \cS)$ into a Polish space, the union correspondence $\Psi(s) = \cup_{m \ge 1}\Psi_m(s)$ is weakly measurable if each $\Psi_m$ is weakly measurable. If each $\Psi_m$ is weakly measurable and compact valued, then the intersection correspondence $\Phi(s) = \cap_{m \ge 1}\Psi_m(s)$ is weakly measurable.
  \item A weakly measurable, nonempty and closed valued correspondence from a measurable space into a Polish space admits a measurable selection.
  \item A correspondence with closed graph between compact metric spaces  is measurable.
  \item A nonempty and compact valued correspondence $\Psi$ from a measurable space $(S, \cS)$ into a Polish space is weakly measurable if and only if there exists a sequence $\{\psi_1, \psi_2. \ldots\}$ of measurable selections of $\Psi$ such that $\Psi(s) = \overline{\{\psi_1(s), \psi_2(s), \ldots\}}$ for each $s\in S$.
  \item The image of a compact set under a compact valued upper hemicontinuous correspondence is compact.\footnote{Given a correspondence $F\colon X\to Y$ and a subset $A$ of $X$,  the image of $A$ under $F$ is defined to be the set
      $\cup_{x\in A} F(x)$.} If the domain is compact, then the graph of a compact valued upper hemicontinuous correspondence is compact.
  \item The intersection of a correspondence with closed graph and an upper hemicontinuous compact valued correspondence is upper hemicontinuous.
  \item If the correspondence $\Psi \colon S\to \bR^l$ is compact valued and upper hemicontinuous, then the convex hull of $\Psi$ is also compact valued and upper hemicontinuous.
\end{enumerate}
\end{lem}

\begin{proof}
Properties~(1)-(7) are Lemmas~18.3 and 18.4, Theorems~18.13 and 18.20, Corollary~18.15, Lemma~17.8 and Theorem~17.25 in \cite{AB2006}, respectively. Property~8 is Proposition~6 in \citet[p.26]{Hd1974}.
\end{proof}

\begin{lem}\label{lem-lusin}
\begin{enumerate}
  \item Lusin's Theorem: Suppose that $S$ is a Borel subset of a Polish space, $\lambda$ is a Borel probability measure on $S$ and $\cS$ is the completion of $\cB(S)$ under $\lambda$. Let $X$ be a Polish space. If $f$ is an $\cS$-measurable mapping from $S$ to $X$, then for any $\epsilon >0$, there exists a compact subset $S_1 \subseteq S$ with $\lambda(S\setminus S_1) < \epsilon$ such that the restriction of $f$ to $S_1$ is continuous.
  \item Let $(S,\cS)$ be a measurable space, $X$ a Polish space, and $Y$ a separable Banach space. Let $\Psi \colon S\times X \to Y$ be an  $\cS\otimes \cB(X)$-measurable, nonempty, convex and compact valued correspondence which is sectionally continuous on $X$. Then there exists an $\cS\otimes \cB(X)$-measurable selection $\psi$ of $\Psi$ that is sectionally continuous on $X$.
  \item Let $(S,\cS, \lambda)$ be a finite measure space, $X$ a Polish space, and $Y$ a locally convex linear topological space. Let $F \colon S \to X$ be a closed-valued correspondence such that $\mbox{Gr}(F) \in \cS\otimes \cB(X)$, and $f\colon \mbox{Gr}(F) \to Y$ a measurable function which is sectionally continuous in $X$. Then there exists a measurable function $f' \colon S\times X \to Y$ such that (1) $f'$ is sectionally continuous in $X$, (2) for $\lambda$-almost all $s\in S$, $f'(s,x) = f(s,x)$ for all $x\in F(s)$ and $f' (s, X) \subseteq \mbox{co} f(s, F(s))$.\footnote{For any set $A$ in a linear topological space, \mbox{co}A denotes the convex hull of $A$.}
\end{enumerate}
\end{lem}

\begin{proof}
Lusin's theorem is Theorem~7.1.13 in \cite{Bogachev2007}. Properties~(2) and (3) are Theorem~1 and Theorem~2.7 in \cite{FMM2006} and \cite{BS1989}, respectively.
\end{proof}

The following lemma presents the convexity, compactness and continuity properties of integrals of correspondences

\begin{lem}\label{lem-lyapunov}
Let $(S, \cS, \lambda)$ be an atomless probability space, $X$ a Polish space, and $F$ a correspondence from $S$ to $\bR^l$. Denote
$$\int_S F(s) \lambda(\rmd s) = \left\{\int_S f(s) \lambda(\rmd s) \colon f \mbox{ is an integrable selection of } F \mbox{ on } S\right\}.$$
\begin{enumerate}
  \item If $F$ is measurable, nonempty and closed valued, and $\lambda$-integrably bounded by some integrable function $\psi \colon S\to \bR_+$ in the sense that for $\lambda$-almost all $s\in S$, $\|y\| \le \psi(s)$ for any $y \in F(s)$, then $\int_S F(s) \lambda(\rmd s)$ is nonempty, convex and compact, and
      $$\int_S F(s) \lambda(\rmd s) = \int_S \mbox{co}F(s) \lambda(\rmd s).$$
  \item If $G$ is a measurable, nonempty and closed valued correspondence from $S\times X\to \bR^l$ such that (1) $G(s, \cdot)$ is upper (resp. lower) hemicontinuous on $X$ for all $s\in S$, and (2) $G$ is $\lambda$-integrably bounded by some integrable function $\psi \colon S\to \bR_+$ in the sense that for $\lambda$-almost all $s\in S$, $\|y\| \le \psi(s)$ for any $x \in X$ and $y \in G(s, x)$,
        then $\int_S G(s,x) \lambda(\rmd s)$ is upper (resp. lower) hemicontinuous on $X$.
\end{enumerate}
\end{lem}

\begin{proof}
See Theorems~2, 3 and 4, Propositions~7 and 8, and Problem~6 in Section~D.II.4 of \cite{Hd1974}.
\end{proof}

The following result proves a measurable version of Lyapunov's theorem, which is taken from \cite{Mertens2003}.
Let $(S, \cS)$ and $(X, \cX)$ be measurable spaces. A transition probability from $S$ to $X$ is a mapping $f$ from $S$ to the space $\cM(X)$ of probability measures on $(X, \cX)$ such that $f(B| \cdot): s\to f(B|s)$ is $\cS$-measurable for each $B\in\cX$.

\begin{lem}\label{lem-mertens}
Let $f(\cdot|s)$ be a transition probability from a measurable space $(S, \cS)$ to another measurable space $(X, \cX)$ ($\cX$ is separable).\footnote{A $\sigma$-algebra is said to be separable if it is generated by a countable collection of sets.} Let $Q$ be a measurable, nonempty and compact valued correspondence from $S\times X$ to $\bR^l$, which is $f$-integrable in the sense that for any measurable selection $q$ of $Q$, $q(\cdot, s)$ is $f(\cdot|s)$-absolutely integrable for any $s \in S$. Let $\int Q\rmd f$ be the correspondence from $S$ to subsets of $\bR^{l}$ defined by
$$M(s) = \left(\int Q\rmd f\right)(s) =  \left\{\int_{X} q(s,x) f(\rmd x|s) \colon q \mbox{ is a measurable selection of } Q \right\}.
$$
Denote the graph of $M$ by $J$. Let $\cJ$ be the restriction of the product $\sigma$-algebra $\cS \otimes \cB (\bR^l)$ to $J$.

Then
\begin{enumerate}
  \item $M$ is a measurable, nonempty and compact valued correspondence;
  \item there exists a measurable, $\bR^l$-valued function $g$ on $(X\times J, \cX \otimes \cJ)$ such that $g(x,e,s) \in Q(x,s)$ and $e = \int_{X} g(x,e,s)f(\rmd x|s)$.
\end{enumerate}
\end{lem}

Suppose that $(S_1, \cS_1)$ is a measurable space, $S_2$ is a Polish space endowed with the Borel $\sigma$-algebra, and $S = S_1 \times S_2$ which is endowed with the product $\sigma$-algebra $\cS$. Let $D$ be an $\cS$-measurable subset of $S$ such that $D(s_1)$ is compact for any $s_1 \in S_1$. The $\sigma$-algebra $\cD$ is the restriction of $\cS$ on $D$. Let $X$ be a Polish space, and $A$ a $\cD$-measurable, nonempty and closed valued correspondence from $D$ to $X$ which is sectionally continuous on $S_2$. The following lemma considers the property of upper hemicontinuity for the correspondence $M$ as defined in Lemma \ref{lem-mertens}.

\begin{lem}\label{lem-upper}
Let $f(\cdot|s)$ be a transition probability from $(D, \cD)$ to $\cM(X)$ such that $f(A(s)|s) = 1$ for any $s\in D$, which is sectionally continuous on $S_2$. Let $G$ be a bounded, measurable, nonempty, convex and compact valued correspondence from $\mbox{Gr}(A)$ to $\bR^l$,
which is
%$f$-integrable and
sectionally upper hemicontinuous on $S_2\times X$. Let $\int G\rmd f$ be the correspondence from $D$ to subsets of $\bR^{l}$ defined by
$$M(s) = \left(\int G\rmd f\right)(s) =  \left\{\int_{X} g(s,x) f(\rmd x|s) \colon g \mbox{ is a measurable selection of } G\right\}.
$$
Then $M$ is $\cS$-measurable, nonempty and compact valued, and sectionally upper hemicontinuous on $S_2$.
\end{lem}

\begin{proof}
Define a correspondence $\tilde{G} \colon S\times X\to \bR^l$ as
$$\tilde{G} =
\begin{cases}
G(s,x), & \mbox{if } (s,x) \in \mbox{Gr}(A); \\
\{0\},  & \mbox{otherwise}.
\end{cases}
$$
Then $M(s) = \left(\int \tilde{G}\rmd f\right)(s) = \left(\int G\rmd f\right)(s)$.
The measurability, nonemptiness and compactness follows from Lemma~\ref{lem-mertens}. Given $s_1 \in S_1$ such that (1) $D(s_1) \neq \emptyset$, (2) $f(s_1, \cdot)$ and $G(s_1, \cdot, \cdot)$ is upper hemicontinuous. The upper hemicontinuity of $M(s_1, \cdot)$ follows from Lemma~2 in \cite{SZ1990} and Lemma~4 in \cite{RR2002}.
\end{proof}

Now we state some properties for transition correspondences.

\begin{lem}\label{lem-cont tran corre}
Suppose that $Y$ and $Z$ are Polish spaces. Let $G$ be a measurable, nonempty, convex and compact valued correspondence from $Y$ to $\cM(Z)$. Define a correspondence $G'$ from $\cM(Y)$ to $\cM(Z)$ as
$$G'(\nu) = \left\{\int_Y g(y) \nu(\rmd y)\colon g \mbox{ is a Borel measurable selection of } G \right\}.
$$
\begin{enumerate}
  \item The correspondence $G'$ is measurable, nonempty, convex and compact valued.
  \item The correspondence $G$ is upper hemicontinuous if and only if $G'$ is upper hemicontinuous. In addition, if $G$ is continuous, then $G'$ is continuous.
\end{enumerate}
\end{lem}

\begin{proof}
(1) is Lemma~19.29 of \cite{AB2006}. By Theorem~19.30 therein, $G$ is upper hemicontinuous if and only if $G'$ is upper hemicontinuous. We need to show that $G'$ is lower hemicontinuous if $G$ is lower hemicontinuous.

Let $Z$ be endowed with a totally bounded metric, and $U(Z)$ the space of bounded, real-valued and uniformly continuous functions on $Z$ endowed with the supremum norm. Pick a countable set $\{f_m\}_{m \ge 1} \subseteq U(Z)$ such that $\{f_m\}$ is dense in the unit ball of $U(Z)$. The weak$^*$ topology of $\cM(Z)$ is metrizable by the metric $d_z$, where
$$d_z(\mu_1,\mu_2) = \sum^\infty_{m=1} {1\over{2^m}} \left| \int_{Z}f_m(z)\mu_1(\rmd z) - \int_{Z}f_m(z)\mu_2(\rmd z) \right|$$
for each pair of $\mu_1, \mu_2 \in \cM(Z)$.
%Similar, one can define a metric $d_y$ on $\cM(Y)$ such that the weak$^*$ topology if metrizable under $d_y$.

Suppose that $\{\nu_j\}_{j \ge 0}$ is a sequence in $\cM(Y)$ such that $\nu_j \to \nu_0$ as $j \to \infty$. Pick an arbitrary point $\mu_0 \in G'(\nu_0)$. By the definition of $G'$, there exists a Borel measurable selection $g$ of $G$ such that $\mu_0 = \int_Y g(y) \nu_0(\rmd y)$.

For each $k \ge 1$, by Lemma~\ref{lem-lusin} (Lusin's theorem), there exists a compact subset $D_k\subseteq Y$ such that $g$ is continuous on $D_k$ and $\nu_0(Y\setminus D_k) < \frac{1}{3k}$. Define a correspondence $G_k \colon Y\to \cM(X)$ as follows:
$$G_k(y) =
\begin{cases}
\{g(y)\}, & y\in D_k; \\
G(y),     & y \in Y\setminus D_k.
\end{cases}
$$
Then $G_k$ is nonempty, convex and compact valued, and lower hemicontinuous. By Theorem~3.22 in \cite{AB2006}, $Y$ is paracompact. Then by Michael's selection theorem (see Theorem~17.66 in \cite{AB2006}), it has a continuous selection $g_k$.

For each $k$, since $\nu_j \to \nu_0$ and $g_k$ is continuous, $\int_Y g_k(y) \nu_j(\rmd y) \to \int_Y g_k(y) \nu_0(\rmd y)$ in the sense that for any $m \ge 1$,
$$\int_Y \int_Z f_m(z) g_k(\rmd z|y) \nu_j(\rmd y) \to \int_Y \int_Z f_m(z) g_k(\rmd z| y) \nu_0(\rmd y).$$
Thus, there exists a point $\nu_{j_k}$ such that $\{j_k\}$ is an increasing sequence and
$$d_z\left(\int_Y g_k(y) \nu_{j_k}(\rmd y), \int_Y g_k(y) \nu_0(\rmd y)\right) < \frac{1}{3k}.$$
In addition, since $g_k$ coincides with $g$ on $D_k$ and $\nu_0(Y\setminus D_k) < \frac{1}{3k}$,
$$d_z\left( \int_{Y} g_k(y) \nu_0(\rmd y), \int_{Y} g(y) \nu_0(\rmd y) \right)  < \frac{2}{3k}.
$$
Thus,
$$d_z\left(\int_Y g_k(y) \nu_{j_k}(\rmd y), \int_{Y} g(y) \nu_0(\rmd y) \right)  < \frac{1}{k}.
$$
Let $\mu_{j_k} = \int_Y g_k(y) \nu_{j_k}(\rmd y)$ for each $k$. Then $\mu_{j_k} \in G'(\nu_{j_k})$ and $\mu_{j_k} \to \mu_0$ as $k \to \infty$. By Lemma~\ref{lem-measurable correspondence}, $G'$ is lower hemicontinuous.
\end{proof}

\begin{lem}\label{lem-measurable composition}
Let $X$, $Y$ and $Z$ be Polish spaces, and $G$ a measurable, nonempty and compact valued correspondence from $X$ to $\cM(Y)$. Suppose that $F$ is a measurable, nonempty, convex and compact valued correspondence from $X\times Y$ to $\cM(Z)$. Define a correspondence $\Pi$ from $X$ to $\cM(Y\times Z)$ as follows:
\begin{align*}
\Pi(x) =  \{
& g(x)\diamond f(x) \colon  g \mbox{ is a Borel measurable selection of } G, \\
& f \mbox{ is a Borel measurable selection of } F \}.
\end{align*}
\begin{enumerate}
  \item If $F$ is sectionally continuous on $Y$, then $\Pi$ is a measurable, nonempty and compact valued correspondence.
  \item If there exists a function $g$ from $X$ to $\cM(Y)$ such that $G(x) = \{g(x)\}$ for any $x \in X$, then $\Pi$ is a measurable, nonempty and compact valued correspondence.
  \item If both $G$ and $F$ are continuous correspondences, then $\Pi$ is a nonempty and compact valued, and continuous correspondence.\footnote{In Lemma~29 of \cite{HRR1995}, they showed that $\Pi$ is upper hemicontinuous if $G$ and $F$ are both upper hemicontinuous.}
  \item If $G(x) \equiv \{\lambda\}$ for some fixed Borel probability measure $\lambda \in \cM(Y)$ and $F$ is sectionally continuous on $X$, then $\Pi$ is a continuous, nonempty and compact valued correspondence.
\end{enumerate}
\end{lem}

\begin{proof}

(1) Define three correspondences $\tilde{F} \colon X\times Y \to \cM(Y \times Z)$, $\hat{F} \colon \cM(X\times Y) \to \cM(Y \times Z)$ and $\check{F} \colon X\times \cM(Y) \to \cM(Y \times Z)$ as follows:
$$\tilde{F}(x,y) = \{\delta_{y} \otimes \mu\colon \mu \in F(x,y)\},
$$
$$\hat{F}(\tau) = \left\{\int_{X\times Y}f(x,y) \tau(\rmd (x,y)) \colon f \mbox{ is a Borel measurable selection of } \tilde{F} \right\},
$$
$$\check{F}(x, \mu) = \hat{F}(\delta_x\otimes \mu).
$$
Since $F$ is measurable, nonempty, convex and compact valued, $\tilde{F}$ is measurable, nonempty, convex and compact valued.
By Lemma~\ref{lem-cont tran corre}, the correspondence $\hat{F}$ is measurable, nonempty, convex and compact valued, and $\check{F}(x, \cdot)$ is continuous on $\cM (Y)$ for any $x\in X$.

Since $G$ is measurable and compact valued, there exists a sequence of Borel measurable selections $\{g_k\}_{k \ge 1}$ of $G$ such that $G(x) = \overline{\{g_1(x), g_2(x), \ldots\}}$ for any $x\in X$ by Lemma~\ref{lem-measurable selection}~(5). For each $k \ge 1$, define a correspondence $\Pi^k$ from $X$ to $\cM(Y\times Z)$ by letting $\Pi^k(x) =  \check{F}(x, g_k(x)) = \hat{F} (x \otimes g_k(x))$. Then $\Pi^k$ is measurable, nonempty, convex and compact valued.

Fix any $x\in X$. It is clear that $\Pi(x) = \check{F}(x, G(x))$ is a nonempty valued. Since $G(x)$ is compact, and $\check{F}(x, \cdot)$ is compact valued and continuous, $\Pi(x)$ is compact by Lemma~\ref{lem-measurable selection}. Thus, $\overline{\{\Pi^1(x), \Pi^2(x), \ldots\}} \subseteq \Pi(x)$.

Fix any $x \in X$ and $\tau \in \Pi(x)$. There exists a point $\nu \in G(x)$ such that $\tau \in \check{F}(x, \nu)$. Since $\{g_k(x)\}_{k \ge 1}$ is dense in $G(x)$, it has a subsequence $\{g_{k_m}(x)\}$ such that $g_{k_m}(x) \to \nu$. As $\check{F}(x, \cdot)$ is continuous, $\check{F}(x, g_{k_m}(x)) \to \check{F}(x, \nu)$. That is,
$$\tau \in  \overline{\{\check{F}(x, g_1(x)), \check{F}(x, g_2(x)), \ldots\}} = \overline{\{\Pi^1(x), \Pi^2(x), \ldots\}}.$$
Therefore, $\overline{\{\Pi^1(x), \Pi^2(x), \ldots\}} = \Pi(x)$ for any $x\in X$. Lemma~\ref{lem-measurable selection} (5) implies that $\Pi$ is measurable.

\

(2) As in (1), the correspondence $\hat{F}$ is measurable, nonempty, convex and compact valued. If $G$ is a measurable function, then $\Pi(x) = \hat{F} (x \otimes G(x))$, which is measurable, nonempty and compact valued.

\

(3) We continue to work with the two correspondences $\tilde{F} \colon X\times Y \to \cM(Y \times Z)$ and $\hat{F} \colon \cM(X\times Y) \to \cM(Y \times Z)$ as in Part (1). By the condition on $F$, it is obvious that the correspondence $\tilde{F}$ is continuous, nonempty, convex and compact valued. Lemma~\ref{lem-cont tran corre} implies the properties for the correspondence $\hat{F}$. Define a correspondence $\hat{G}\colon X\to \cM(X\times Y)$ as $\hat{G}(x) = \delta_{x}\otimes G(x)$. Since $\hat{G}$ and $\hat{F}$ are both nonempty valued, $\Pi(x) = \hat{F}(\hat{G}(x))$ is nonempty. As $\hat{G}$ is compact valued and $\hat{F}$ is continuous, $\Pi$ is compact valued by Lemma~\ref{lem-measurable selection}. As $\hat{G}$ and $\hat{F}$ are both continuous, $\Pi$ is continuous by Lemma~\ref{lem-measurable correspondence}~(7).

\

(4) The lower hemicontinuity is from Proposition~4.8 in \cite{Sun1997}. The nonemptiness and compactness follow from Corollary~18.37 of \cite{AB2006} while the upper hemicontinuity follows from the compactness property easily.
\end{proof}

The following result presents a variant of Lemma~\ref{lem-mertens} in terms of transition correspondences.

\begin{lem}\label{lem-measurable payoff}
Let $X$ and $Y$ be Polish spaces, and $Z$ a compact subset of $\bR_+^l$.  Let $G$ be a measurable, nonempty and compact valued correspondence from $X$ to $\cM(Y)$. Suppose that $F$ is a measurable, nonempty, convex and compact valued correspondence from $X\times Y$ to $Z$. Define a correspondence $\Pi$ from $X$ to $Z$ as follows:
\begin{align*}
\Pi(x) =  \{
& \int_{Y} f(x,y) g(\rmd y |x) \colon  g \mbox{ is a Borel measurable selection of } G, \\
& f \mbox{ is a Borel measurable selection of } F \}.
\end{align*}
If $F$ is sectionally continuous on $Y$, then
\begin{enumerate}
  \item the correspondence $\tilde{F}\colon X\times \cM(Y) \to Z$ as $\tilde{F}(x,\nu) = \int_Y F(x,y) \nu(\rmd y)$ is sectionally continuous on $\cM(Y)$; and
  \item $\Pi$ is a measurable, nonempty and compact valued correspondence.
  \item If $F$ and $G$ are both continuous, then $\Pi$ is continuous.
\end{enumerate}
\end{lem}

\begin{proof}
The upper hemicontinuity of $\tilde{F}(x, \cdot)$ follows from Lemma~\ref{lem-upper}, and the proof for the lower hemicontinuity of $\tilde{F}(x, \cdot)$ is similar to that of Lemma~\ref{lem-cont tran corre}. The proof of (2) and (3) follows a similar argument as in the proof of Lemma~\ref{lem-measurable composition}.
\end{proof}

\begin{lem}\label{lem-lusin prep}
Let $S$, $X$ and $Y$ be Polish spaces endowed with the Borel $\sigma$-algebras, and $\lambda$ a Borel probability measure on $S$. Denote $\cS$ as the completion of the Borel $\sigma$-algebra $\cB(S)$ of $S$ under the probability measure $\lambda$. Suppose that $D$ is a $\cB(S) \otimes \cB(Y)$-measurable subset of $S\times Y$, where $D(s)$ is nonempty and compact for all $s\in S$. Let $A$ be a nonempty and compact valued correspondence from $D$ to $X$, which is sectionally continuous on $Y$ and has a $\cB(S \times Y\times X)$-measurable graph. Then
\begin{enumerate}
  \item[(i)] $\tilde{A}(s) = \mbox{Gr}(A(s, \cdot))$ is an $\cS$-measurable mapping from $S$ to the set of nonempty and compact subsets $\cK_{Y\times X}$ of $Y\times X$;
  \item[(ii)] there exist countably many disjoint compact subsets $\{S_m\}_{m\ge 1}$ of $S$ such that (1) $\lambda(\cup_{m \ge 1} S_m) = 1$, and (2) for each $m \ge 1$, $D_m = D\cap (S_m \times Y)$ is compact, and $A$ is nonempty and compact valued, and continuous on each $D_m$.
\end{enumerate}
\end{lem}

\begin{proof}
(i) $A(s,\cdot)$ is continuous and $D(s)$ is compact, $\mbox{Gr}(A(s,\cdot)) \subseteq Y\times X$ is compact by Lemma~\ref{lem-measurable selection}.  Thus, $\tilde{A}$ is nonempty and compact valued. Since $A$ has a measurable graph, $\tilde{A}$ is an $\cS$-measurable mapping from $S$ to the set of nonempty and compact subsets $\cK_{Y\times X}$ of $Y\times X$ by Lemma~\ref{lem-measurable correspondence}~(4).

(ii) Define a correspondence $\tilde{D}$ from $S$ to $Y$ such that $\tilde{D}(s) = \{y \in Y \colon (s,y) \in D\}$. Then $\tilde{D}$ is nonempty and compact valued. As in (i), $\tilde{D}$ is $\cS$-measurable. By Lemma~\ref{lem-lusin} (Lusin's~Theorem), there exists a compact subset $S_1 \subseteq S$ such that $\lambda(S\setminus S_1) < \frac{1}{2}$, $\tilde{D}$ and $\tilde{A}$ are continuous functions on $S_1$. By Lemma~\ref{lem-measurable correspondence}~(3), $\tilde{D}$ and $\tilde{A}$ are continuous correspondences on $S_1$. Let $D_1  = \{(s,y)\in D\colon s\in S_1, y \in \tilde{D}(s)\}$. Since $S_1$ is compact and $\tilde{D}$ is continuous, $D_1$ is compact (see Lemma~\ref{lem-measurable selection}~(6)).

Following the same procedure, for any $m\ge 1$, there exists a compact subset $S_m \subseteq S$ such that (1) $S_m\cap (\cup_{1\le k \le m-1} S_k) = \emptyset$ and $D_m = D\cap (S_m \times Y)$ is compact, (2) $\lambda(S_m) > 0$ and $\lambda\left(S\setminus (\cup_{1\le k \le m}S_m)\right) < \frac{1}{2m}$, and (3) $A$ is nonempty and compact valued, and continuous on $D_m$. This completes the proof.
\end{proof}

\begin{lem}\label{lem-countable partition}
Let $S$ and $X$ be Polish spaces, and $\lambda$ a Borel probability measure on $S$. Suppose that $\{S_k\}_{k \ge 1}$ is a sequence of disjoint compact subsets of $S$ such that $\lambda(\cup_{k \ge 1} S_k) = 1$. For each $k$, define a probability measure on $S_k$ as $\lambda_k(D) = \frac{\lambda(D)}{\lambda(S_k)}$ for any measurable subset $D \subseteq S_k$. Let $\{\nu_m\}_{m \ge 0}$ be a sequence of transition probabilities from $S$ to $\cM(X)$, and $\tau_m = \lambda\diamond \nu_m$ for any $m \ge 0$. Then $\tau_m$ weakly converges to $\tau_0$ if and only if $\lambda_k \diamond \nu_m$ weakly converges to $\lambda_k \diamond \nu_0$ for each $k \ge 1$.
\end{lem}

\begin{proof}
First, we assume that $\tau_m$ weakly converges to $\tau_0$. For any closed subset $E \subseteq S_k \times X$, we have $\limsup_{m \to \infty} \tau_m(E) \le \tau_0(E)$. That is, $\limsup_{m \to \infty} \lambda \diamond \nu_m(E) \le \lambda \diamond \nu_0(E)$. For any $k$, $\frac{1}{\lambda(S_k)}\limsup_{m \to \infty} \lambda \diamond \nu_m(E) \le \frac{1}{\lambda(S_k)} \lambda \diamond \nu_0(E)$, which implies that $\limsup_{m \to \infty} \lambda_k \diamond \nu_m(E) \le \lambda_k \diamond \nu_0(E)$. Thus, $\lambda_k \diamond \nu_m$ weakly converges to $\lambda_k \diamond \nu_0$ for each $k \ge 1$.

Second, we consider the case that $\lambda_k \diamond \nu_m$ weakly converges to $\lambda_k \diamond \nu_0$ for each $k \ge 1$. For any closed subset $E \subseteq S\times X$, let $E_k = E\cap (S_k \times X)$ for each $k \ge 1$. Then $\{E_k\}$ are disjoint closed subsets and $\limsup_{m \to \infty} \lambda_k \diamond \nu_m(E_k) \le  \lambda_k \diamond \nu_0(E_k)$. Since $\lambda_k \diamond \nu_m (E') = \frac{1}{\lambda(S_k)} \lambda\diamond \nu_m(E')$ for any $k$, $m$ and measurable subset $E' \subseteq S_k \times X$, we have that $\limsup_{m \to \infty} \lambda \diamond \nu_m(E_k) \le  \lambda \diamond \nu_0(E_k)$. Thus,
$$\sum_{k \ge 1}\limsup_{m \to \infty} \lambda \diamond \nu_m(E_k) \le  \sum_{k \ge 1}\lambda \diamond \nu_0(E_k) = \lambda \diamond \nu_0(E).
$$
Since the limit superior is subadditive, we have
$$\sum_{k \ge 1}\limsup_{m \to \infty} \lambda \diamond \nu_m(E_k) \ge  \limsup_{m \to \infty} \sum_{k \ge 1}\lambda \diamond \nu_m(E_k) = \limsup_{m \to \infty} \lambda \diamond \nu_m(E).
$$
Therefore, $\limsup_{m \to \infty} \lambda \diamond \nu_m(E) \le \lambda \diamond \nu_0(E)$, which implies that $\tau_m$ weakly converges to $\tau_0$.
\end{proof}

\begin{lem}\label{lem-continuous composition}
Suppose that $X, Y$ and $S$ are Polish spaces and $Z$ is a compact metric space. Let $\lambda$ be a Borel probability measure on $S$, and $A$ a nonempty and compact valued correspondence from $Z\times S$ to $X$ which is sectionally upper hemicontinuous on $Z$ and has a $\cB(Z \times S \times X)$-measurable graph. Let $G$ be a nonempty and compact valued, continuous correspondence from $Z$ to $\cM(X\times S)$. We assume that for any $z\in Z$ and $\tau \in G(z)$, the marginal of $\tau$ on $S$ is $\lambda$ and $\tau (\mbox{Gr}(A(z, \cdot))) = 1$. Let $F$ be a measurable, nonempty, convex and compact valued correspondence from $\mbox{Gr}(A) \to \cM(Y)$ such that $F$ is sectionally continuous on $Z\times X$. Define a correspondence $\Pi$ from $Z$ to $\cM (X \times S \times Y)$ by letting
\begin{align*}
\Pi(z)
& = \{g(z) \diamond f(z, \cdot) \colon  g \mbox{ is a Borel measurable selection of } G,  \\
& \qquad f \mbox{ is a Borel measurable selection of } F \}.
\end{align*}
Then the correspondence $\Pi$ is nonempty and compact valued, and continuous.
\end{lem}

\begin{proof}
Let $\cS$ be the completion of $\cB(S)$ under the probability measure $\lambda$. By Lemma~\ref{lem-lusin prep}, $\tilde{A}(s) = \mbox{Gr}(A(s, \cdot))$ can be viewed as an $\cS$-measurable mapping from $S$ to the set of nonempty and compact subsets $\cK_{Z\times X}$ of $Z\times X$. For any $s \in S$, the correspondence $F_s = F(\cdot, s)$ is continuous on $\tilde{A}(s)$. By Lemma~\ref{lem-lusin}, there exists a measurable, nonempty and compact valued correspondence $\tilde{F}$ from $Z \times X \times S$ to $\cM(Y)$ and a Borel measurable subset $S'$ of $S$ with $\lambda(S')=1$
such that for each $s \in S'$, $\tilde{F}_s$ is continuous on $Z \times X$, and the restriction of $\tilde{F}_s$ to $\tilde{A}(s)$ is $F_s$.

By Lemma~\ref{lem-lusin} (Lusin's theorem), there exists a compact subset $S_1 \subseteq S'$ such that $\tilde{A}$ is continuous on $S_1$ and $\lambda(S_1) > \frac{1}{2}$. Let $K_1 = \tilde{A}(S_1)$. Then $K_1 \subseteq Z\times X$ is compact.

Let $C(K_1, \cK_{\cM(Y)})$ be the space of continuous functions from $K_1$ to $\cK_{\cM(Y)}$, where $\cK_{\cM(Y)}$ is the set of nonempty and compact subsets of $\cM(Y)$. Suppose that the restriction of $\cS$ on $S_1$ is $\cS_1$.  Let $\tilde{F}_1$ be the restriction of $\tilde{F}$ to $K_1 \times S_1$. Then $\tilde{F}_1$ can be viewed as an $\cS_1$-measurable function from $S_1$ to $C(K_1, \cK_{\cM(Y)})$ (see Theorem~4.55 in \cite{AB2006}). Again by Lemma~\ref{lem-lusin} (Lusin's theorem), there exists a compact subset of $S_1$, say itself, such that $\lambda(S_1) > \frac{1}{2}$ and $\tilde{F}_1$ is continuous on $S_1$. As a result, $\tilde{F}_1$ is a continuous correspondence on $\mbox{Gr}(A) \cap (S_1\times Z\times X)$, so is $F$. Let $\lambda_1$ be a probability measure on $S_1$ such that $\lambda_1(D) = \frac{\lambda(D)}{\lambda(S_1)}$ for any measurable subset $D\subseteq S_1$.

Fix any $z\in Z$ and $\tau \in G(z)$. By the definition of $G$, there exists a transition probability $\nu$ from $S$ to $X$ such that $\lambda\diamond \nu = \tau$. Define a correspondence $G_1$ from $Z$ to $\cM(X\times S)$ as follows: for any $z\in Z$, $G_1(z)$ is the set of all $\tau_1 = \lambda_1\diamond \nu$ such that $\tau = \lambda\diamond \nu \in G(z)$. It can be easily checked that $G_1$ is also a nonempty and compact valued, and continuous correspondence. Let
\begin{align*}
\Pi_1(z)
& = \{\tau_1 \diamond f(z, \cdot) \colon \tau_1 = \lambda_1\diamond \nu \in G_1(z),  \\
& \qquad f \mbox{ is a Borel measurable selection of } \tilde{F} \}.
\end{align*}
By Lemma~\ref{lem-measurable composition}, $\Pi_1$ is nonempty and compact valued, and continuous. Furthermore, it is easy to see that for any $z$, $\Pi_1(z)$ coincides with the set
$$\{(\lambda_1\diamond \nu) \diamond f(z, \cdot) \colon \lambda\diamond \nu \in G(z), f \mbox{ is a Borel measurable selection of } F \}.
$$

Repeat this procedure, one can find a sequence of compact subsets $\{S_t\}$ such that (1) for any $t\ge 1$, $S_t\subseteq S'$, $S_t \cap (S_1\cup \ldots S_{t-1}) = \emptyset$ and $\lambda(S_1\cup \ldots \cup S_t) \ge \frac{t}{t+1}$, (2) $F$ is continuous on $\mbox{Gr}(A) \cap (S_t\times Z\times X)$, $\lambda_t$ is a probability measure on $S_t$ such that $\lambda_t(D) = \frac{\lambda(D)}{\lambda(S_t)}$ for any measurable subset $D\subseteq S_t$, and (3) the correspondence
\begin{align*}
\Pi_t(z)
& = \{(\lambda_t\diamond \nu) \diamond f(z, \cdot) \colon \lambda\diamond \nu \in G(z),  \\
& \qquad f \mbox{ is a Borel measurable selection of } F \}.
\end{align*}
is nonempty and compact valued, and continuous.

Pick a sequence $\{z_k\}$, $\{\nu_k\}$ and $\{f_k\}$ such that $(\lambda\diamond \nu_k) \diamond  f_k(z_k, \cdot) \in \Pi(z_k)$, $z_k \to z_0$ and $(\lambda\diamond \nu_k) \diamond  f_k(z_k, \cdot)$ weakly converges to some $\kappa$. It is easy to see that $(\lambda_t \diamond \nu_k) \diamond  f_k(z_k, \cdot) \in \Pi_t(z_k)$ for each $t$. As $\Pi_1$ is compact valued and continuous, it has a subsequence, say itself, such that $z_k$ converges to some $z_0 \in Z$ and $(\lambda_1 \diamond \nu_k) \diamond  f_k(z_k, \cdot)$ weakly converges to some $(\lambda_1 \diamond \mu^1) \diamond  f^1(z_0, \cdot) \in \Pi_1(z_0)$. Repeat this procedure, one can get a sequence of $\{\mu^m\}$ and $f^m$. Let  $\mu(s) = \mu^m(s)$ and $f(z_0, s,x) = f^m(z_0, s,x)$ for any $x\in A(z_0, s)$ when $s\in S_m$. By Lemma~\ref{lem-countable partition}, $(\lambda \diamond \mu) \diamond  f(z_0, \cdot) = \kappa$, which implies that $\Pi$ is upper hemicontinuous.

Similarly, the compactness and lower hemicontinuity of $\Pi$ follow from the compactness and lower hemicontinuity of $\Pi_t$ for each $t$.
\end{proof}

\begin{lem}\label{lem-topology}
Let $S$ and $X$ be Polish spaces, and $A$ a measurable, nonempty and compact valued correspondence from $S$ to $X$. Suppose that $\lambda$ is a Borel probability measure on $S$ and $\{\nu_n\}_{1 \le n \le \infty}$ is a sequence of transition probabilities from $S$ to $\cM(X)$ such that $\nu_n(A(s)|s) = 1$ for each $s$ and $n$. For each $n \ge 1$, let $\tau_n = \lambda \diamond \nu_n$. Assume that the sequence $\{\tau_n\}$ of Borel probability measures on $S \times X$ converges weakly to a Borel probability measure $\tau_\infty$ on $S \times X$.
Let $\{g_n\}_{1 \le n \le \infty}$ be a sequence of functions satisfying the following three properties.
\begin{enumerate}

\item For each $n$ between $1$ and $\infty$, $g_n \colon S\times X \to \bR_+$ is measurable and sectionally continuous on $X$.

\item For any $s \in S$ and any sequence $x_n \to x_\infty$ in $X$, $g_n(s,x_n) \to g_\infty(s,x_\infty)$ as $n \to \infty$.

\item The sequence $\{g_n\}_{1 \le n \le \infty}$ is integrably bounded in the sense that there exists a $\lambda$-integrable function $\psi \colon S \to \bR_+$ such that for any $n$, $s$ and $x$, $g_n(s,x) \le \psi(s)$.
\end{enumerate}
Then we have
$$\int_{S\times X} g_n(s,x) \tau_n(\rmd (s,x))  \to \int_{S\times X} g_\infty(s,x) \tau_\infty(\rmd (s,x)).
$$
\end{lem}

\begin{proof}
By Theorem~2.1.3 in \cite{CDV2004}, for any integrably bounded function $g \colon S\times X\to \bR_+$ which is sectionally continuous on $X$, we have
\begin{equation}\label{equa-rcd}
\int_{S\times X} g(s,x) \tau_n(\rmd (s,x))  \to \int_{S\times X} g(s,x) \tau_\infty(\rmd (s,x)).
\end{equation}

Let $\{y_n\}_{1 \le n \le \infty}$ be a sequence such that $y_n  = \frac{1}{n}$ and $y_\infty = 0$. Then $y_n \to y_\infty$. Define a mapping $\tilde{g}$ from $S\times X \times  \{y_1, \ldots, y_\infty\}$ such that $\tilde{g}(s,x,y_n) = g_n(s,x)$. Then $\tilde{g}$ is measurable in $S$ and continuous in $X \times  \{y_1, \ldots, y_\infty\}$. Define a correspondence $G$ from $S$ to $X \times \{y_1, \ldots, y_\infty\} \times \bR_+$ such that
$$G(s) = \left\{(x,y_n,c) \colon c \in \tilde{g}(s,x,y_n), x\in A(s), 1\le n \le \infty \right\}.
$$
For any $s$, $A(s)\times \{y_1, \ldots, y_\infty\}$ is compact and $\tilde{g}(s, \cdot, \cdot)$ is continuous. By Lemma~\ref{lem-measurable selection}~(6), $G(s)$ is compact. By Lemma~\ref{lem-measurable correspondence}~(2), $G$ can be viewed as a measurable mapping from $S$ to the space of nonempty compact subsets of $X \times \{y_1, \ldots, y_\infty\} \times \bR_+$. Similarly, $A$ can be viewed as a measurable mapping from $S$ to the space of nonempty compact subsets of $X$.

Fix an arbitrary $\epsilon > 0$. By Lemma~\ref{lem-lusin} (Lusin's theorem), there exists a compact subset $S_1 \subseteq S$ such that $A$ and $G$ are continuous on $S_1$ and $\lambda(S\setminus S_1) < \epsilon$. Without loss of generality, we can assume that $\lambda(S\setminus S_1)$ is sufficiently small such that $\int_{S\setminus S_1} \psi(s) \lambda(\rmd s) < \frac{\epsilon}{6}$. Thus, for any $n$,
$$\int_{(S\setminus S_1) \times X} \psi(s) \tau_n(\rmd (s,x)) = \int_{(S\setminus S_1)} \psi(s) \nu_n(X) \lambda(\rmd s) < \frac{\epsilon}{6}.$$

By Lemma~\ref{lem-measurable selection}~(6), the set $E = \{(s,x) \colon s\in S_1, x\in A(s)\}$ is compact. Since $G$ is continuous on $S_1$, $\tilde{g}$ is continuous on $E \times \{y_1, \ldots, y_\infty\}$. Since $E \times \{y_1, \ldots, y_\infty\}$ is compact, $\tilde{g}$ is uniformly continuous on $E \times \{y_1, \ldots, y_\infty\}$. Thus, there exists a positive integer $N_1 > 0$ such that for any $n \ge N_1$, $|g_n(s,x) - g_\infty(s,x)| < \frac{\epsilon}{3}$ for any $(s,x) \in E$.

By Equation~(\ref{equa-rcd}), there exists a positive integer $N_2$ such that for any $n \ge N_2$,
$$\left|\int_{S\times X} g_\infty(s,x) \tau_n(\rmd (s,x)) - \int_{S\times X} g_\infty(s,x) \tau_\infty(\rmd (s,x)) \right| < \frac{\epsilon}{3}.
$$
Let $N_0 = \max\{N_1, N_2\}$. For any $n \ge N_0$,
\begin{align*}
& \quad \left|\int_{S\times X} g_n(s,x) \tau_n(\rmd (s,x)) - \int_{S\times X} g_\infty(s,x) \tau_\infty(\rmd (s,x)) \right|  \\
& \le \left|\int_{S\times X} g_n(s,x) \tau_n(\rmd (s,x)) - \int_{S\times X} g_\infty(s,x) \tau_n(\rmd (s,x)) \right|  \\
& + \left|\int_{S\times X} g_\infty(s,x) \tau_n(\rmd (s,x)) - \int_{S\times X} g_\infty(s,x) \tau_\infty(\rmd (s,x)) \right|  \\
& \le \left|\int_{S_1\times X} g_n(s,x) \tau_n(\rmd (s,x)) - \int_{S_1\times X} g_\infty(s,x) \tau_n(\rmd (s,x)) \right|  \\
& + \left|\int_{(S\setminus S_1) \times X} g_n(s,x) \tau_n(\rmd (s,x)) - \int_{(S\setminus S_1)\times X} g_\infty(s,x) \tau_n(\rmd (s,x)) \right|  \\
& + \left|\int_{S\times X} g_\infty(s,x) \tau_n(\rmd (s,x)) - \int_{S\times X} g_\infty(s,x) \tau_\infty(\rmd (s,x)) \right| \\
& \le \int_{E} \left| g_n(s,x) - g_\infty(s,x) \right|  \tau_n(\rmd (s,x)) + 2\cdot \int_{(S\setminus S_1) \times X} \psi(s) \tau_n(\rmd (s,x))   \\
& + \left|\int_{S\times X} g_\infty(s,x) \tau_n(\rmd (s,x)) - \int_{S\times X} g_\infty(s,x) \tau_\infty(\rmd (s,x)) \right|  \\
& < \frac{\epsilon}{3} + 2\cdot \frac{\epsilon}{6} + \frac{\epsilon}{3} \\
& = \epsilon.
\end{align*}
This completes the proof.
\end{proof}

The following result is Lemma~6 of \cite{RR2002}.

\begin{lem}\label{lem-RR}
Suppose that $H$ and $X$ are compact metric spaces. Let $P \colon H \times X \to \bR^n$ be a nonempty valued and upper hemicontinuous correspondence, and the mappings $f \colon H \to \cM(X)$ and  $\mu \colon H \to \triangle(X)$ be measurable. In addition, suppose that $\mu(\cdot|h) = p(h,\cdot)\circ f(\cdot|h)$ such that $p(h,\cdot)$ is a measurable selection of $P(h, \cdot)$. Then there exists a jointly Borel measurable selection $g$ of $P$ such that $\mu(\cdot|h) = g(h,\cdot)\circ f(\cdot|h)$; that is, $g(h, x) = p(h, x)$ for $f(\cdot|h)$-almost all $x$.
\end{lem}

\subsection{Discontinuous games with endogenous stochastic sharing rules}\label{subsec-discontinuous game}

\cite{SZ1990} proved the existence of a Nash equilibrium in discontinuous games with endogenous sharing rules. In particular, they considered a static game with a payoff correspondence $P$ that is bounded, nonempty, convex and compact valued, and upper hemicontinuous. They showed that there exists a Borel measurable selection $p$ of the payoff correspondence, namely the endogenous sharing rule, and a mixed strategy profile $\alpha$ such that $\alpha$ is a Nash equilibrium when players take $p$ as the payoff function.

In this subsection, we shall consider discontinuous games with endogenous stochastic sharing rules. That is, we allow the payoff correspondence to depend on some state variable in a measurable way as follows:
\begin{enumerate}
  \item let $S$ be a Borel subset of a Polish space, $Y$ a Polish space, and $\lambda$ a Borel probability measure on $S$;
  \item $D$ is a $\cB(S) \otimes \cB(Y)$-measurable subset of $S\times Y$, where $D(s)$ is compact for all $s\in S$ and $\lambda\left(\{s\in S\colon D(s) \neq \emptyset\}\right) > 0$;
  \item $X = \prod_{1\le i \le n} X_i$, where each $X_i$ is a Polish space;
  \item for each $i$, $A_i$ is a measurable, nonempty and compact valued correspondence from $D$ to $X_i$, which is sectionally continuous on $Y$;
  \item $A = \prod_{1 \le i\le n}A_i$ and $E = \mbox{Gr}(A)$;
  \item $P$ is a bounded, measurable, nonempty, convex and compact valued correspondence from $E$ to $\bR^n$ which is essentially sectionally upper hemicontinuous on $Y\times X$.
\end{enumerate}
A stochastic sharing rule at $(s,y) \in D$ is a Borel measurable selection of the correspondence $P(s,y,\cdot)$; i.e., a Borel measurable function $p \colon A(s,y) \to \bR^n$ such that $p(x) \in P(s,y,x)$ for  all $x \in A(s,y)$. Given $(s,y) \in D$, $P(s,y,\cdot)$ represents the set of all possible payoff profiles, and a sharing rule $p$ is a particular choice of the payoff profile.

%A behavioral strategy of player~$i$ is a measurable mapping from $H\times S$ to $\cM(X_i)$. A solution is a sharing rule $q$ and a behavioral strategy profile $f = (f_1, \ldots, f_n)$ with the property that, given the sharing rule, each player's strategy is a best response to the strategies of other players for any $h\in H$. More precisely,  for each $i$ and for any behavioral strategy $g_i$,
%\begin{align*}
%E_i(q_i, f_i, f_{-i})
%& =   \int_S \int_{X_{-i}}\int_{X_i} q_i(s, x_i, x_{-i}) f_i(s, \rmd x_i) f_{-i}(s, \rmd x_{-i}) \lambda(\rmd s) \\
%& \ge \int_S \int_{X_{-i}}\int_{X_i} q_i(s, x_i, x_{-i}) g_i(s, \rmd x_i) f_{-i}(s, \rmd x_{-i}) \lambda(\rmd s).
%\end{align*}

Now we shall prove the following proposition.

\begin{prop}\label{prop-extension SZ}
There exists a $\cB(D)$-measurable, nonempty and compact valued correspondence $\Phi$ from $D$ to $\bR^n \times \cM(X) \times \triangle(X)$ such that $\Phi$ is essentially sectionally upper hemicontinuous on $Y$, and for $\lambda$-almost all $s \in S$ with $D(s)\neq \emptyset$ and $y\in D(s)$, $\Phi(s,y)$ is the set of points $(v,\alpha,\mu)$ that
\begin{enumerate}
  \item $v = \int_{X} p(s,y,x) \alpha(\rmd x)$ such that $p(s,y,\cdot)$ is a Borel measurable selection of $P(s,y, \cdot)$;\footnote{Note that we require $p(s,y,\cdot)$ to be measurable for each $(s,y)$, but $p$ may not be jointly measurable.}
  \item $\alpha \in \otimes_{i\in I}\cM(A_i(s,y))$ is a Nash equilibrium in the subgame $(s,y)$ with payoff $p(s,y, \cdot)$ and action space $A_i(s,y)$ for each player $i$;
  \item $\mu = p(s,y,\cdot)\circ \alpha$.\footnote{The finite measure $\mu = p(s,y,\cdot) \circ \alpha$ if $\mu(B) = \int_{B} p(s,y, x) \alpha(\rmd x)$  for any Borel subset $B \subseteq X$.}
\end{enumerate}
In addition, denote the restriction of $\Phi$ on the first component $\bR^n$ as $\Phi|_{\bR^n}$, which is a correspondence from $D$ to $\bR^n$. Then $\Phi|_{\bR^n}$ is bounded, measurable, nonempty and compact valued, and essentially sectionally upper hemicontinuous on $Y$.
\end{prop}

If $D$ is a closed subset, $P$ is upper hemicontinuous on $E$ and $A_i$ is continuous on $D$ for each $i \in I$, then Proposition~\ref{prop-extension SZ} is reduced to be the following lemma (see  \cite{SZ1990} and \citet[Lemma~4]{RR2002}).

%%%YN: In RR, D is compact. Will this be a problem?

\begin{lem}\label{lem-SZ}
Assume that $D$ is a closed subset, $P$ is upper hemicontinuous on $E$ and $A_i$ is continuous on $D$ for each $i \in I$.
Consider the correspondence $\Phi \colon D \to \bR^n \times \cM(X) \times \triangle(X)$ defined as follows: $(v,\alpha,\mu) \in \Phi(s,y)$ if
\begin{enumerate}
  \item $v = \int_{X} p(s,y,x) \alpha(\rmd x)$ such that $p(s,y,\cdot)$ is a Borel measurable selection of $P(s,y, \cdot)$;
  \item $\alpha \in \otimes_{i\in I}\cM(A_i(s,y))$ is a Nash equilibrium in the subgame $(s,y)$ with payoff $p(s,y, \cdot)$ and action space $A_i(s,y)$ for each player $i$;
  \item $\mu = p(s,y,\cdot)\circ \alpha$.
\end{enumerate}
Then $\Phi$ is nonempty and compact valued, and upper hemicontinuous on $D$.
\end{lem}

We shall now prove Proposition~\ref{prop-extension SZ}.

\begin{proof}[Proof of Proposition~\ref{prop-extension SZ}]
There exists a Borel subset $\hat{S}\subseteq S$ with $\lambda(\hat{S}) = 1$ such that $D(s) \neq \emptyset$ for each $s\in \hat{S}$, and $P$ is sectionally upper hemicontinuous on $Y$ when it is restricted on $D\cap (\hat{S}\times Y)$. Without loss of generality, we assume that $\hat{S} = S$.

Suppose that $\cS$ is the completion of $\cB(S)$ under the probability measure $\lambda$. Let $\cD$ and $\cE$ be the restrictions of $\cS\otimes \cB(Y)$ and $\cS\otimes \cB(Y)\otimes \cB(X)$ on $D$ and $E$, respectively.

Define a correspondence $\tilde{D}$ from $S$ to $Y$ such that $\tilde{D}(s) = \{y\in Y\colon (s,y)\in D\}$. Then $\tilde{D}$ is nonempty and compact valued. By Lemma~\ref{lem-measurable correspondence}~(4), $\tilde{D}$ is $\cS$-measurable.

Since $\tilde{D}(s)$ is compact and $A(s, \cdot)$ is upper hemicontinuous for any $s\in S$, $E(s)$ is compact by Lemma~\ref{lem-measurable selection}~(6).
Define a correspondence $\Gamma$ from $S$ to $Y\times X \times \bR^n$ as $\Gamma(s) = \mbox{Gr}(P(s,\cdot,\cdot))$.  For all $s$, $P(s,\cdot,\cdot)$ is bounded, upper hemicontinuous and compact valued on $E(s)$, hence it has a compact graph. As a result, $\Gamma$ is compact valued.  By Lemma~\ref{lem-measurable correspondence}~(1), $P$ has an $\cS\otimes \cB(Y\times X\times \bR^n)$-measurable graph. Since $\mbox{Gr}(\Gamma) = \mbox{Gr}(P)$, $\mbox{Gr}(\Gamma)$ is $\cS\otimes \cB(Y\times X\times \bR^n)$-measurable. Due to Lemma~\ref{lem-measurable correspondence}~(4), the correspondence $\Gamma$ is $\cS$-measurable. We can view $\Gamma$ as a function from $S$ into the space $\cK$ of nonempty compact subsets of $Y\times X\times \bR^n$.  By Lemma~\ref{lem-polish compact}, $\cK$ is a Polish space endowed with the Hausdorff metric topology. Then by Lemma~\ref{lem-measurable correspondence} (2), $\Gamma$ is an $\cS$-measurable function from $S$ to $\cK$. One can also define a correspondence $\tilde{A_i}$ from $S$ to $Y \times X$ as $\tilde{A_i}(s) = \mbox{Gr}(A_i(s,\cdot))$. It is easy to show that $\tilde{A_i}$ can be viewed as an $\cS$-measurable function from $S$ to the space of nonempty compact subsets of $Y\times X$, which is endowed with the Hausdorff metric topology. By a similar argument, $\tilde{D}$ can be viewed as an $\cS$-measurable function from $S$ to the space of nonempty compact subsets of $Y$.

By Lemma~\ref{lem-lusin} (Lusin's~Theorem), there exists a compact subset $S_1 \subseteq S$ such that $\lambda(S\setminus S_1) < \frac{1}{2}$, $\Gamma$, $\tilde{D}$ and $\{\tilde{A_i}\}_{1\le i \le n}$ are continuous functions on $S_1$. By Lemma~\ref{lem-measurable correspondence}~(3), $\Gamma$, $\tilde{D}$ and $\tilde{A_i}$ are continuous correspondences on $S_1$. Let $D_1  = \{(s,y)\in D\colon s\in S_1, y \in \tilde{D}(s)\}$. Since $S_1$ is compact and $\tilde{D}$ is continuous, $D_1$ is compact (see Lemma~\ref{lem-measurable selection}~(6)). Similarly, $E_1 = E \cap (S_1\times Y\times X)$ is also compact. Thus, $P$ is an upper hemicontinuous correspondence on $E_1$. Define a correspondence $\Phi_1$ from $D_1$ to $\bR^n \times \cM(X) \times \triangle(X)$ as in Lemma~\ref{lem-SZ}, then it is nonempty and compact valued, and upper hemicontinuous on $D_1$.

Following the same procedure, for any $m\ge 1$, there exists a compact subset $S_m \subseteq S$ such that (1) $S_m\cap (\cup_{1\le k \le m-1} S_k) = \emptyset$ and $D_m = D\cap (S_m \times Y)$ is compact, (2) $\lambda(S_m) > 0$ and $\lambda\left(S\setminus (\cup_{1\le k \le m}S_m)\right) < \frac{1}{2m}$, and (3) there is a nonempty and compact valued, upper hemicontinuous correspondence $\Phi_m$ from $D_m$ to $\bR^n \times \cM(X) \times \triangle(X)$, which satisfies conditions~(1)-(3) in Lemma~\ref{lem-SZ}. Thus, we have countably many disjoint sets $\{S_m\}_{m\ge 1}$ such that (1) $\lambda(\cup_{m \ge 1} S_m) = 1$, (2) $\Phi_m$ is nonempty and compact valued, and upper hemicontinuous on each $D_m$, $m \ge 1$.

Since $A_i$ is a $\cB(S)\otimes\cB(Y)$-measurable, nonempty and compact valued correspondence, it has a Borel measurable selection $a_i$ by Lemma~\ref{lem-measurable selection}~(3). Fix a Borel measurable selection $p$ of $P$ (such a selection exists also due to Lemma~\ref{lem-measurable selection}~(3)). Define a mapping $(v_0, \alpha_0, \mu_0)$ from $D$ to $\bR^n \times \cM(X) \times \triangle(X)$ such that (1) $\alpha_i(s,y) = \delta_{a_i(s,y)}$ and $\alpha_0(s,y) = \otimes_{i \in I} \alpha_i(s,y)$; (2) $v_0(s,y) = p(s,y,a_1(s,y). \ldots, a_n(s,y)) $ and (3) $\mu_0(s,y) = p(s,y,\cdot)\circ \alpha_0$. Let $D_0 = D\setminus (\cup_{m \ge 1} D_m)$ and $\Phi_0(s,y) = \{(v_0(s,y), \alpha_0(s,y), \mu_0(s,y))\}$ for $(s,y) \in D_0$. Then, $\Phi_0$ is $\cB(S)\otimes\cB(Y)$-measurable, nonempty and compact valued.

%Let $D_0 = D\setminus (\cup_{m \ge 1} D_m)$. Define a correspondence $\Phi_0$ from $D_0$ to $\bR^n \times \cM(X) \times \triangle(X)$ as follow: for any point $(s,y) \in D_0$, $(v,\alpha,\mu) \in \Phi_0(s,y)$ if
%\begin{enumerate}
%  \item $v = \int_{X} p(s,y,x) \alpha(\rmd x)$ such that $p(s,y,\cdot)$ is a measurable selection of $P(s,y, \cdot)$;
%  \item $\alpha \in \otimes_{i\in I}\cM(A_i(s,y))$;
%  \item $\mu = p(s,y,\cdot)\circ \alpha$.
%\end{enumerate}

Let $\Phi(s,y) = \Phi_m(s,y)$ if $(s,y) \in D_m$ for some $m \ge 0$. Then, $\Phi(s,y)$ satisfies conditions~(1)-(3) if $(s,y) \in D_m$ for $m \ge 1$. That is, $\Phi$ is $\cB(D)$-measurable, nonempty and compact valued, and essentially sectionally upper hemicontinuous on $Y$, and satisfies conditions~(1)-(3) for $\lambda$-almost all $s\in S$.

Then consider $\Phi|_{\bR^n}$, which is the restriction of $\Phi$ on the first component $\bR^n$. Let $\Phi_m|_{\bR^n}$ be the restriction of $\Phi_m$ on the first component $\bR^n$ with the domain $D_m$ for each $m \ge 0$. It is obvious that $\Phi_0|_{\bR^n}$ is measurable, nonempty and compact valued. For each $m \ge 1$, $D_m$ is compact and $\Phi_m$ is upper hemicontinuous and compact valued. By Lemma~\ref{lem-measurable selection}~(6), $\mbox{Gr}(\Phi_m)$ is compact. Thus, $\mbox{Gr}(\Phi_m|_{\bR^n})$ is also compact. By Lemma~\ref{lem-measurable selection}~(4), $\Phi_m|_{\bR^n}$ is measurable. In addition, $\Phi_m|_{\bR^n}$ is nonempty and compact valued, and upper hemicontinuous on $D_m$. Notice that $\Phi|_{\bR^n}(s,y) = \Phi_m|_{\bR^n}(s,y)$ if $(s,y) \in D_m$ for some $m \ge 0$. Thus, $\Phi|_{\bR^n}$ is measurable, nonempty and compact valued, and essentially sectionally upper hemicontinuous on $Y$.

The proof is complete.
\end{proof}

\subsection{Proofs of Theorem~\ref{thm-general} and Proposition~\ref{prop-general}}\label{subsec-main proof}

\subsubsection{Backward induction}\label{subsubsec-backward}

For any $t\ge 1$, suppose that  the correspondence $Q_{t+1}$ from $H_{t}$ to $\bR^n$ is bounded, measurable, nonempty and compact valued, and essentially sectionally upper hemicontinuous on $X^t$. For any $h_{t-1} \in H_{t-1}$ and $x_{t} \in A_{t}(h_{t-1})$, let
\begin{align*}
P_{t}(h_{t-1}, x_{t})
& = \int_{S_{t}} Q_{t+1}(h_{t-1}, x_{t}, s_{t}) f_{t0}(\rmd s_{t}|h_{t-1}) \\
& = \int_{S_{t}} Q_{t+1}(h_{t-1}, x_{t}, s_{t}) \varphi_{t0}(h_{t-1}, s_{t}) \lambda_{t}(\rmd s_{t}).
\end{align*}
It is obvious that the correspondence $P_{t}$ is measurable and nonempty valued.  Since $Q_{t+1}$ is bounded, $P_t$ is bounded. For $\lambda^t$-almost all $s^{t} \in S^{t}$, $Q_{t+1}(\cdot, s^{t})$ is bounded and upper hemicontinuous on $H_t(s^t)$, and $\varphi_{t 0}(s^{t}, \cdot)$ is continuous on $\mbox{Gr}(A_0^{t})(s^{t})$. As $\varphi_{t0}$ is integrably bounded, $P_{t}(s^{t-1},\cdot)$ is also upper hemicontinuous on $\mbox{Gr}(A^{t})(s^{t-1})$ for $\lambda^{t-1}$-almost all $s^{t-1} \in S^{t-1}$ (see Lemma~\ref{lem-lyapunov}); that is, the correspondence $P_t$ is essentially sectionally upper hemicontinuous on $X^{t}$. Again by Lemma~\ref{lem-lyapunov}, $P_{t}$ is convex and compact valued since $\lambda_{t}$ is an atomless probability measure. That is, $P_t \colon \mbox{Gr}(A^{t}) \to \bR^n$ is a bounded, measurable, nonempty, convex and compact valued correspondence which is essentially sectionally upper hemicontinuous on $X^t$.

By Proposition~\ref{prop-extension SZ}, there exists a bounded, measurable, nonempty and compact valued correspondence $\Phi_t$ from $H_{t-1}$ to $\bR^n \times \cM(X_t) \times \triangle(X_t)$ such that $\Phi_t$ is essentially sectionally upper hemicontinuous on $X^{t-1}$, and for $\lambda^{t-1}$-almost all $h_{t-1} \in H_{t-1}$, $(v,\alpha,\mu) \in \Phi_t(h_{t-1})$ if
\begin{enumerate}
  \item $v = \int_{A_t(h_{t-1})} p_t(h_{t-1},x) \alpha(\rmd x)$ such that $p_t(h_{t-1},\cdot)$ is a Borel measurable selection of $P_t(h_{t-1}, \cdot)$;
  \item $\alpha \in \otimes_{i\in I}\cM(A_{ti}(h_{t-1}))$ is a Nash equilibrium in the subgame $h_{t-1}$ with payoff $p_t(h_{t-1}, \cdot)$ and action space $\prod_{i\in I} A_{ti}(h_{t-1})$;
  \item $\mu = p_t(h_{t-1},\cdot)\circ \alpha$.
  %\footnote{The finite measure $\mu = p^t(h^{t-1},\cdot) \circ \alpha$ if $\mu(E) = \int_E p^t(h^{t-1}, x) \alpha(\rmd x)$  for any Borel set $E \subseteq X^t$.}
\end{enumerate}
Denote the restriction of $\Phi_t$ on the first component $\bR^n$ as $\Phi(Q_{t+1})$, which is a correspondence from $H_{t-1}$ to $\bR^n$. By Proposition~\ref{prop-extension SZ}, $\Phi(Q_{t+1})$ is bounded, measurable, nonempty and compact valued, and essentially sectionally upper hemicontinuous  on $X^{t-1}$.

\subsubsection{Forward induction}\label{subsubsec-forward}

%Denote $H^{-1} = \{h^{-1}\}$. Let $f^{0}(h^{-1}) = \delta_{s^0}$, where $\delta_{s^0}$ is the delta measure concentrated at  $s^0$. Let $P^0(x^0) = Q^1(x^0, s^0)$ for any $x^0 \in X^0$.

The following proposition presents the result on the step of forward induction.

\begin{prop}\label{prop-forward}
For any $t\ge 1$ and any Borel measurable selection $q_{t}$ of $\Phi(Q_{t+1})$, there exists a Borel measurable selection $q_{t+1}$ of $Q_{t+1}$ and a Borel measurable mapping $f_{t} \colon H_{t-1}\to \otimes_{i\in I}\cM(X_{ti})$ such that for $\lambda^{t-1}$-almost all $h_{t-1} \in H_{t-1}$,
\begin{enumerate}
   \item $f_t(h_{t-1}) \in \otimes_{i\in I}\cM(A_{ti}(h_{t-1}))$;
   \item $q_{t}(h_{t-1}) = \int_{A_{t}(h_{t-1})} \int_{S_{t}} q_{t+1}(h_{t-1},x_t,s_{t}) f_{t0}(\rmd s_{t}|h_{t-1}) f_t(\rmd x_t|h_{t-1})$;
   \item $f_t(\cdot|h_{t-1})$ is a Nash equilibrium in the subgame $h_{t-1}$ with action spaces $A_{ti}(h_{t-1}), i \in I$ and the payoff functions
        $$\int_{S_{t}} q_{t+1}(h_{t-1},\cdot,s_{t}) f_{t0}(\rmd s_{t}|h_{t-1}).$$
\end{enumerate}
\end{prop}

\begin{proof}
We divide the proof into three steps. In step~1, we show that there exist Borel measurable mappings $f_t\colon H_{t-1} \to \otimes_{i\in I}\cM(X_{ti})$ and $\mu_t\colon H_{t-1}\to \triangle(X_t)$ such that $(q_t,f_t,\mu_t)$ is a selection of $\Phi_t$. In step~2, we obtain a Borel measurable selection $g_t$ of $P_t$ such that for $\lambda^{t-1}$-almost all $h_{t-1} \in H_{t-1}$,
\begin{enumerate}
  \item $q_t(h_{t-1}) = \int_{A_t(h_{t-1})} g_t(h_{t-1},x) f_t(\rmd x|h_{t-1})$;
  \item $f_t(h_{t-1})$ is a Nash equilibrium in the subgame $h_{t-1}$ with payoff $g_t(h_{t-1}, \cdot)$ and action space $A_{t}(h_{t-1})$;
\end{enumerate}
In step~3, we show that there exists a Borel measurable selection $q_{t+1}$ of $Q_{t+1}$ such that for all $h_{t-1} \in H_{t-1}$ and $x_t \in A_t(h_{t-1})$,
$$g_t(h_{t-1},x_t) = \int_{S_{t}} q_{t+1}(h_{t-1},x_t,s_{t}) f_{t0}(\rmd s_{t}|h_{t-1}).$$
Combining Steps~1-3, the proof is complete.

\

Step~1. Let $\Psi_t\colon \mbox{Gr}(\Phi_t(Q_{t+1})) \to \cM(X_t)\times \triangle(X_t)$ be
$$\Psi_t(h_{t-1}, v) = \{(\alpha, \mu) \colon (v,\alpha,\mu)\in \Phi_t(h_{t-1})\}.$$
Recall the construction of $\Phi_t$ and the proof of Proposition~\ref{prop-extension SZ}, $H_{t-1}$ can be divided into countably many Borel subsets $\{H^m_{t-1}\}_{m \ge 0}$ such that
\begin{enumerate}
  \item $H_{t-1} = \cup_{m \ge 0} H_{t-1}^m$ and $\frac{\lambda^{t-1}(\cup_{m \ge 1}\mbox{proj}_{S^{t-1}}(H_{t-1}^m))}{\lambda^{t-1}(\mbox{proj}_{S^{t-1}}(H_{t-1}))} = 1$, where $\mbox{proj}_{S^{t-1}}(H_{t-1}^m)$ and $\mbox{proj}_{S^{t-1}}(H_{t-1})$ are projections of $H_{t-1}^m$ and $H_{t-1}$ on $S^{t-1}$;
  \item for $m \ge 1$, $H_{t-1}^m$ is compact, $\Phi_t$ is upper hemicontinuous on $H_{t-1}^m$, and $P_t$ is upper hemicontinuous on
      $$\{(h_{t-1}, x_t)\colon h_{t-1} \in H_{t-1}^m, x_t \in A_t(h_{t-1})\};
      $$
  \item there exists a Borel measurable mapping $(v_0,\alpha_0,\mu_0)$ from $H^0_{t-1}$ to $\bR^n \times \cM(X_t) \times \triangle(X_t)$ such that $\Phi_t(h_{t-1}) \equiv \{(v_0(h_{t-1}),\alpha_0(h_{t-1}),\mu_0(h_{t-1}))\}$ for any $h_{t-1} \in H^0_{t-1}$.
\end{enumerate}
Denote the restriction of $\Phi_t$ on $H_{t-1}^m$ as $\Phi^m_t$. For $m \ge 1$, $\mbox{Gr}(\Phi^m_t)$ is compact, and hence the correspondence $\Psi^m_t(h_{t-1}, v) = \{(\alpha, \mu) \colon (v,\alpha,\mu)\in \Phi^m_t(h_{t-1})\}$ has a compact graph. For $m \ge 1$, $\Psi^m_t$ is measurable by Lemma~\ref{lem-measurable selection}~(4), and has a Borel measurable selection $\psi^m_t$ due to Lemma~\ref{lem-measurable selection}~(3). Define $\psi^0_t(h_{t-1}, v_0(h_{t-1})) = (\alpha_0(h_{t-1}),\mu_0(h_{t-1}))$ for $h_{t-1}\in H^0_{t-1}$. For $(h_{t-1}, v) \in \mbox{Gr}(\Phi(Q_{t+1}))$, let $\psi_t(h_{t-1}, v) = \psi^m_t(h_{t-1}, v)$ if $h_{t-1}\in H^m_{t-1}$. Then $\psi_t$ is a Borel measurable selection of $\Psi_t$.

Given a Borel measurable selection $q_t$ of $\Phi(Q_{t+1})$,  let
$$\phi_t(h_{t-1}) = (q_t(h_{t-1}), \psi_t(h_{t-1}, q_t(h_{t-1}))).$$
Then $\phi_t$ is a Borel measurable selection of $\Phi_t$.
Denote $\tilde{H}_{t-1} = \cup_{m \ge 1}H^m_{t-1}$. By the construction of $\Phi_t$, there exists Borel measurable mappings $f_t\colon H_{t-1} \to \otimes_{i\in I}\cM(X_{ti})$ and $\mu_t\colon H_{t-1}\to \triangle(X_t)$ such that for all $h_{t-1} \in \tilde{H}_{t-1}$,
\begin{enumerate}
  \item $q_t(h_{t-1}) = \int_{A_t(h_{t-1})} p_t(h_{t-1},x) f_t(\rmd x|h_{t-1})$ such that $p_t(h_{t-1},\cdot)$ is a Borel measurable selection of $P_t(h_{t-1}, \cdot)$;
  \item $f_t(h_{t-1}) \in \otimes_{i\in I}\cM(A_{ti}(h_{t-1}))$ is a Nash equilibrium in the subgame $h_{t-1}$ with payoff $p_t(h_{t-1}, \cdot)$ and action space $\prod_{i\in I} A_{ti}(h_{t-1})$;
  \item $\mu_t(\cdot|h_{t-1}) = p_t(h_{t-1},\cdot)\circ f_t(\cdot|h_{t-1})$.
\end{enumerate}

\

Step~2. Since $P_t$ is upper hemicontinuous on $\{(h_{t-1}, x_t)\colon h_{t-1} \in H_{t-1}^m, x_t \in A_t(h_{t-1})\}$, due to Lemma~\ref{lem-RR}, there exists a Borel measurable mapping $g^m$ such that (1) $g^m(h_{t-1}, x_t) \in P_t(h_{t-1}, x_t)$ for any $h_{t-1} \in H_{t-1}^m$ and $x_t \in A_t(h_{t-1})$, and (2) $g^m(h_{t-1}, x_t) = p_t(h_{t-1}, x_t)$ for $f_{t}(\cdot|h_{t-1})$-almost all $x_{t}$. Fix an arbitrary Borel measurable selection $g'$ of $P_t$. Define a Borel measurable mapping from $\mbox{Gr}(A_t)$ to $\bR^n$ as
$$g(h_{t-1},x_t) =
\begin{cases}
g^m(h_{t-1},x_t)  & \mbox{if } h_{t-1} \in H^m_{t-1} \mbox{ for } m \ge 1; \\
g'(h_{t-1},x_t) & \mbox{otherwise.}
\end{cases}
$$
Then $g$ is a Borel measurable selection of $P_t$.

In a subgame $h_{t-1} \in \tilde{H}_{t-1}$,  let
$$B_{ti}(h_{t-1}) = \{ y_i \in A_{ti}(h_{t-1}) \colon$$
$$ \int_{A_{t(-i)}(h_{t-1})} g_i(h_{t-1}, y_i,x_{t(-i)}) f_{t(-i)}(\rmd x_{t(-i)}|h_{t-1}) > \int_{A_{t}(h_{t-1})} p_{ti}(h_{t-1}, x_t) f_t(\rmd x_t|h_{t-1}) \}.
$$
Since $g(h_{t-1}, x_t) = p_t(h_{t-1}, x_t)$ for $f_{t}(\cdot|h_{t-1})$-almost all $x_{t}$,
$$\int_{A_{t}(h_{t-1})} g(h_{t-1}, x_t) f_t(\rmd x_t|h_{t-1}) =  \int_{A_{t}(h_{t-1})} p_{t}(h_{t-1}, x_t) f_t(\rmd x_t|h_{t-1}).$$
Thus, $B_{ti}$ is a measurable correspondence from $\tilde{H}_{t-1}$ to $A_{ti}(h_{t-1})$. Let $B_{ti}^c(h_{t-1}) = A_{ti}(h_{t-1})\setminus B_{ti}(h_{t-1})$ for each $h_{t-1} \in H_{t-1}$. Then $B_{ti}^c$ is a measurable and closed valued correspondence, which has a Borel measurable graph by Lemma~\ref{lem-measurable correspondence}. As a result, $B_{ti}$ also has a Borel measurable graph. As $f_t(h_{t-1})$ is a Nash equilibrium in the subgame $h_{t-1}\in \tilde{H}_{t-1}$ with payoff $p_t(h_{t-1}, \cdot)$, $f_{ti}(B_{ti}(h_{t-1})|h_{t-1}) = 0$.

Denote $\beta_i(h_{t-1}, x_t) = \min P_{ti}(h_{t-1},x_t)$, where $P_{ti}(h_{t-1},x_t)$ is the projection of $P_{t}(h_{t-1},x_t)$ on the $i$-th dimension. Then the correspondence $P_{ti}$ is measurable and compact valued, and $\beta_i$ is Borel measurable. Let $\Lambda_i(h_{t-1}, x_t) = \{\beta_i(h_{t-1}, x_t)\} \times [0, \gamma]^{n-1}$, where $\gamma > 0$ is the upper bound of $P_t$. Denote $\Lambda_i'(h_{t-1}, x_t) = \Lambda_i(h_{t-1}, x_t) \cap P_t(h_{t-1}, x_t)$. Then $\Lambda'_i$ is a measurable and compact valued correspondence, and hence has a Borel measurable selection $\beta'_i$. Note that $\beta'_i$ is a Borel measurable selection of $P_t$. Let
$$g_t(h_{t-1},x_t) =$$
$$\begin{cases}
\beta_i'(h_{t-1},x_t)  & \mbox{if } h_{t-1} \in \tilde{H}_{t-1}, x_{ti} \in B_{ti}(h_{t-1}) \mbox{ and } x_{tj} \notin B_{tj}(h_{t-1}), \forall j \neq i; \\
g(h_{t-1},x_t) & \mbox{otherwise.}
\end{cases}
$$
Notice that
\begin{align*}
& \quad \{(h_{t-1}, x_t) \in \mbox{Gr}(A_t) \colon h_{t-1} \in \tilde{H}_{t-1}, x_{ti} \in B_{ti}(h_{t-1}) \mbox{ and } x_{tj} \notin B_{tj}(h_{t-1}), \forall j \neq i;\}  \\
& = \mbox{Gr}(A_t) \cap \cup_{i \in I} \left( (\mbox{Gr}(B_{ti}) \times \prod_{j \neq i} X_{tj} ) \setminus (\cup_{j \neq i}  (\mbox{Gr}(B_{tj}) \times \prod_{k \neq j} X_{tk}) ) \right),
\end{align*}
which is a Borel set. As a result, $g_t$  is a Borel measurable selection of $P_t$. Moreover, $g_t(h_{t-1}, x_t) = p_t(h_{t-1}, x_t)$ for all $h_{t-1} \in \tilde{H}_{t-1}$ and $f_{t}(\cdot|h_{t-1})$-almost all $x_{t}$.

Fix a subgame $h_{t-1} \in \tilde{H}_{t-1}$. We will show that $f_t(\cdot|h_{t-1})$ is a Nash equilibrium given the payoff $g_t(h_{t-1}, \cdot)$ in the subgame $h_{t-1}$. Suppose that player~$i$ deviates to some action $\tilde{x}_{ti}$.

If $\tilde{x}_{ti} \in B_{ti}(h_{t-1})$, then player~$i$'s expected payoff is
\begin{align*}
& \qquad \int_{A_{t(-i)}(h_{t-1})} g_{ti}(h_{t-1}, \tilde{x}_{ti}, x_{t(-i)}) f_{t(-i)}(\rmd x_{t(-i)} |h_{t-1})  \\
& = \int_{\prod_{j \neq i} B^c_{tj}(h_{t-1})} g_{ti}(h_{t-1}, \tilde{x}_{ti}, x_{t(-i)}) f_{t(-i)}(\rmd x_{t(-i)} |h_{t-1})  \\
& = \int_{\prod_{j \neq i} B^c_{tj}(h_{t-1})} \beta_i(h_{t-1}, \tilde{x}_{ti}, x_{t(-i)}) f_{t(-i)}(\rmd x_{t(-i)} |h_{t-1})  \\
& \le \int_{\prod_{j \neq i} B^c_{tj}(h_{t-1})} p_{ti}(h_{t-1}, \tilde{x}_{ti}, x_{t(-i)}) f_{t(-i)}(\rmd x_{t(-i)} |h_{t-1})  \\
& = \int_{A_{t(-i)}(h_{t-1})} p_{ti}(h_{t-1}, \tilde{x}_{ti}, x_{t(-i)}) f_{t(-i)}(\rmd x_{t(-i)} |h_{t-1}) \\
& \le \int_{A_{t}(h_{t-1})} p_{ti}(h_{t-1}, x_{t}) f_{t}(\rmd x_{t} |h_{t-1}) \\
& = \int_{A_{t}(h_{t-1})} g_{ti}(h_{t-1}, x_{t}) f_{t}(\rmd x_{t} |h_{t-1}).
\end{align*}
The first and the third equalities hold since $f_{tj}(B_{tj}(h_{t-1})|h_{t-1}) = 0$ for each $j$, and hence $f_{t(-i)}(\prod_{j \neq i} B^c_{tj}(h_{t-1})|h_{t-1}) = f_{t(-i)}(A_{t(-i)}(h_{t-1})|h_{t-1})$. The second equality and the first inequality are due to the fact that $g_{ti}(h_{t-1}, \tilde{x}_{ti}, x_{t(-i)})  = \beta_i(h_{t-1}, \tilde{x}_{ti}, x_{t(-i)}) = \min P_{ti}(h_{t-1}, \tilde{x}_{ti}, x_{t(-i)}) \le p_{ti}(h_{t-1}, \tilde{x}_{ti}, x_{t(-i)})$ for $x_{t(-i)} \in \prod_{j \neq i} B^c_{tj}(h_{t-1})$. The second inequality holds since $f_t(\cdot|h_{t-1})$ is a Nash equilibrium given the payoff $p_t(h_{t-1}, \cdot)$ in the subgame $h_{t-1}$. The fourth equality follows from the fact that $g_t(h_{t-1}, x_t) = p_t(h_{t-1}, x_t)$ for $f_{t}(\cdot|h_{t-1})$-almost all $x_{t}$.

If $\tilde{x}_{ti} \notin B_{ti}(h_{t-1})$, then player~$i$'s expected payoff is
\begin{align*}
& \qquad \int_{A_{t(-i)}(h_{t-1})} g_{ti}(h_{t-1}, \tilde{x}_{ti}, x_{t(-i)}) f_{t(-i)}(\rmd x_{t(-i)} |h_{t-1})  \\
& = \int_{\prod_{j \neq i} B^c_{tj}(h_{t-1})} g_{ti}(h_{t-1}, \tilde{x}_{ti}, x_{t(-i)}) f_{t(-i)}(\rmd x_{t(-i)} |h_{t-1})  \\
& = \int_{\prod_{j \neq i} B^c_{tj}(h_{t-1})} g_i(h_{t-1}, \tilde{x}_{ti}, x_{t(-i)}) f_{t(-i)}(\rmd x_{t(-i)} |h_{t-1})  \\
& = \int_{A_{t(-i)}(h_{t-1})} g_i(h_{t-1}, \tilde{x}_{ti}, x_{t(-i)}) f_{t(-i)}(\rmd x_{t(-i)} |h_{t-1})  \\
& \le \int_{A_{t}(h_{t-1})} p_{ti}(h_{t-1}, x_{t}) f_{t}(\rmd x_{t} |h_{t-1})  \\
& = \int_{A_{t}(h_{t-1})} g_{ti}(h_{t-1}, x_{t}) f_{t}(\rmd x_{t} |h_{t-1}).
\end{align*}
The first and the third equalities hold since
$$f_{t(-i)}\left(\prod_{j \neq i} B^c_{tj}(h_{t-1})|h_{t-1} \right) = f_{t(-i)}(A_{t(-i)}(h_{t-1})|h_{t-1}).$$
The second equality is due to the fact that $g_{ti}(h_{t-1}, \tilde{x}_{ti}, x_{t(-i)})  = g_i(h_{t-1}, \tilde{x}_{ti}, x_{t(-i)})$ for $x_{t(-i)} \in \prod_{j \neq i} B^c_{tj}(h_{t-1})$. The first inequality follows from the definition of $B_{ti}$, and the fourth equality holds since $g_t(h_{t-1}, x_t) = p_t(h_{t-1}, x_t)$ for $f_{t}(\cdot|h_{t-1})$-almost all $x_{t}$.

Thus, player~$i$ cannot improve his payoff in the subgame $h_t$ by a unilateral change in his strategy for any $i \in I$, which implies that $f_t(\cdot|h_{t-1})$ is a Nash equilibrium given the payoff $g_t(h_{t-1}, \cdot)$ in the subgame $h_{t-1}$.

\

Step~3.
%Define a correspondence $\tilde{Q}_{t+1}$ from $\mbox{Gr}(A_{t}) \times S_t$ to $\bR^n$ as follows: for any $(h_{t-1}, x_t, s_{t}) \in \mbox{Gr}(A_{t})\times S_t$,
%$$\tilde{Q}_{t+1}(h_{t-1},x_t,s_{t}) =
%\begin{cases}
%Q_{t+1}(h_{t-1},x_t,s_{t}), & (h_{t-1},x_t,s_{t}) \in H_t; \\
%\{0\}, & \mbox{otherwise}.
%\end{cases}$$ Then $\tilde{Q}_{t+1}$ is measurable, nonempty and compact valued.
For any $(h_{t-1}, x_{t}) \in \mbox{Gr}(A_{t})$,
$$ P_{t}(h_{t-1}, x_{t})  = \int_{S_{t}} Q_{t+1}(h_{t-1}, x_{t}, s_{t}) f_{t0}(\rmd s_{t}|h_{t-1}).
$$
By Lemma~\ref{lem-mertens}, there exists a Borel measurable mapping $q$ from $\mbox{Gr}(P_{t})\times S_t$ to $\bR^n$ such that
\begin{enumerate}
  \item $q(h_{t-1}, x_t, e, s_{t}) \in Q_{t+1}(h_{t-1},x_t,s_{t})$ for any $(h_{t-1}, x_t,e,  s_{t}) \in \mbox{Gr}(P_{t})\times S_t$;
  \item $e = \int_{S_{t}} q(h_{t-1}, x_t,e,  s_{t}) f_{t0}(\rmd s_{t}|h_{t-1})$ for any $(h_{t-1}, x_t,e) \in \mbox{Gr}(P_{t})$, where $(h_{t-1}, x_{t}) \in \mbox{Gr}(A_{t})$.
\end{enumerate}
Let
$$q_{t+1}(h_{t-1},x_t,s_{t}) = q(h_{t-1},x_t, g_{t}(h_{t-1},x_t), s_{t})$$
for any $(h_{t-1}, x_t, s_{t}) \in H_t$. Then $q_{t+1}$ is a Borel measurable selection of $Q_{t+1}$.

For $(h_{t-1}, x_{t}) \in \mbox{Gr}(A_{t})$,
\begin{align*}
g_{t}(h_{t-1},x_t)
& = \int_{S_{t}} q(h_{t-1}, x_t,g_{t}(h_{t-1},x_t), s_{t}) f_{t0}(\rmd s_{t}|h_{t-1}) \\
& = \int_{S_{t}} q_{t+1}(h_{t-1},x_t,s_{t}) f_{t0}(\rmd s_{t}|h_{t-1}).
\end{align*}

Therefore, we have a Borel measurable selection $q_{t+1}$ of $Q_{t+1}$, and a Borel measurable mapping $f_{t} \colon H_{t-1}\to \otimes_{i\in I}\cM(X_{ti})$ such that for all $h_{t-1} \in \tilde{H}_{t-1}$, properties~(1)-(3) are satisfied. The proof is complete.
\end{proof}

If a dynamic game has only $T$ stages for some positive integer $T \ge 1$, then let $Q_{T+1}(h_T) = \{u(h_T)\}$ for any $h_T \in H_T$, and $Q_t = \Phi(Q_{t+1})$ for $1 \le t \le T-1$. We can start with the backward induction from the last period and stop at the initial period, then run the forward induction from the initial period to the last period. Thus, the following corollary is immediate.

\begin{coro}\label{coro-finite horizon}
Any finite-horizon dynamic game with the ARM condition has a subgame-perfect equilibrium.
\end{coro}

\subsubsection{Infinite horizon case}\label{subsubsec-proof general}

Pick a sequence $\xi = (\xi_1, \xi_2, \ldots)$ such that (1) $\xi_{m}$ is a transition probability from $H_{m-1}$ to $\cM(X_{m})$ for any $m \ge 1$, and (2) $\xi_{m}(A_{m}(h_{m-1})|h_{m-1}) = 1$ for any $m \ge 1$ and $h_{m-1} \in H_{m-1}$. Denote the set of all such $\xi$ as $\Upsilon$.

Fix any $t\ge 1$, define correspondences $\Xi_t^t$ and $\Delta_t^t$ as follows: in the subgame $h_{t-1}$,
$$\Xi_t^t(h_{t-1}) = \cM(A_{t}(h_{t-1})) \otimes \lambda_{t},$$
and
$$\Delta_t^t(h_{t-1}) = \cM(A_{t}(h_{t-1})) \otimes f_{t0}(h_{t-1}).$$
For any $m_1 > t$, suppose that the correspondences $\Xi_t^{m_1 - 1}$ and $\Delta_t^{m_1 - 1}$ have been defined. Then we can define correspondences $\Xi_t^{m_1} \colon H_{t-1}\to \cM\left(\prod_{t \le m \le m_1}(X_m \times S_m) \right)$ and $\Delta_t^{m_1} \colon H_{t-1}\to \cM\left(\prod_{t \le m \le m_1}(X_m \times S_m) \right)$ as follows:
\begin{align*}
\Xi_t^{m_1}(h_{t-1}) =
& \{ g(h_{t-1})\diamond (\xi_{m_1}(h_{t-1}, \cdot)\otimes \lambda_{m_1}) \colon  \\
& g \mbox{ is a Borel  measurable selection of } \Xi_t^{m_1-1}, \\
& \xi_{m_1} \mbox{ is a Borel measurable selection of } \cM(A_{m_1})\},
\end{align*}
and
\begin{align*}
\Delta_t^{m_1}(h_{t-1}) =
& \{ g(h_{t-1})\diamond (\xi_{m_1}(h_{t-1}, \cdot)\otimes f_{m_10}(h_{t-1}, \cdot)) \colon  \\
& g \mbox{ is a Borel  measurable selection of } \Delta_t^{m_1-1}, \\
& \xi_{m_1} \mbox{ is a Borel measurable selection of } \cM(A_{m_1})\},
\end{align*}
where $\cM(A_{m_1})$ is regarded as a correspondence from $H_{m_1-1}$ to the set of Borel probabilities on $X_{m_1}$. For any $m_1 \ge t$, let $\rho^{m_1}_{(h_{t-1}, \xi)} \in \Xi_t^{m_1}$ be the probability on $\prod_{t \le m \le m_1}(X_m \times S_m)$ which is induced by $\{\lambda_{m}\}_{t \le m \le m_1}$ and $\{\xi_m\}_{t \le m \le m_1}$, and $\varrho^{m_1}_{(h_{t-1}, \xi)} \in \Delta_t^{m_1}$ be the probability on $\prod_{t \le m \le m_1}(X_m \times S_m)$ which is induced by $\{f_{m0}\}_{t \le m \le m_1}$ and $\{\xi_m\}_{t \le m \le m_1}$. Then $\Xi^{m_1}_t(h_{t-1})$ is the set of all such $\rho^{m_1}_{(h_{t-1}, \xi)}$, and $\Delta^{m_1}_t(h_{t-1})$ is the set of all such $\varrho^{m_1}_{(h_{t-1}, \xi)}$. Notice that $\varrho^{m_1}_{(h_{t-1}, \xi)} \in \Delta^{m_1}_t(h_{t-1})$ if and only if $\rho^{m_1}_{(h_{t-1}, \xi)} \in \Xi^{m_1}_t(h_{t-1})$, and $\varrho^{m_1}_{(h_{t-1}, \xi)}$ and $\rho^{m_1}_{(h_{t-1}, \xi)}$ can be both regarded as probability measures on $H_{m_1}(h_{t-1})$.

Similarly, let $\rho_{(h_{t-1}, \xi)}$ be the probability on $\prod_{m \ge t}(X_m \times S_m)$ induced by $\{\lambda_{m}\}_{m \ge t}$ and $\{\xi_m\}_{m \ge t}$, and $\varrho_{(h_{t-1}, \xi)}$ the probability on $\prod_{m \ge t}(X_m \times S_m)$ induced by $\{f_{m0}\}_{m \ge t}$ and $\{\xi_m\}_{m \ge t}$. Denote the correspondence
$$\Xi_t \colon H_{t-1} \to \cM(\prod_{ m \ge t}(X_m \times S_m))$$
as the set of all such $\rho_{(h_{t-1}, \xi)}$, and
$$\Delta_t \colon H_{t-1} \to \cM(\prod_{ m \ge t}(X_m \times S_m))$$
as the set of all such $\varrho_{(h_{t-1}, \xi)}$.

\begin{lem}\label{lem-strategy}
For any $m_1 \ge t$ and $h_{t-1} \in H_{t-1}$,
$$\varrho^{m_1}_{(h_{t-1}, \xi)} = \left(\prod_{t \le m \le m_1} \varphi_{m0}(h_{t-1}, \cdot) \right) \circ \rho^{m_1}_{(h_{t-1}, \xi)}.\footnote{For $m \ge t \ge 1$ and $h_{t-1} \in H_{t-1}$, the function $\varphi_{m0}(h_{t-1}, \cdot)$ is defined on $H_{m-1}(h_{t-1})\times S_m$, which is measurable and sectionally continuous on $\prod_{t \le k \le m-1} X_k$. By Lemma~\ref{lem-lusin}, $\varphi_{m0}(h_{t-1}, \cdot)$ can be extended to be a measurable function $\acute{\varphi}_{m0}(h_{t-1},\cdot)$ on the product space $\prod_{t \le k \le m-1} X_k \times \prod_{t \le k \le m} S_k$, which is also sectionally continuous on $\prod_{t \le k \le m-1} X_k$. Given any $\xi \in \Upsilon$, since $\rho^{m}_{(h_{t-1}, \xi)}$ concentrates on $H_{m}(h_{t-1})$,
$\varphi_{m0}(h_{t-1}, \cdot) \circ \rho^{m}_{(h_{t-1}, \xi)} = \acute{\varphi}_{m0}(h_{t-1}, \cdot)  \circ \rho^{m}_{(h_{t-1}, \xi)}$.
For notational simplicity, we still use $\varphi_{m0}(h_{t-1}, \cdot)$, instead of $\acute{\varphi}_{m0}(h_{t-1},\cdot)$, to denote the above extension. Similarly, we can work with a suitable extension of the payoff function $u$ as needed.}
$$
\end{lem}

\begin{proof}
Fix $\xi \in \Upsilon$, and Borel subsets $C_m \subseteq X_m$ and $D_m \subseteq S_m$ for $m \ge t$. First, we have
\begin{align*}
& \qquad \varrho^{t}_{(h_{t-1}, \xi)}(C_t \times D_t) \\
& = \xi_t(C_t|h_{t-1}) \cdot f_{t0} (D_t|h_{t-1}) \\
& = \int_{X_t\times S_t} \delta_{C_t\times D_t}(x_t,s_t)\varphi_{t0}(h_{t-1}, s_t) (\xi_t(h_{t-1})\otimes \lambda_t) (\rmd (x_t,s_t)),
\end{align*}
which implies that $\varrho^{t}_{(h_{t-1}, \xi)} = \varphi_{t0}(h_{t-1}, \cdot) \circ \rho^{t}_{(h_{t-1}, \xi)}$.

Suppose that $\varrho^{m_2}_{(h_{t-1}, \xi)} = \left(\prod_{t \le m \le m_2} \varphi_{m0}(h_{t-1}, \cdot) \right) \circ \rho^{m_2}_{(h_{t-1}, \xi)}$ for some $m_2 \ge t$. Then
\begin{align*}
& \quad \varrho^{m_2+1}_{(h_{t-1}, \xi)} \left( \prod_{t \le m \le m_2 + 1} (C_m \times D_m) \right) \\
& = \varrho^{m_2}_{(h_{t-1}, \xi)} \diamond (\xi_{m_2 + 1}(h_{t-1}, \cdot) \otimes f_{(m_2+1)0} (h_{t-1}, \cdot)) \left( \prod_{t \le m \le m_2 + 1} (C_m \times D_m) \right) \\
& = \int_{\prod_{t\le m \le m_2} (X_m \times S_m)} \int_{X_{m_2+1} \times S_{m_2+1} } \delta_{\prod_{t \le m \le m_2 + 1} (C_m \times D_m)} (x_t, \ldots, x_{m_2+1}, s_t, \ldots, s_{m_2+1}) \cdot \\
& \quad \xi_{m_2 + 1} \otimes f_{(m_2+1)0} (\rmd (x_{m_2+1}, s_{m_2+1}) | h_{t-1}, x_t, \ldots, x_{m_2}, s_t, \ldots, s_{m_2}) \\
& \quad \varrho^{m_2}_{(h_{t-1}, \xi)}(\rmd (x_t, \ldots, x_{m_2}, s_t, \ldots, s_{m_2}) | h_{t-1}) \\
& = \int_{\prod_{t\le m \le m_2} (X_m \times S_m)} \int_{S_{m_2+1}}\int_{X_{m_2+1} } \delta_{\prod_{t \le m \le m_2 + 1} (C_m \times D_m)} (x_t, \ldots, x_{m_2+1}, s_t, \ldots, s_{m_2+1}) \cdot \\
& \quad \varphi_{(m_2+1)0}(h_{t-1}, x_t, \ldots, x_{m_2}, s_t, \ldots, s_{m_2+1}) \xi_{m_2 + 1} (\rmd x_{m_2+1} | h_{t-1}, x_t, \ldots, x_{m_2}, s_t, \ldots, s_{m_2})\\
& \quad  \lambda_{(m_2+1)0}(\rmd s_{m_2+1}) \prod_{t \le m \le m_2} \varphi_{m0}(h_{t-1}, x_t, \ldots, x_{m-1}, s_t, \ldots, s_{m}) \\
& \quad \rho^{m_2}_{(h_{t-1}, \xi)}(\rmd (x_t, \ldots, x_{m_2}, s_t, \ldots, s_{m_2}) | h_{t-1}) \\
& = \int_{\prod_{t\le m \le m_2+1} (X_m \times S_m)} \delta_{\prod_{t \le m \le m_2 + 1} (C_m \times D_m)} (x_t, \ldots, x_{m_2+1}, s_t, \ldots, s_{m_2+1}) \cdot \\
& \quad \prod_{t \le m \le m_2 +1} \varphi_{m0}(h_{t-1}, x_t, \ldots, x_{m-1}, s_t, \ldots, s_{m}) \rho^{m_2+1}_{(h_{t-1}, \xi)}(\rmd (x_t, \ldots, x_{m_2}, s_t, \ldots, s_{m_2}) | h_{t-1}),
\end{align*}
which implies that
$$\varrho^{m_2+1}_{(h_{t-1}, \xi)} = \left(\prod_{t \le m \le m_2+1} \varphi_{m0}(h_{t-1}, \cdot) \right) \circ \rho^{m_2+1}_{(h_{t-1}, \xi)}.$$
The proof is complete.
\end{proof}

\begin{lem}\label{lem-measurable history}
\begin{enumerate}
  \item For any $t \ge 1$, the correspondence $\Delta_t^{m_1}$ is nonempty and compact valued, and sectionally continuous on $X^{t-1}$ for any $m_1 \ge t$.
  \item For any $t \ge 1$, the correspondence $\Delta_t$ is nonempty and compact valued, and sectionally continuous on $X^{t-1}$.
\end{enumerate}
\end{lem}

\begin{proof}
(1) We first show that the correspondence $\Xi_t^{m_1}$ is nonempty and compact valued, and sectionally continuous on $X^{t-1}$ for any $m_1 \ge t$

Consider the case $m_1 = t \ge 1$, where
$$\Xi_t^t(h_{t-1}) = \cM(A_{t}(h_{t-1})) \otimes \lambda_{t}.$$
Since $A_{ti}$ is nonempty and compact valued, and sectionally continuous on $X^{t-1}$, $\Xi_t^t$ is nonempty and compact valued, and sectionally continuous on $X^{t-1}$.

Now suppose that $\Xi_t^{m_2}$ is nonempty and compact valued, and sectionally continuous on $X^{t-1}$ for some $m_2 \ge t\ge 1$. Notice that
\begin{align*}
\Xi_t^{m_2+1}(h_{t-1}) =
& \{ g(h_{t-1})\diamond (\xi_{m_2+1}(h_{t-1}, \cdot)\otimes \lambda_{(m_2+1)}) \colon  \\
& g \mbox{ is a Borel  measurable selection of } \Xi_t^{m_2}, \\
& \xi_{m_2+1} \mbox{ is a Borel measurable selection of } \cM(A_{m_2+1})\}.
\end{align*}

Define a correspondence $A_t^t$ from $H_{t-1}\times S_t$ to $X_t$ as $A_t^t(h_{t-1}, s_t) = A_t(h_{t-1})$. Then $A_t^t$ is nonempty and compact valued, sectionally continuous on $X_{t-1}$, and has a $\cB(X^{t} \times S^{t})$-measurable graph. For any $(s_1, \ldots, s_{t})$, since $H_{t-1}(s_1, \ldots, s_{t-1})$ is compact and $A_{t}(\cdot, s_1, \ldots, s_{t-1})$ is continuous and compact valued, $A_t^t(\cdot, s_1, \ldots, s_{t})$ has a compact graph by Lemma~\ref{lem-measurable selection}~(6). For any $h_{t-1}\in H_{t-1}$ and $\tau \in \Xi_t^t(h_{t-1})$, the marginal of $\tau$ on $S_t$ is $\lambda_t$ and $\tau (\mbox{Gr}(A_t^t(h_{t-1}, \cdot))) = 1$.

For any $m_1 > t$, suppose that the correspondence
$$A_t^{m_1 - 1} \colon H_{t-1}\times \prod_{t \le m \le m_1-1} S_m \to \prod_{t \le m \le m_1-1} X_m$$
has been defined such that
\begin{enumerate}
  \item it is nonempty and compact valued, sectionally upper hemicontinuous on $X_{t-1}$, and has a $\cB(X^{m_1-1} \times S^{m_1-1})$-measurable graph;
  \item for any $(s_1, \ldots. s_{m_1-1})$, $A_t^{m_1 - 1}(\cdot, s_1, \ldots. s_{m_1-1})$ has a compact graph;
  \item for any $h_{t-1}\in H_{t-1}$ and $\tau \in \Xi_t^{m_1-1}(h_{t-1})$, the marginal of $\tau$ on $\prod_{t \le m \le m_1-1} S_m$ is $\otimes_{t \le m \le m_1-1} \lambda_m$ and $\tau (\mbox{Gr}(A_t^{m_1 - 1}(h_{t-1}, \cdot))) = 1$.
\end{enumerate}

We define a correspondence $A_t^{m_1} \colon H_{t-1}\times \prod_{t \le m \le m_1} S_m \to \prod_{t \le m \le m_1} X_m$ as follows:
\begin{align*}
A_t^{m_1}(h_{t-1}, s_t, \ldots, s_{m_1})  =
& \{(x_t, \ldots, x_{m_1}) \colon \\
& x_{m_1} \in A_{m_1}(h_{t-1}, x_t, \ldots, x_{m_1-1}, s_t, \ldots, s_{m_1-1}), \\
& (x_t, \ldots, x_{m_1-1}) \in A_t^{m_1 - 1}(h_{t-1}, s_t, \ldots, s_{m_1-1}) \}.
\end{align*}
It is obvious that $A_t^{m_1}$ is nonempty valued. For any $(s_1, \ldots, s_{m_1})$, since $A_t^{m_1 - 1}(\cdot, s_1, \ldots. s_{m_1-1})$ has a compact graph and $A_{m_1}(\cdot, s_1, \ldots, s_{m_1-1})$ is continuous and compact valued, $A_t^{m_1}(\cdot, s_1, \ldots. s_{m_1})$ has a compact graph by Lemma~\ref{lem-measurable selection}~(6), which implies that $A_t^{m_1}$ is compact valued and sectionally upper hemicontinuous on $X_{t-1}$. In addition, $\mbox{Gr}(A_t^{m_1}) = \mbox{Gr}(A_{m_1}) \times S_{m_1}$, which is $\cB(X^{m_1} \times S^{m_1})$-measurable. For any $h_{t-1}\in H_{t-1}$ and $\tau \in \Xi_t^{m_1}(h_{t-1})$, it is obvious that the marginal of $\tau$ on $\prod_{t \le m \le m_1} S_m$ is $\otimes_{t \le m \le m_1} \lambda_m$ and $\tau (\mbox{Gr}(A_t^{m_1}(h_{t-1}, \cdot))) = 1$.

By Lemma~\ref{lem-continuous composition}, $\Xi_t^{m_2+1}$ is nonempty and compact valued, and sectionally continuous on $X^{t-1}$.

\

Now we show that the correspondence $\Delta_t^{m_1}$ is nonempty and compact valued, and sectionally continuous on $X^{t-1}$ for any $m_1 \ge t$.

Given $s^{t-1}$ and a sequence $\{x_0^k,x_1^k, \ldots, x_{t-1}^k\} \in H_{t-1}(s^{t-1})$ for $1\le k \le \infty$. Let $h^k_{t-1} = (s^{t-1}, (x_0^k,x_1^k, \ldots, x_{t-1}^k))$. It is obvious that $\Delta_t^{m_1}$ is nonempty valued, we first show that $\Delta_t^{m_1}$ is sectionally upper hemicontinuous on $X^{t-1}$.  Suppose that $\varrho^{m_1}_{(h^k_{t-1}, \xi^k)} \in \Delta_t^{m_1}(h_{t-1}^k)$ for $1 \le k < \infty$ and $(x_0^k,x_1^k, \ldots, x_{t-1}^k) \to (x_0^\infty,x_1^\infty, \ldots, x_{t-1}^\infty)$, we need to show that there exists some $\xi^\infty$ such that a subsequence of $\varrho^{m_1}_{(h^k_{t-1}, \xi^k)}$ weakly converges to $\varrho^{m_1}_{(h^\infty_{t-1}, \xi^\infty)}$ and $\varrho^{m_1}_{(h^\infty_{t-1}, \xi^\infty)} \in \Delta_t^{m_1}(h_{t-1}^\infty)$.

Since $\Xi_t^{m_1}$ is sectionally upper hemicontinuous on $X^{t-1}$, there exists some $\xi^\infty$ such that a subsequence of $\rho^{m_1}_{(h^k_{t-1}, \xi^k)}$, say itself, weakly converges to $\rho^{m_1}_{(h^\infty_{t-1}, \xi^\infty)}$ and $\rho^{m_1}_{(h^\infty_{t-1}, \xi^\infty)} \in \Xi_t^{m_1}(h_{t-1}^\infty)$. Then $\varrho^{m_1}_{(h^\infty_{t-1}, \xi^\infty)} \in \Delta_t^{m_1}(h_{t-1}^\infty)$.

For any bounded continuous function $\psi$ on $\prod_{t \le m \le m_1} (X_m\times S_m)$, let
$$\chi_k(x_t, \ldots, x_{m_1}, s_t, \ldots, s_{m_1}) = $$
$$\psi(x_t, \ldots, x_{m_1}, s_t, \ldots, s_{m_1}) \cdot \prod_{t \le m \le m_1} \varphi_{m0}(h^k_{t-1}, x_t, \ldots, x_{m-1}, s_t, \ldots, s_{m}).$$
Then $\{\chi_k\}$ is a sequence of functions satisfying the following three properties.
\begin{enumerate}

\item For each $k$, $\chi_k$ is jointly measurable and sectionally continuous on $X$.

\item For any $(s_t, \ldots, s_{m_1})$ and any sequence $(x_t^k, \ldots, x_{m_1}^k) \to (x_t^\infty, \ldots, x^\infty_{m_1})$ in $X$, $\chi_k(x_t^k, \ldots, x_{m_1}^k, s_t, \ldots, s_{m_1}) \to \chi_\infty(x_t^\infty, \ldots, x_{m_1}^\infty, s_t, \ldots, s_{m_1})$ as $k \to \infty$.

\item The sequence $\{\chi_k\}_{1 \le k \le \infty}$ is integrably bounded.
\end{enumerate}
By Lemma~\ref{lem-topology}, as $k \to \infty$,
$$\int_{\prod_{t\le m \le m_1}(X_m \times S_m)} \chi_k(x_t, \ldots, x_{m_1}, s_t, \ldots, s_{m_1}) \rho^{m_1}_{(h^k_{t-1}, \xi^k)} (\rmd (x_t, \ldots, x_{m_1}, s_t, \ldots, s_{m_1}))
$$
$$\to \int_{\prod_{t\le m \le m_1}(X_m \times S_m)} \chi_\infty(x_t, \ldots, x_{m_1}, s_t, \ldots, s_{m_1}) \rho^{m_1}_{(h^\infty_{t-1}, \xi^\infty)} (\rmd (x_t, \ldots, x_{m_1}, s_t, \ldots, s_{m_1})).
$$
Then by Lemma~\ref{lem-strategy},
$$\int_{\prod_{t\le m \le m_1}(X_m \times S_m)} \psi(x_t, \ldots, x_{m_1}, s_t, \ldots, s_{m_1}) \varrho^{m_1}_{(h^k_{t-1}, \xi^k)} (\rmd (x_t, \ldots, x_{m_1}, s_t, \ldots, s_{m_1}))
$$
$$\to \int_{\prod_{t\le m \le m_1}(X_m \times S_m)} \psi(x_t, \ldots, x_{m_1}, s_t, \ldots, s_{m_1}) \varrho^{m_1}_{(h^\infty_{t-1}, \xi^\infty)} (\rmd (x_t, \ldots, x_{m_1}, s_t, \ldots, s_{m_1})),
$$
which implies that $\varrho^{m_1}_{(h^k_{t-1}, \xi^k)}$ weakly converges to $\varrho^{m_1}_{(h^\infty_{t-1}, \xi^\infty)}$. Therefore, $\Delta_t^{m_1}$ is sectionally upper hemicontinuous on $X^{t-1}$. If one chooses $h_{t-1}^1 = h_{t-1}^2 = \cdots = h_{t-1}^\infty$, then we indeed show that $\Delta_t^{m_1}$ is compact valued.

In the argument above, we indeed proved that if $\rho^{m_1}_{(h^k_{t-1}, \xi^k)}$ weakly converges to $\rho^{m_1}_{(h^\infty_{t-1}, \xi^\infty)}$, then $\varrho^{m_1}_{(h^k_{t-1}, \xi^k)}$ weakly converges to $\varrho^{m_1}_{(h^\infty_{t-1}, \xi^\infty)}$.

The left is to show that $\Delta_t^{m_1}$ is sectionally lower hemicontinuous on $X^{t-1}$. Suppose that $(x_0^k,x_1^k, \ldots, x_{t-1}^k) \to (x_0^\infty,x_1^\infty, \ldots, x_{t-1}^\infty)$ and $\varrho^{m_1}_{(h^\infty_{t-1}, \xi^\infty)} \in \Delta_t^{m_1}(h_{t-1}^\infty)$, we need to show that there exists a subsequence $\{(x_0^{k_m},x_1^{k_m}, \ldots, x_{t-1}^{k_m})\}$ of $\{(x_0^k,x_1^k, \ldots, x_{t-1}^k)\}$ and $\varrho^{m_1}_{(h^{k_m}_{t-1}, \xi^{k_m})} \in \Delta_t^{m_1}(h_{t-1}^{k_m})$ for each $k_m$ such that $\varrho^{m_1}_{(h^{k_m}_{t-1}, \xi^{k_m})}$ weakly converges to $\varrho^{m_1}_{(h^\infty_{t-1}, \xi^\infty)}$.

Since $\varrho^{m_1}_{(h^\infty_{t-1}, \xi^\infty)} \in \Delta_t^{m_1}(h_{t-1}^\infty)$, we have $\rho^{m_1}_{(h^\infty_{t-1}, \xi^\infty)} \in \Xi_t^{m_1}(h_{t-1}^\infty)$. Because $\Xi_t^{m_1}$ is sectionally lower hemicontinuous on $X^{t-1}$, there exists a subsequence of $\{(x_0^k,x_1^k, \ldots, x_{t-1}^k)\}$, say itself, and $\rho^{m_1}_{(h^{k}_{t-1}, \xi^{k})} \in \Xi_t^{m_1}(h_{t-1}^{k})$ for each $k$ such that $\rho^{m_1}_{(h^{k}_{t-1}, \xi^{k})}$ weakly converges to $\rho^{m_1}_{(h^\infty_{t-1}, \xi^\infty)}$. As a result, $\varrho^{m_1}_{(h^{k}_{t-1}, \xi^{k})}$ weakly converges to $\varrho^{m_1}_{(h^\infty_{t-1}, \xi^\infty)}$, which implies that $\Delta_t^{m_1}$ is sectionally lower hemicontinuous on $X^{t-1}$.

Therefore, $\Delta_t^{m_1}$ is nonempty and compact valued, and sectionally continuous on $X^{t-1}$ for any $m_1 \ge t$.

\

(2) We show that $\Delta_t$ is nonempty and compact valued, and sectionally continuous on $X^{t-1}$.

It is obvious that $\Delta_t$ is nonempty valued, we first prove that it is compact valued.

Given $h_{t-1}$ and a sequence $\{\tau^k\}\subseteq \Delta_t(h_{t-1})$,  there exists a sequence of $\{\xi^k\}_{k \ge 1}$ such that $\xi^k = (\xi^k_1, \xi^k_2, \ldots) \in \Upsilon$ and $\tau^k = \varrho_{(h_{t-1}, \xi^k)}$ for each $k$.

By (1), $\Xi_t^t$ is compact. Then there exists a measurable mapping $g_t$ such that (1) $g^t = (\xi^1_1, \ldots, \xi^1_{t-1}, g_t, \xi^1_{t+1}, \ldots) \in \Upsilon$, and (2) $\rho^{t}_{(h_{t-1}, \xi^k)}$ weakly converges to $\rho^{t}_{(h_{t-1}, g^{t})}$. Note that $\{\xi_{t+1}^k\}$ is a Borel measurable selection of $\cM(A_{t+1})$. By Lemma~\ref{lem-continuous composition}, there is a Borel measurable selection $g_{t+1}$ of $\cM(A_{t+1})$ such that there is a subsequence of $\{\rho^{t+1}_{(h_{t-1}, \xi^k)}\}$, say itself, which weakly converges to $\rho^{t+1}_{(h_{t-1}, g^{t+1})}$, where $g^{t+1} = (\xi^1_1, \ldots, \xi^1_{t-1}, g_t, g_{t+1}, \xi^1_{t+2}, \ldots) \in \Upsilon$.

Repeat this procedure, one can construct a Borel measurable mapping $g$ such that $\rho_{(h_{t-1}, g)}$ is a convergent point of $\{\rho_{(h_{t-1}, \xi^k)}\}$. Thus, $\varrho_{(h_{t-1}, g)}$ is a convergent point of $\{\varrho_{(h_{t-1}, \xi^k)}\}$.

The sectional upper hemicontinuity of $\Delta_t$ follows a similar argument as above. In particular, given $s^{t-1}$ and a sequence $\{x_0^k,x_1^k, \ldots, x_{t-1}^k\}\subseteq H_{t-1}(s^{t-1})$ for $k \ge 0$. Let $h^k_{t-1} = (s^{t-1}, (x_0^k,x_1^k, \ldots, x_{t-1}^k))$. Suppose that $(x_0^k,x_1^k, \ldots, x_{t-1}^k) \to (x_0^0,x_1^0, \ldots, x_{t-1}^0)$. If $\{\tau^k\}\subseteq \Delta_t(h_{t-1}^k)$ for $k \ge 1$ and $\tau^k \to \tau^0$, then one can show that $\tau^0 \in \Delta_t(h^0_{t-1})$ by repeating a similar argument as in the proof above.

Finally, we consider the sectional lower hemicontinuity of $\Delta_t$. Suppose that $\tau^0 \in \Delta_t(h^0_{t-1})$. Then there exists some $\xi \in \Upsilon$ such that $\tau^0 = \varrho_{(h^0_{t-1}, \xi)}$. Denote $\tilde{\tau}^m = \varrho^m_{(h^0_{t-1}, \xi)} \in \Delta^m_t(h_{t-1}^0)$ for $m \ge t$. As $\Delta^m_t$ is continuous, for each $m$, there exists some $\xi^m \in \Upsilon$ such that $d(\varrho^m_{(h^{k_m}_{t-1}, \xi^m)}, \tilde{\tau}^m) \le \frac{1}{m}$ for $k_m$ sufficiently large, where $d$ is the Prokhorov metric. Let $\tau^m = \varrho_{(h^{k_m}_{t-1}, \xi^m)}$. Then $\tau^m$ weakly converges to $\tau^0$, which implies that $\Delta_t$ is sectionally lower hemicontinuous.
\end{proof}

%For $h_{\infty} = (x_0, s_0, x_1, s_1, \ldots)$, let
%$$u'_t(h_{\infty}) = u(h_{\infty})\cdot\prod_{m \ge t+1} \varphi_m(x_0, s_0, \ldots, x_{m-1}, s_{m-1}, s_m).$$
%Then $u'_t$ is a $\lambda^\infty$-integrably bounded function which is sectionally continuous on $X^\infty$.
Define a correspondence $Q^\tau_t \colon H_{t-1} \to \bR^n_{++}$ as follows:
$$Q^\tau_t(h_{t-1}) =$$
$$\begin{cases}
\{\int_{\prod_{m \ge t} (X_m \times S_m)} u(h_{t-1}, x, s) \varrho_{(h_{t-1}, \xi)}(\rmd (x,s)) \colon  \varrho_{(h_{t-1}, \xi)} \in \Delta_t(h_{t-1}) \}; & t > \tau; \\
\Phi (Q^\tau_{t+1})(h_{t-1}) & t \le \tau.
\end{cases}
$$
Denote $Q_t^\infty = \cap_{\tau \ge 1} Q_t^\tau$.

\begin{lem}\label{lem-continuation set}
For any $t, \tau \ge 1$, $Q^\tau_{t}$ is bounded, measurable, nonempty and compact valued, and essentially sectionally upper hemicontinuous  on $X^{t-1}$.
\end{lem}

\begin{proof}
We prove the lemma  in three steps.

Step 1. Fix $t > \tau$. We will show that $Q^\tau_{t}$ is bounded, nonempty and compact valued, and sectionally upper hemicontinuous on $X^{t-1}$.

The boundedness and nonemptiness of $Q^\tau_{t}$ are obvious. We shall prove that $Q^\tau_{t}$ is sectionally upper hemicontinuous on $X^{t-1}$. Given $s^{t-1}$ and a sequence $\{x_0^k,x_1^k, \ldots, x_{t-1}^k\}\subseteq H_{t-1}(s^{t-1})$ for $k \ge 0$. Let $h^k_{t-1} = (s^{t-1}, (x_0^k,x_1^k, \ldots, x_{t-1}^k))$. Suppose that $a^k \in Q_t^\tau(h_{t-1}^k)$ for $k \ge 1$, $(x_0^k,x_1^k, \ldots, x_{t-1}^k) \to (x_0^0,x_1^0, \ldots, x_{t-1}^0)$ and $a^k \to a^0$, we need to show that $a^0 \in Q_t^\tau(h_{t-1}^0)$.

By the definition, there exists a sequence $\{\xi^k\}_{k \ge 1}$ such that
$$a^k = \int_{\prod_{m \ge t} (X_m \times S_m)} u(h^k_{t-1}, x, s) \varrho_{(h^k_{t-1}, \xi^k)}(\rmd (x,s)),
$$
where $\xi^k = (\xi^k_1,\xi^k_2, \ldots) \in \Upsilon$ for each $k$. As $\Delta_t$ is compact valued and sectionally continuous on $X^{t-1}$, there exist some $\varrho_{(h^0_{t-1}, \xi^0)} \in \Delta_t(h^0_{t-1})$ and a subsequence of $\varrho_{(h^k_{t-1}, \xi^k)}$, say itself, which weakly converges to $\varrho_{(h^0_{t-1}, \xi^0)}$ for $\xi^0 = (\xi^0_1, \xi^0_2, \ldots) \in \Upsilon$.

We shall show that
$$a^0 = \int_{\prod_{m \ge t} (X_m \times S_m)} u(h^0_{t-1}, x, s) \varrho_{(h^0_{t-1}, \xi^0)}(\rmd (x,s)).
$$
For this aim, we only need to show that for any $\delta > 0$,
\begin{equation}\label{equa-a}
\left|a^0 - \int_{\prod_{m \ge t} (X_m \times S_m)} u(h^0_{t-1}, x, s) \varrho_{(h^0_{t-1}, \xi^0)}(\rmd (x,s)) \right| < \delta.
\end{equation}

Since the game is continuous at infinity, there exists a positive integer $M \ge t$ such that $w^m < \frac{1}{5}\delta$ for any $m > M$.

For each $j > M$, by Lemma~\ref{lem-lusin}, there exists a measurable selection $\xi'_{j}$ of $\cM(A_{j})$ such that $\xi'_{j}$ is sectionally continuous on $X^{j-1}$. Let $\mu \colon H_M \to \prod_{m > M}(X_m \times S_m)$ be the transition probability which is induced by $(\xi'_{M+1}, \xi'_{M+2}, \ldots)$ and $\{f_{(M+1)0}, f_{(M+2)0}, \ldots\}$. By Lemma~\ref{lem-measurable composition}, $\mu$ is measurable and sectionally continuous on $X^M$. Let
$$V_M(h_{t-1}, x_t, \ldots, x_{M}, s_t, \ldots, s_{M}) =$$
$$\int_{\prod_{m > M} (X_m \times S_m)} u(h_{t-1}, x_t, \ldots, x_{M}, s_t, \ldots, s_{M}, x, s) \rmd \mu (x,s|h_{t-1}, x_t, \ldots, x_{M}, s_t, \ldots, s_{M}).$$
Then $V_M$ is bounded and measurable. In addition, $V_M$ is sectionally continuous on $X^{M}$ by Lemma~\ref{lem-topology}.

For any $k \ge 0$, we have
\begin{align*}
& \quad \big|\int_{\prod_{m \ge t} (X_m \times S_m)} u(h^k_{t-1}, x, s) \varrho_{(h^k_{t-1}, \xi^k)}(\rmd (x,s)) \\
& - \int_{\prod_{t \le m \le M} (X_m \times S_m)} V_M(h^k_{t-1}, x_t, \ldots, x_{M}, s_t, \ldots, s_{M}) \varrho_{(h^k_{t-1}, \xi^k)}^M(\rmd (x_t, \ldots, x_{M}, s_t, \ldots, s_{M})) \big|   \\
& \le w^{M+1} \\
& < \frac{1}{5}\delta.
\end{align*}

Since $\varrho_{(h^k_{t-1}, \xi^k)}$ weakly converges to $\varrho_{(h^0_{t-1}, \xi^0)}$ and $\varrho^M_{(h^k_{t-1}, \xi^k)}$  is the marginal of $\varrho_{(h^k_{t-1}, \xi^k)}$  on $\prod_{t \le m \le M}(X_m\times S_m)$ for any $k \ge 0$, the sequence $\varrho^M_{(h^k_{t-1}, \xi^k)}$ also weakly converges to $\varrho^M_{(h^0_{t-1}, \xi^0)}$. By Lemma~\ref{lem-topology}, we have
\begin{align*}
& \quad |\int_{\prod_{t \le m \le M} (X_m \times S_m)}V_M(h^k_{t-1}, x_t, \ldots, x_{M}, s_t, \ldots, s_{M})  \varrho^M_{(h^k_{t-1}, \xi^k)}(\rmd (x_t, \ldots, x_{M}, s_t, \ldots, s_{M})) \\
& - \int_{\prod_{t \le m \le M} (X_m \times S_m)}V_M(h^0_{t-1}, x_t, \ldots, x_{M}, s_t, \ldots, s_{M}) \varrho^M_{(h^0_{t-1}, \xi^0)} (\rmd (x_t, \ldots, x_{M}, s_t, \ldots, s_{M})) |   \\
& < \frac{1}{5}\delta
\end{align*}
for $k \ge K_1$, where $K_1$ is a sufficiently large positive integer. In addition, there exists a positive integer $K_2$ such that $|a^k - a^0| < \frac{1}{5}\delta$ for $k \ge K_2$.  Combining the inequalities above, we prove inequality~(\ref{equa-a}), which implies that $Q^\tau_{t}$ is sectionally upper hemicontinuous on $X^{t-1}$ for $t > \tau$.

Furthermore, to prove that $Q^\tau_{t}$ is compact valued, we only need to consider the case that $\{x_0^k,x_1^k, \ldots, x_{t-1}^k\} = \{x_0^0,x_1^0, \ldots, x_{t-1}^0\}$ for any $k \ge 0$, and repeat the above proof.

\

Step 2. Fix $t > \tau$, we will show that $Q_{t}^\tau$ is measurable.

Fix a sequence $(\xi'_{1}, \xi'_{2}, \ldots)$, where $\xi'_{j}$ is a selection of $\cM(A_{j})$ measurable in $s^{j-1}$ and continuous in $x^{j-1}$ for each $j$. For any $M \ge t$, let
$$ W^M_M(h_{t-1}, x_t, \ldots, x_{M}, s_t, \ldots, s_{M}) = $$
$$\left\{\int_{\prod_{m > M} (X_m \times S_m)} u(h_{t-1}, x_t, \ldots, x_{M}, s_t, \ldots, s_{M}, x, s)  \varrho_{(h_{t-1}, x_t, \ldots, x_{M}, s_t, \ldots, s_{M}, \xi')} (\rmd(x,s)) \right\}.
$$
By Lemma~\ref{lem-measurable composition}, $\varrho_{(h_{t-1}, x_t, \ldots, x_{M}, s_t, \ldots, s_{M}, \xi')}$ is measurable from $H_M$ to $\cM\left(\prod_{m > M} (X_m \times S_m)\right)$, and sectionally continuous on $X^M$. Thus, $W^M_M$ is bounded, measurable, nonempty, convex and compact valued. By Lemma~\ref{lem-topology}, $W^M_M$ is sectionally continuous on $X^{M}$.

Suppose that  for some $t \le j \le M$, $W^j_M$ has been defined such that it is bounded, measurable, nonempty, convex and compact valued, and sectionally continuous on $X^{j}$. Let
\begin{align*}
& \quad W^{j-1}_M(h_{t-1}, x_t,\ldots, x_{j-1},s_t,\ldots, s_{j-1}) = \\
& \big\{\int_{X_j \times S_j} w^j_M(h_{t-1},x_t,\ldots, x_j,s_t,\ldots, s_j) \varrho^j_{(h_{t-1}, x_t,\ldots, x_{j-1},s_t,\ldots, s_{j-1}, \xi)}(\rmd (x_j,s_j)) \colon \\
& \quad \varrho^j_{(h_{t-1}, x_t,\ldots, x_{j-1},s_t,\ldots, s_{j-1}, \xi)} \in \Delta_j^j(h_{t-1}, x_t,\ldots, x_{j-1},s_t,\ldots, s_{j-1}), \\
& \quad w^j_M  \mbox{ is a Borel measurable selection of } W^j_M \big\}.
\end{align*}

Let $\check{S}_j = S_j$. Since
\begin{align*}
& \quad \int_{X_j \times S_j} W^j_M(h_{t-1},x_t,\ldots, x_j,s_t,\ldots, s_j) \varrho^j_{(h_{t-1}, x_t,\ldots, x_{j-1},s_t,\ldots, s_{j-1}, \xi)}(\rmd (x_j,s_j)) \\
& = \int_{S_j} \int_{X_j \times \check{S}_j} W^j_M(h_{t-1},x_t,\ldots, x_j,s_t,\ldots, s_j) \rho^j_{(h_{t-1}, x_t,\ldots, x_{j-1},s_t,\ldots, s_{j-1}, \xi)}(\rmd (x_j,\check{s}_j)) \\
& \qquad \cdot  \varphi_{j0}(h_{t-1},x_t,\ldots, x_{j-1},s_t,\ldots, s_j)\lambda_j(\rmd s_j),
\end{align*}
we have
\begin{align*}
& \quad W^{j-1}_M(h_{t-1},x_t,\ldots, x_{j-1},s_t,\ldots, s_{j-1}) = \\
& \big\{\int_{S_j}\int_{X_j \times \check{S}_j} w^j_M(h_{t-1},x_t,\ldots, x_j,s_t,\ldots, s_j) \rho^j_{(h_{t-1}, x_t,\ldots, x_{j-1},s_t,\ldots, s_{j-1}, \xi)}(\rmd (x_j,\check{s}_j)) \\
& \quad \cdot \varphi_{j0}(h_{t-1},x_t,\ldots, x_{j-1},s_t,\ldots, s_j) \lambda_j(\rmd s_j)\colon \\
& \quad \rho^j_{(h_{t-1}, x_t,\ldots, x_{j-1},s_t,\ldots, s_{j-1}, \xi)} \in \Xi_j^j(h_{t-1}, x_t,\ldots, x_{j-1},s_t,\ldots, s_{j-1}), \\
& \quad w^j_M  \mbox{ is a Borel measurable selection of } W^j_M \big\}.
\end{align*}

Let
\begin{align*}
& \quad \check{W}^{j}_M(h_{t-1},x_t,\ldots, x_{j-1},s_t,\ldots, s_j) = \\
& \big\{\int_{X_j \times \check{S}_j} w^j_M(h_{t-1},x_t,\ldots, x_j,s_t,\ldots, s_j) \cdot \rho^j_{(h_{t-1}, x_t,\ldots, x_{j-1},s_t,\ldots, s_{j-1}, \xi)}(\rmd (x_j,\check{s}_j)) \colon \\
& \quad \rho^j_{(h_{t-1}, x_t,\ldots, x_{j-1},s_t,\ldots, s_{j-1}, \xi)} \in \Xi_j^j(h_{t-1}, x_t,\ldots, x_{j-1},s_t,\ldots, s_{j-1}), \\
& \quad w^j_M  \mbox{ is a Borel measurable selection of } W^j_M \big\}.
\end{align*}
Since $W^j_M(h_{t-1},x_t,\ldots, x_j,s_t,\ldots, s_j)$ is continuous in $x_j$ and does not depend on $\check{s}_j$, it is continuous in $(x_j, \check{s}_j)$. In addition, $W^j_M$ is bounded, measurable, nonempty, convex and compact valued. By Lemma~\ref{lem-measurable payoff}, $\check{W}^{j}_M$ is bounded, measurable, nonempty and compact valued, and sectionally continuous on $X^{j-1}$.

It is easy to see that
$$W^{j-1}_M(h_{t-1},x_t,\ldots, x_{j-1},s_t,\ldots, s_{j-1}) = $$
$$\int_{S_j} \check{W}^{j}_M(h_{t-1},x_t,\ldots, x_{j-1},s_t,\ldots, s_j) \varphi_{j0}(h_{t-1},x_t,\ldots, x_{j-1},s_t,\ldots, s_j) \lambda_j(\rmd s_j).
$$
By Lemma~\ref{lem-lyapunov}, it is bounded, measurable, nonempty and compact valued, and sectionally continuous on $X^{j-1}$.

Let $W = \overline{\cup_{M\ge t} W^{t-1}_M}$. That is, $W$ is the closure of $\cup_{M\ge t} W_M$, which is measurable due to Lemma~\ref{lem-measurable selection}.

First, $W \subseteq Q_t^\tau$ because $W^{t-1}_M \subseteq Q_t^\tau$ for each $M \ge t$ and $Q_t^\tau$ is compact valued. Second, fix $h_{t-1}$ and $q \in Q_t^\tau(h_{t-1})$. Then there exists a mapping $\xi \in \Upsilon$ such that
$$q = \int_{\prod_{m \ge t} (X_m \times S_m)} u(h_{t-1}, x, s) \varrho_{(h_{t-1}, \xi)}(\rmd (x,s)).$$
For $M\ge t$, let
$$V_M(h_{t-1}, x_t, \ldots, x_M, s_t, \ldots, s_M) = $$
$$\int_{\prod_{m > M} (X_m \times S_m)} u(h_{t-1}, x_t, \ldots, x_{M}, s_t, \ldots, s_{M}, x, s)  \varrho_{(h_{t-1}, x_t, \ldots, x_{M}, s_t, \ldots, s_{M}, \xi)} (x,s)$$
and
$$q_M = \int_{\prod_{t \le m \le M} (X_m \times S_m)} V_M(h_{t-1}, x, s) \varrho^M_{(h_{t-1}, \xi)}(\rmd (x,s)).
$$
Because the dynamic game is continuous at infinity, $q_M \to q$, which implies that $q \in W(h_{t-1})$ and $Q_t^\tau \subseteq W$.

Therefore, $W = Q_t^\tau$, and hence $Q_t^\tau$  is measurable for $t > \tau$.

\

Step 3. For $t \le \tau$, we can start with $Q_{\tau+1}^\tau$. Repeating the backward induction in Section~\ref{subsubsec-backward}, we have that $Q_{t}^\tau$ is also bounded, measurable, nonempty and compact valued, and essentially sectionally upper hemicontinuous on $X^{t-1}$.
\end{proof}

The following three lemmas show that $Q_t^{\infty}(h_{t-1}) = \Phi (Q_{t+1}^\infty)(h_{t-1}) = E_t(h_{t-1})$ for $\lambda^{t-1}$-almost all $h_{t-1}\in H_{t-1}$.\footnote{The proofs for Lemmas~\ref{lem-equi} and \ref{lem-equilibrium set} follow the standard ideas with some modifications; see, for example, \cite{Harris1990}, \cite{HRR1995} and \cite{Mariotti2000}.}

\begin{lem}\label{lem-equi}
\begin{enumerate}
  \item The correspondence $Q_t^{\infty}$ is bounded, measurable, nonempty and compact valued, and essentially sectionally upper hemicontinuous  on $X^{t-1}$.
  \item For any $t\ge 1$, $Q_t^{\infty}(h_{t-1}) = \Phi (Q_{t+1}^\infty)(h_{t-1})$ for $\lambda^{t-1}$-almost all $h_{t-1}\in H_{t-1}$.
\end{enumerate}
\end{lem}

\begin{proof}
(1) It is obvious that $Q_t^{\infty}$ is bounded. By Lemma~\ref{lem-measurable selection}~(2), $Q_t^{\infty}$ is measurable. It is easy to see that if $\tau_1 \ge \tau_2$, then $Q_t^{\tau_1} \subseteq Q_t^{\tau_2}$. Since $Q_t^\tau$ is nonempty and compact valued, $Q_t^\infty$ is nonempty and compact valued.

Fix any $s^{t-1} \in S^{t-1}$ such that $Q_t^\tau(\cdot, s^{t-1})$ is upper hemicontinuous on $H_{t-1}(s^{t-1})$ for any $\tau$. By Lemma~\ref{lem-measurable selection}~(7), $Q_t^\infty(\cdot, s^{t-1})$ is upper hemicontinuous on $H_{t-1}(s^{t-1})$. Since $Q_t^\tau$ is essentially upper hemicontinuous on $X^{t-1}$ for each $\tau$, $Q_t^\infty$ is essentially upper upper hemicontinuous on $X^{t-1}$.

\

(2) For any $\tau$, $\Phi (Q_{t+1}^\infty)(h_{t-1}) \subseteq \Phi (Q_{t+1}^\tau) (h_{t-1}) \subseteq Q_{t}^\tau(h_{t-1})$, and hence $\Phi (Q_{t+1}^\infty) (h_{t-1}) \subseteq Q_t^{\infty}(h_{t-1})$.

The space $\{1, 2, \ldots \infty\}$ is a countable compact set endowed with the following metric: $d(k,m) = |\frac{1}{k} - \frac{1}{m}|$ for any $1 \le k, m \le \infty$. The sequence $\{Q_{t+1}^\tau\}_{1\le \tau \le \infty}$ can be regarded as a correspondence $Q_{t+1}$ from $H_{t} \times \{1, 2, \ldots, \infty\}$ to $\bR^n$, which is measurable, nonempty and compact valued, and essentially sectionally upper hemicontinuous on $X^{t}\times \{1, 2, \ldots, \infty\}$.
The backward induction in Section~\ref{subsubsec-backward} shows that $\Phi (Q_{t+1})$ is measurable, nonempty and compact valued, and essentially sectionally upper hemicontinuous on $X^{t}\times \{1, 2, \ldots, \infty\}$.

Since $\Phi (Q_{t+1})$ is essentially sectionally upper hemicontinuous on $X^{t}\times \{1, 2, \ldots, \infty\}$, there exists a measurable subset $\check{S}^{t-1} \subseteq S^{t-1}$ such that $\lambda^{t-1}(\check{S}^{t-1}) = 1$, and $\Phi (Q_{t+1})(\cdot, \cdot, \check{s}^{t-1})$ is upper hemicontinuous for any $\check{s}^{t-1} \in \check{S}^{t-1}$. Fix $\check{s}^{t-1} \in \check{S}^{t-1}$. For $h_{t-1} = (x^{t-1}, \check{s}^{t-1})\in H_{t-1}$ and $a \in Q_t^\infty(h_{t-1})$, by its definition, $a \in Q_t^\tau(h_{t-1}) = \Phi (Q_{t+1}^\tau) (h_{t-1})$ for $\tau \ge t$. Thus, $a \in \Phi (Q_{t+1}^\infty) (h_{t-1})$.

In summary, $Q_t^{\infty}(h_{t-1}) = \Phi (Q_{t+1}^\infty) (h_{t-1})$ for $\lambda^{t-1}$-almost all $h_{t-1}\in H_{t-1}$.
\end{proof}

Though the definition of $Q^\tau_t$ involves correlated strategies for $\tau < t$, the following lemma shows that one can work with mixed strategies in terms of equilibrium payoffs, due to the combination of backward inductions in multiple steps.

\begin{lem}\label{lem-equi strategy}
If $c_{t}$ is a measurable selection of $\Phi (Q_{t+1}^\infty)$, then $c_{t}(h_{t-1})$ is a subgame-perfect equilibrium payoff vector for $\lambda^{t-1}$-almost all $h_{t-1}\in H_{t-1}$.
\end{lem}

\begin{proof}
Without loss of generality, we only prove the case $t = 1$. Suppose that $c_{1}$ is a measurable selection of $\Phi (Q_{2}^\infty)$. Apply Proposition~\ref{prop-forward} recursively to obtain Borel measurable mappings $\{f_{ki} \}_{i\in I}$ for $k\ge 1$. That is, for any $k\ge 1$, there exists a Borel measurable selection $c_k$ of $Q_k^\infty$ such that for $\lambda_{k-1}$-almost all $h_{k-1}\in H_{k-1}$,
\begin{enumerate}
  \item $f_k(h_{k-1})$ is a Nash equilibrium in the subgame $h_{k-1}$, where the action space is $A_{ki}(h_{k-1})$ for player~$i \in I$, and the payoff function is given by
        $$\int_{S_k} c_{k+1} (h_{k-1}, \cdot, s_k) f_{k0}(\rmd s_k| h_{k-1}).$$
  \item $$c_k(h_{k-1}) = \int_{A_k(h_{k-1})} \int_{S_k} c_{k+1}(h_{k-1}, x_k, s_k) f_{k0}(\rmd s_k|h_{k-1}) f_k(\rmd x_k|h_{k-1}).$$
\end{enumerate}
We need to show that $c_{1}(h_{0})$ is a subgame-perfect equilibrium payoff vector for $\lambda_{0}$-almost all $h_{0}\in H_{0}$.

First, we show that $\{f_{ki} \}_{i\in I}$ is a subgame-perfect equilibrium. Fix a player~$j$ and a strategy $g_j = \{g_{kj} \}_{k \ge 1}$. By the one-step deviation principle, it suffices to show that for any $t' \ge 1$, $\lambda^{t'-1}$-almost all $h_{t'-1}$, and any $\delta > 0$,
\begin{align*}
& \quad \int_{\prod_{m \ge t'} (X_m \times S_m)} u_{j}(h_{t'-1}, x, s) \varrho_{(h_{t'-1}, f)}(\rmd (x,s)) \\
& > \int_{\prod_{m \ge t'} (X_m \times S_m)} u_{j}(h_{t'-1}, x, s) \varrho_{(h_{t'-1}, (f_{-j}, \tilde{g}_j))}(\rmd (x,s)) - \delta,
\end{align*}
where $\tilde{g}_j = (g_{1j}, \ldots, g_{t'j}, f_{(t'+1)j}, f_{(t'+2)j}, \ldots )$.

Since the game is continuous at infinity, there exists a positive integer $M > t'$ such that $w^m < \frac{1}{3}\delta$ for any $m \ge M$. By Lemma~\ref{lem-equi}, $c_{k}(h_{k-1}) \in \Phi (Q_{k+1}^\infty)(h_{k-1}) = Q_k^\infty(h_{k-1}) = \cap_{\tau \ge 1} Q_k^\tau(h_{k-1})$ for $\lambda_{k-1}$-almost all $h_{k-1}\in H_{k-1}$. Since $Q_k^\tau = \Phi^{\tau-k+1} (Q_{\tau+1}^\tau)$ for $k \le \tau$, $c_{k}(h_{k-1}) \in \cap_{\tau \ge 1} \Phi^{\tau-k+1} (Q_{\tau+1}^\tau)(h_{k-1}) \subseteq \Phi^{M-k+1} (Q_{M+1}^M)(h_{k-1})$ for $\lambda_{k-1}$-almost all $h_{k-1}\in H_{k-1}$.

Thus, there exists a Borel measurable selection $w$ of $Q_{M+1}^M$ and a strategy profile $\xi$ such that for $\lambda_{M-1}$-almost all $h_{M-1}\in H_{M-1}$,
\begin{enumerate}
  \item $f_{M}(h_{M-1})$ is a Nash equilibrium in the subgame $h_{M-1}$, where the action space is $A_{Mi}(h_{M-1})$ for player~$i \in I$, and the payoff function is given by
        $$\int_{S_M} w (h_{M-1}, \cdot, s_M) f_{M0}(\rmd s_M| h_{M-1}).$$
  \item $$c_M(h_{M-1}) = \int_{A_M(h_{M-1})} \int_{S_M} w(h_{M-1}, x_M, s_M) f_{M0}(\rmd s_M|h_{M-1}) f_M(\rmd x_M|h_{M-1}).$$
  \item $w (h_M) = \int_{\prod_{m \ge M+1} (X_m \times S_m)} u(h_{M}, x, s) \varrho_{(h_{M}, \xi)}(\rmd (x,s))$.
\end{enumerate}

Therefore, we have
\begin{enumerate}
  \item for $\lambda_{0}$-almost all $h_{0}\in H_{0}$,
      $$c_1(h_0) = \quad \int_{\prod_{m \ge 1} (X_m \times S_m)} u(h_{0}, x, s) \varrho_{(h_{0}, f')}(\rmd (x,s)),$$
      where $f'_k$ is $f_k$ if $k \le M$, and $\xi_k$ if $k \ge M+1$;
  \item for $1\le k \le M$, $f_k(h_{k-1})$ is a Nash equilibrium in the subgame $h_{k-1}$ for $\lambda^{k-1}$-almost all $h_{k-1} \in H_{k-1}$, where the action space is $A_{ki}(h_{k-1})$ for player~$i \in I$, and the payoff function is given by
        $$\int_{S_k}\int_{\prod_{m \ge k+1} (X_m \times S_m)} u (h_{k-1}, x_k,s_k, x, s) \varrho_{((h_{k-1},x_k,s_k), f')}(\rmd (x,s)) f_{k0}(\rmd s_k| h_{k-1}).$$
\end{enumerate}

Let $\tilde{g}'_j = (g_{1j}, \ldots, g_{t'j}, f_{(t'+1)j}, \ldots, f_{Mj}, \xi_{(M+1)j}, \ldots)$. By (2), for $\lambda^{t'-1}$-almost all $h_{t'-1} \in H_{t'-1}$,
\begin{align*}
& \quad \int_{\prod_{m \ge t'} (X_m \times S_m)} u_{j} (h_{t'-1}, x, s) \varrho_{(h_{t'-1}, (f'_{-j}, f'_j) )}(\rmd (x,s)) \\
& \ge  \int_{\prod_{m \ge t'} (X_m \times S_m)} u_{j} (h_{t'-1}, x, s) \varrho_{(h_{t'-1}, (f'_{-j},\tilde{g}'_j) )}(\rmd (x,s)).
\end{align*}

In addition, for any $h_{t'-1}$,
\begin{align*}
& \quad \int_{\prod_{m \ge t'} (X_m \times S_m)} u_{j} (h_{t'-1}, x, s) \varrho_{(h_{t'-1}, (f'_{-j}, \tilde{g}'_j))}(\rmd (x,s)) \\
& > \int_{\prod_{m \ge t'} (X_m \times S_m)} u_{j} (h_{t'-1}, x, s) \varrho_{(h_{t'-1}, (f_{-j}, \tilde{g}_j) )}(\rmd (x,s)) - \frac{1}{3}\delta,
\end{align*}
and
\begin{align*}
& \quad \int_{\prod_{m \ge t'} (X_m \times S_m)} u_{j} (h_{t'-1}, x, s) \varrho_{(h_{t'-1}, (f_{-j}, f_j))}(\rmd (x,s)) \\
& > \int_{\prod_{m \ge t'} (X_m \times S_m)} u_{j} (h_{t'-1}, x, s) \varrho_{(h_{t'-1}, (f'_{-j}, f'_j) )}(\rmd (x,s)) - \frac{1}{3}\delta.
\end{align*}
Therefore, for $\lambda^{t'-1}$-almost all $h_{t'-1} \in H_{t'-1}$,
\begin{align*}
& \quad \int_{\prod_{m \ge t'} (X_m \times S_m)} u_{j} (h_{t'-1}, x, s) \varrho_{(h_{t'-1}, (f_{-j}, f_j))}(\rmd (x,s)) \\
& > \int_{\prod_{m \ge t'} (X_m \times S_m)} u_{j} (h_{t'-1}, x, s) \varrho_{(h_{t'-1}, (f_{-j}, \tilde{g}_j) )}(\rmd (x,s)) - \delta,
\end{align*}
which implies that $\{f_{ki} \}_{i\in I}$ is a subgame-perfect equilibrium.

Then we show that for $\lambda_{0}$-almost all $h_{0}\in H_{0}$,
$$c_1(h_0) = \quad \int_{\prod_{m \ge 1} (X_m \times S_m)} u(h_{0}, x, s) \varrho_{(h_{0}, f)}(\rmd (x,s)).$$
As shown in (1), for any positive integer $M$, there exists a strategy profile $\xi$ such that for $\lambda_{0}$-almost all $h_{0}\in H_{0}$,
$$c_1(h_0) = \quad \int_{\prod_{m \ge 1} (X_m \times S_m)} u(h_{0}, x, s) \varrho_{(h_{0}, f')}(\rmd (x,s)),$$
where $f'_k$ is $f_k$ if $k \le M$, and $\xi_k$ if $k \ge M+1$. Since the game is continuous at infinity, $\int_{\prod_{m \ge 1} (X_m \times S_m)} u(h_{0}, x, s) \varrho_{(h_{0}, f)}(\rmd (x,s))$ and $\int_{\prod_{m \ge 1} (X_m \times S_m)} u(h_{0}, x, s) \varrho_{(h_{0}, f')}(\rmd (x,s))$ are arbitrarily close when $M$ is sufficiently large. Thus, for $\lambda_{0}$-almost all $h_{0}\in H_{0}$,
$$c_1(h_0) = \quad \int_{\prod_{m \ge 1} (X_m \times S_m)} u(h_{0}, x, s) \varrho_{(h_{0}, f)}(\rmd (x,s)).$$
This completes the proof.
\end{proof}

For $t \ge 1$ and $h_{t-1} \in H_{t-1}$, recall that $E_t(h_{t-1})$ is the set of payoff vectors of subgame-perfect equilibria in the subgame $h_{t-1}$. Then we shall show the following lemma.

\begin{lem}\label{lem-equilibrium set}
For any $t\ge 1$, $E_t(h_{t-1}) = Q_t^\infty(h_{t-1})$ for $\lambda^{t-1}$-almost all $h_{t-1}\in H_{t-1}$.
\end{lem}

\begin{proof}
(1) We will first prove the following claim: for any $t$ and $\tau$, if $E_{t+1}(h_t) \subseteq Q_{t+1}^\tau(h_t)$ for $\lambda^{t}$-almost all $h_{t}\in H_{t}$, then $E_{t}(h_{t-1}) \subseteq Q_{t}^\tau(h_{t-1})$  for $\lambda^{t-1}$-almost all $h_{t-1}\in H_{t-1}$. We only need to consider the case that $t \le \tau$.

%Suppose that $E_{t+1}(h_{t}) \subseteq Q_{t+1}^\tau(h_{t})$ for $\lambda^{t}$-almost all $h_{t}\in H_{t}$. Then there exists a measurable subset $\check{S}^t \subseteq S^t$ such that $\lambda^t(\check{S}^t) = 1$ and $E_{t+1}(x^t,\check{s}^t) \subseteq Q_{t+1}^\tau(x^t,\check{s}^t)$ for any $\check{s}^t \in \check{S}^t$ and $x^t \in H_t(\check{s}^t)$. Denote the projection of $\check{S}^t$ on $S^{t-1}$ as $\check{S}^{t-1}$. Then $\lambda^{t-1}(\check{S}^{t-1}) = 1$.

By the construction of $\Phi (Q_{t+1}^\tau)$ in Subsection~\ref{subsubsec-backward}, there exists a measurable subset $\acute{S}^{t-1} \subseteq S^{t-1}$ with $\lambda^{t-1}(\acute{S}^{t-1}) = 1$ such that for any $c_t$ and $h_{t-1} = (x^{t-1}, \acute{s}^{t-1}) \in H_{t-1}$ with $\acute{s}^{t-1} \in \acute{S}^{t-1}$, if
\begin{enumerate}
   \item $c_{t} = \int_{A_{t}(h_{t-1})} \int_{S_{t}} q_{t+1}(h_{t-1},x_t,s_{t}) f_{t0}(\rmd s_{t}|h_{t-1}) \alpha(\rmd  x_t)$, where $q_{t+1}(h_{t-1},\cdot)$ is measurable and $q_{t+1}(h_{t-1},x_t,s_t) \in Q^\tau_{t+1}(h_{t-1},x_t,s_t)$ for $\lambda_t$-almost all $s_t \in S_t$ and $x_t \in A_t(h_{t-1})$;
   \item $\alpha \in \otimes_{i\in I}\cM(A_{ti}(h_{t-1}))$ is a Nash equilibrium in the subgame $h_{t-1}$ with payoff $\int_{S_{t}} q_{t+1}(h_{t-1},\cdot,s_{t}) f_{t0}(\rmd s_{t}|h_{t-1})$ and action space $\prod_{i\in I} A_{ti}(h_{t-1})$,
\end{enumerate}
then $c_t \in \Phi (Q_{t+1}^\tau)(h_{t-1})$.

Fix a subgame $h_{t-1} = (x^{t-1}, \acute{s}^{t-1})$ such that $\acute{s}^{t-1} \in \acute{S}^{t-1}$. Pick a point $c_t \in E_{t}(\acute{s}^{t-1})$. There exists a strategy profile $f$ such that $f$ is a subgame-perfect equilibrium in the subgame $h_{t-1}$ and the payoff is $c_t$. Let $c_{t+1}(h_{t-1}, x_{t},s_t)$ be the payoff vector induced by $\{f_{ti} \}_{i\in I}$ in the subgame $(h_t, x_{t},s_t) \in \mbox{Gr}(A_t)\times S_t$. Then we have
\begin{enumerate}
  \item $c_{t} = \int_{A_{t}(h_{t-1})} \int_{S_{t}} c_{t+1}(h_{t-1},x_t,s_{t}) f_{t0}(\rmd s_{t}|h_{t-1}) f_t(\rmd x_t|h_{t-1})$;
  \item $f_t(\cdot|h_{t-1})$ is a Nash equilibrium in the subgame $h_{t-1}$ with action space $A_t(h_{t-1})$ and payoff $\int_{S_{t}} c_{t+1}(h_{t-1},\cdot,s_{t}) f_{t0}(\rmd s_{t}|h_{t-1})$.
\end{enumerate}
Since $f$ is a subgame-perfect equilibrium in the subgame $h_{t-1}$,  $c_{t+1}(h_{t-1}, x_t, s_t) \in E_{t+1}(h_{t-1},x_t, s_t) \subseteq Q_{t+1}^\tau (h_{t-1},x_t,s_t) $ for $\lambda_{t}$-almost all $s_{t}\in S_{t}$ and $x_t \in A_t(h_{t-1})$, which implies that $c_t \in  \Phi (Q_{t+1}^\tau) (h_{t-1}) = Q_{t}^\tau(h_{t-1})$.

Therefore, $E_t(h_{t-1}) \subseteq Q_t^\tau(h_{t-1})$ for $\lambda^{t-1}$-almost all $h_{t-1}\in H_{t-1}$.

(2) For any $t > \tau$, $E_{t} \subseteq Q_{t}^\tau$. If $t \le \tau$, we can start with $E_{\tau+1} \subseteq Q_{\tau+1}^\tau$ and repeat the argument in (1), then we can show that $E_t(h_{t-1}) \subseteq Q_t^\tau(h_{t-1})$ for $\lambda^{t-1}$-almost all $h_{t-1}\in H_{t-1}$. Thus, $E_t(h_{t-1}) \subseteq Q_t^\infty(h_{t-1})$ for $\lambda^{t-1}$-almost all $h_{t-1}\in H_{t-1}$.

(3) Suppose that $c_{t}$ is a measurable selection from $\Phi (Q_{t+1}^\infty)$. Apply Proposition~\ref{prop-forward} recursively to obtain Borel measurable mappings $\{f_{ki} \}_{i\in I}$ for $k\ge t$. By Lemma~\ref{lem-equi strategy}, $c_{t}(h_{t-1})$ is a subgame-perfect equilibrium payoff vector for $\lambda^{t-1}$-almost all $h_{t-1}\in H_{t-1}$. Consequently, $\Phi (Q_{t+1}^\infty)(h_{t-1}) \subseteq E_t(h_{t-1})$ for $\lambda^{t-1}$-almost all $h_{t-1}\in H_{t-1}$.

By Lemma~\ref{lem-equi}, $E_t(h_{t-1}) = Q_t^\infty(h_{t-1}) = \Phi (Q_{t+1}^\infty)(h_{t-1})$ for $\lambda^{t-1}$-almost all $h_{t-1}\in H_{t-1}$.
\end{proof}

Therefore, we have proved  Theorem~\ref{thm-general} and Proposition~\ref{prop-general}.

\subsection{Proof of Proposition~\ref{prop-infinite arm'}}\label{subsec-proof acc'}

We will highlight the needed changes in comparison with the proofs presented in Subsections~\ref{subsubsec-backward}-\ref{subsubsec-proof general}.

1. Backward induction. We first consider stage~$t$ with $N_t = 1$.

If $N_t = 1$, then $S_t = \{\acute{s_t}\}$. Thus, $P_{t}(h_{t-1}, x_{t}) = Q_{t+1}(h_{t-1}, x_{t}, \acute{s_t})$, which is nonempty and compact valued, and essentially sectionally upper hemicontinuous on $X^t \times \hat{S}^{t-1}$. Notice that $P_t$ may not be convex valued.

We first assume that $P_t$ is upper hemicontinuous. Suppose that $j$ is the player who is active in this period. Consider the correspondence $\Phi_t \colon H_{t-1} \to \bR^n \times \cM(X_t) \times \triangle(X_t)$ defined as follows: $(v,\alpha,\mu) \in \Phi_t(h_{t-1})$ if
\begin{enumerate}
  \item $v =  p_t(h_{t-1},A_{t(-j)}(h_{t-1}), x^*_{tj})$ such that $p_t(h_{t-1},\cdot)$ is a measurable selection of $P_t(h_{t-1}, \cdot)$;\footnote{Note that $A_{t(-j)}$ is point valued since all players other than $j$ are inactive.}
  \item $x^*_{tj} \in A_{tj}(h_{t-1})$ is a maximization point of player~$j$ given the payoff function $p_{tj}(h_{t-1}, A_{t(-j)}(h_{t-1}), \cdot)$ and the action space $A_{tj}(h_{t-1})$, $\alpha_i = \delta_{A_{ti}(h_{t-1})}$ for $i \neq j$ and $\alpha_j = \delta_{x^*_{tj}}$;
  \item $\mu = \delta_{p_t(h_{t-1},A_{t(-j)}(h_{t-1}), x^*_{tj})}$.
\end{enumerate}
This is a single agent problem. We need to show that $\Phi_t$ is nonempty and compact valued, and upper hemicontinuous.

If $P_t$ is nonempty, convex and compact valued, and upper hemicontinuous, then we can use Lemma~\ref{lem-SZ}, the main result of \cite{SZ1990}, to   prove the nonemptiness, compactness, and upper hemicontinuity of $\Phi_t$. In \cite{SZ1990}, the only step they need the convexity of $P_t$ for the proof of their main theorem is Lemma~2 therein. However, the one-player pure-strategy version of their Lemma~2, stated in the following, directly follows from the upper hemicontinuity of $P_t$ without requiring the convexity.
{\small\begin{quote}
Let $Z$ be a compact metric space, and $\{z_n\}_{n \ge 0} \subseteq Z$. Let $P\colon Z \to \bR_+$ be a bounded, upper hemicontinuous correspondence with nonempty and compact values. For each $n \ge 1$, let $q_n$ be a Borel measurable selection of $P$ such that $q_n(z_n) = d_n$. If $z_n$ converges to $z_0$ and $d_n$ converges to some $d_0$, then $d_0 \in P(z_0)$.
\end{quote}}
Repeat the argument in the proof of the main theorem of \cite{SZ1990}, one can show that $\Phi_t$ is nonempty and compact valued, and upper hemicontinuous.

Then we go back to the case that $P_t$ is nonempty and compact valued, and essentially sectionally upper hemicontinuous on $X^t \times \hat{S}^{t-1}$. Recall that we proved Proposition~\ref{prop-extension SZ} based on Lemma~\ref{lem-SZ}. If $P_t$ is essentially sectionally upper hemicontinuous on $X^t \times \hat{S}^{t-1}$, we can show the following result based on a similar argument as in Subsections~\ref{subsec-discontinuous game}: there exists a bounded,  measurable, nonempty and compact valued correspondence $\Phi_t$ from $H_{t-1}$ to $\bR^n \times \cM(X_t) \times \triangle(X_t)$ such that $\Phi_t$ is essentially sectionally upper hemicontinuous on $X^{t-1} \times \hat{S}^{t-1}$, and for $\lambda^{t-1}$-almost all $h_{t-1} \in H_{t-1}$, $(v,\alpha,\mu) \in \Phi_t(h_{t-1})$ if
\begin{enumerate}
  \item $v =  p_t(h_{t-1},A_{t(-j)}(h_{t-1}), x^*_{tj})$ such that $p_t(h_{t-1},\cdot)$ is a measurable selection of $P_t(h_{t-1}, \cdot)$;
  \item $x^*_{tj} \in A_{tj}(h_{t-1})$ is a maximization point of player~$j$ given the payoff function $p_{tj}(h_{t-1}, A_{t(-j)}(h_{t-1}), \cdot)$ and the action space $A_{tj}(h_{t-1})$, $\alpha_i = \delta_{A_{ti}(h_{t-1})}$ for $i \neq j$ and $\alpha_j = \delta_{x^*_{tj}}$;
  \item $\mu = \delta_{p_t(h_{t-1},A_{t(-j)}(h_{t-1}), x^*_{tj})}$.
\end{enumerate}

Next we consider the case that $N_t = 0$. Suppose that  the correspondence $Q_{t+1}$ from $H_{t}$ to $\bR^n$ is bounded, measurable, nonempty and compact valued, and essentially sectionally upper hemicontinuous on $X^t \times \hat{S}^t$. For any $(h_{t-1}, x_t, \hat{s}_t) \in \mbox{Gr}(\hat{A}_t)$, let
\begin{align*}
R_{t}(h_{t-1}, x_{t}, \hat{s}_t)
& = \int_{\tilde{S}_t} Q_{t+1}(h_{t-1}, x_{t}, \hat{s}_t, \tilde{s}_{t}) \tilde{f}_{t0}(\rmd \tilde{s}_{t}|h_{t-1},x_t, \hat{s}_t) \\
& = \int_{\tilde{S}_t} Q_{t+1}(h_{t-1}, x_{t}, \hat{s}_t, \tilde{s}_{t}) \varphi_{t0}(h_{t-1},x_t, \hat{s}_t, \tilde{s}_{t}) \lambda_{t}(\rmd \tilde{s}_{t}).
\end{align*}
Then following the same argument as in Section~\ref{subsubsec-backward}, one can show that $R_t$ is a nonempty, convex and compact valued, and essentially sectionally upper hemicontinuous correspondence on $X^t \times \hat{S}^{t}$.

For any $h_{t-1} \in H_{t-1}$ and $x_{t} \in A_{t}(h_{t-1})$, let
$$ P_{t}(h_{t-1}, x_{t}) = \int_{\hat{A}_{t0}(h_{t-1},x_t)} R_{t}(h_{t-1}, x_{t}, \hat{s}_{t}) \hat{f}_{t0}(\rmd \hat{s}_{t}|h_{t-1},x_t).
$$
By Lemma~\ref{lem-upper}, $P_t$ is nonempty, convex and compact valued, and essentially sectionally upper hemicontinuous on $X^t \times \hat{S}^{t-1}$. The rest of the step remains the same as in Subsection~\ref{subsubsec-backward}.

\

2. Forward induction: unchanged.

\

3. Infinite horizon: we need to slightly modify the definition of $\Xi_t^{m_1}$ for any $m_1 \ge t\ge 1$. Fix any $t\ge 1$. Define a correspondence $\Xi_t^t$ as follows: in the subgame $h_{t-1}$,
$$\Xi_t^t(h_{t-1}) = (\cM(A_{t}(h_{t-1})) \diamond \hat{f}_{t0}(h_{t-1},\cdot)) \otimes \lambda_{t}.$$
For any $m_1 > t$, suppose that the correspondence $\Xi_t^{m_1 - 1}$ has been defined. Then we can define a correspondence $\Xi_t^{m_1} \colon H_{t-1}\to \cM\left(\prod_{t \le m \le m_1}(X_m \times S_m) \right)$ as follows:
\begin{align*}
\Xi_t^{m_1}(h_{t-1}) =
& \big\{ g(h_{t-1})\diamond \left((\xi_{m_1}(h_{t-1}, \cdot) \diamond \hat{f}_{m_10}(h_{t-1},\cdot)) \otimes \lambda_{m_1}\right) \colon  \\
& g \mbox{ is a Borel  measurable selection of } \Xi_t^{m_1-1}, \\
& \xi_{m_1} \mbox{ is a Borel measurable selection of } \cM(A_{m_1}) \big\}.
\end{align*}
Then the result in Subsection~\ref{subsubsec-proof general} is true with the above $\Xi_t^{m_1}$.

Consequently, a subgame-perfect equilibrium exists.

\subsection{Proof of Proposition~\ref{prop-continuous}}\label{subsec-proof continuous}

We will describe the necessary changes in comparison with the proofs presented in Subsections~\ref{subsubsec-backward}-\ref{subsubsec-proof general} and \ref{subsec-proof acc'}.

1. Backward induction. For any $t\ge 1$, suppose that  the correspondence $Q_{t+1}$ from $H_{t}$ to $\bR^n$ is bounded, nonempty and compact valued, and upper hemicontinuous on $X^t$.

If $N_t = 1$, then $S_t = \{\acute{s_t}\}$. Thus, $P_{t}(h_{t-1}, x_{t}) = Q_{t+1}(h_{t-1}, x_{t}, \acute{s_t})$, which is nonempty and compact valued, and upper hemicontinuous. Then define the correspondence $\Phi_t$ from $H_{t-1}$ to $\bR^n \times \cM(X_t) \times \triangle(X_t)$ as $(v,\alpha,\mu) \in \Phi_t(h_{t-1})$ if
\begin{enumerate}
  \item $v =  p_t(h_{t-1},A_{t(-j)}(h_{t-1}), x^*_{tj})$ such that $p_t(h_{t-1},\cdot)$ is a measurable selection of $P_t(h_{t-1}, \cdot)$;
  \item $x^*_{tj} \in A_{tj}(h_{t-1})$ is a maximization point of player~$j$ given the payoff function $p_{tj}(h_{t-1}, A_{t(-j)}(h_{t-1}), \cdot)$ and the action space $A_{tj}(h_{t-1})$, $\alpha_i = \delta_{A_{ti}(h_{t-1})}$ for $i \neq j$ and $\alpha_j = \delta_{x^*_{tj}}$;
  \item $\mu = \delta_{p_t(h_{t-1},A_{t(-j)}(h_{t-1}), x^*_{tj})}$.
\end{enumerate}
As discussed in Subsection~\ref{subsec-proof acc'}, $\Phi_t$ is nonempty and compact valued, and upper hemicontinuous.

When $N_t = 0$, for any $h_{t-1} \in H_{t-1}$ and $x_{t} \in A_{t}(h_{t-1})$,
$$ P_{t}(h_{t-1}, x_{t}) = \int_{A_{t0}(h_{t-1},x_t)} Q_{t+1}(h_{t-1}, x_{t}, s_{t}) f_{t0}(\rmd s_{t}|h_{t-1},x_t).
$$
Let $\mbox{co} Q_{t+1}(h_{t-1}, x_{t}, s_{t})$ be the convex hull of $Q_{t+1}(h_{t-1}, x_{t}, s_{t})$. Because $Q_{t+1}$ is bounded, nonempty and compact valued, $\mbox{co} Q_{t+1}$ is bounded, nonempty, convex and compact valued. By Lemma~\ref{lem-measurable selection}~(8), $\mbox{co} Q_{t+1}$ is upper hemicontinuous.

Notice that $f_{t0}(\cdot|h_{t-1},x_t)$ is atomless and $Q_{t+1}$ is nonempty and compact valued. By Lemma~\ref{lem-lyapunov},
$$ P_{t}(h_{t-1}, x_{t}) = \int_{A_{t0}(h_{t-1},x_t)} \mbox{co} Q_{t+1}(h_{t-1}, x_{t}, s_{t}) f_{t0}(\rmd s_{t}|h_{t-1},x_t).
$$
By Lemma~\ref{lem-upper}, $P_t$ is bounded, nonempty, convex and compact valued, and upper hemicontinuous.
Then instead of relying on Proposition~\ref{prop-extension SZ}, we now use Lemma~\ref{lem-SZ} to conclude that $\Phi_t$ is bounded, nonempty and compact valued, and upper hemicontinuous.

\

2. Forward induction. The first step is much simpler.

For any $\{(h^k_{t-1}, v^k)\}_{1 \le k \le \infty} \subseteq \mbox{Gr}(\Phi_t(Q_{t+1}))$ such that $(h^k_{t-1}, v^k)$ converges to $(h^\infty_{t-1}, v^\infty)$, pick $(\alpha^k, \mu^k)$  such that $(v^k, \alpha^k, \mu^k) \in \Phi_t(h^k_{t-1})$ for $1 \le k < \infty$. Since $\Phi_t$ is upper hemicontinuous and compact valued, there exists a subsequence of $(v^k, \alpha^k, \mu^k)$, say itself, such that $(v^k, \alpha^k, \mu^k)$ converges to some $(v^\infty, \alpha^\infty, \mu^\infty) \in \Phi_t(h^\infty_{t-1})$ due to Lemma~\ref{lem-measurable correspondence}~(6). Thus, $(\alpha^\infty, \mu^\infty) \in \Psi_t(h^\infty_{t-1},v^\infty)$, which implies that $\Psi_t$ is also upper hemicontinuous and compact valued. By Lemma~\ref{lem-measurable selection}~(3), $\Psi_t$ has a Borel measurable selection $\psi_t$. Given a Borel measurable selection $q_t$ of $\Phi(Q_{t+1})$, one can let $\phi_t(h_{t-1}) = (q_t(h_{t-1}), \psi_t(h_{t-1}, q_t(h_{t-1})))$. Then $\phi_t$ is a Borel measurable selection of $\Phi_t$.

Steps~2 and 3 are unchanged.

\

3. Infinite horizon. We do not need to consider $\Xi_t^{m_1}$ for any $m_1 \ge t\ge 1$. Instead of relying on
Lemma~\ref{lem-continuous composition}, we can use Lemma~\ref{lem-measurable composition}~(3) to prove Lemma~\ref{lem-measurable history}. The proof of Lemma~\ref{lem-continuation set} is much simpler. Notice that the boundedness, nonemptiness, compactness and upper hemicontinuity of $Q^\tau_{t}$ for the case $t > \tau$ is immediate. Then one can apply the backward induction as in Lemma~\ref{lem-continuation set} to show the corresponding properties of $Q^\tau_{t}$ for the case $t \le \tau$. Following the same arguments, one can show that Lemmas~\ref{lem-equi}-\ref{lem-equilibrium set} now hold for all $h_{t-1}\in H_{t-1}$ and all $t \ge 1$.

{\small
\singlespacing

\end{document}